\documentclass{lmcs}
\pdfoutput=1

\usepackage{lastpage}
\lmcsdoi{18}{1}{40}
\lmcsheading{}{\pageref{LastPage}}{}{}%
{Apr.~09,~2021}{Mar.~15,~2022}{}

\usepackage{mathtools}
\usepackage{amssymb}
\usepackage{amsthm}
\usepackage{amsmath}
\usepackage{hyperref}
\usepackage[capitalise]{cleveref}

\newcommand{\renewtheorem}[1]{\expandafter\let\csname #1\endcsname\relax
  \expandafter\let\csname c@#1\endcsname\relax
  \expandafter\let\csname end#1\endcsname\relax
  \newtheorem{#1}}

\theoremstyle{plain}
\renewtheorem{thm}{Theorem}[section]
\renewtheorem{cor}[thm]{Corollary}
\renewtheorem{lem}[thm]{Lemma}
\renewtheorem{prop}[thm]{Proposition}
\renewtheorem{asm}[thm]{Assumption}
\renewtheorem{clm}[thm]{Claim}
\theoremstyle{definition}
\renewtheorem{rem}[thm]{Remark}
\renewtheorem{rems}[thm]{Remarks}
\renewtheorem{exa}[thm]{Example}
\renewtheorem{exas}[thm]{Examples}
\renewtheorem{defi}[thm]{Definition}
\renewtheorem{conv}[thm]{Convention}
\renewtheorem{conj}[thm]{Conjecture}
\renewtheorem{prob}[thm]{Problem}
\renewtheorem{oprob}[thm]{Open Problem}
\renewtheorem{algo}[thm]{Algorithm}
\renewtheorem{obs}[thm]{Observation}
\renewtheorem{qu}[thm]{Question}
\renewtheorem{fact}[thm]{Fact}
\renewtheorem{pty}[thm]{Property}
\newtheorem{properties}[thm]{Properties}

\theoremstyle{remark}

\newtheorem{dir}{Direction}
\newtheorem{construction}[thm]{Construction}

\usepackage{etoolbox}
\AtBeginEnvironment{proof}{\setcounter{case}{0}}
\AtBeginEnvironment{proof}{\setcounter{dir}{0}}

\newenvironment{subproof}[1][\proofname]{\begin{proof}[#1]}{\end{proof}}

\usepackage{graphicx}

\newcommand\acc{\mathit{acc}}
\newcommand\rej{\mathit{rej}}
\newcommand\sink{\mathit{sink}}

\usepackage{savesym}

\savesymbol{Bbbk}

\usepackage{float}

\usepackage{caption}
\usepackage{subcaption}

\newcommand\abra[1]{\langle #1 \rangle}
\newcommand\pbra[1]{\left( #1 \right)}
\newcommand\abs[1]{\left\vert #1 \right\vert}
\newcommand\lng[2]{\mathcal{L}_{#1}(#2)}

\newcommand\lmc{\mathcal{M}}

\newcommand\waannotated{\wa^\dagger}

\newcommand\pr{\mathbb{P}}
\newcommand\tvra{\mathit{r}}
\newcommand\tvrs{\mathit{rd}}
\newcommand\tvstandard{\mathit{tv}}

\newcommand\revar{x}
\newcommand\revec{\vec{x}}
\newcommand\chr{a}
\newcommand\intvar{\lambda}
\newcommand\basevecraw{b}
\newcommand\basevec{\vec{\basevecraw}}
\newcommand\pvecraw{r}

\newcommand\period{T}
\newcommand\z[1]{z(#1)}

\newcommand\dgf{d}

\newcommand{\SCC}{\ensuremath{\mathit{SCC}}}
\newcommand\adm{\mathcal{D}}
\usepackage{xspace}

\newcommand{\scc}{\varphi}

\newcommand\pa{\mathcal{A}}
\newcommand\nfa{\mathcal{N}}
\newcommand\wa{\mathcal{W}}
\newcommand\nfaof[2]{\nfa_{#2}(#1)}

\renewcommand\nu{f}

\newcommand\nat[1]{\ensuremath{\{1,\dots,#1\}}}

\newcommand\abfsubset{\mathrel{{\ooalign{\hss\raisebox{-0.6ex}{\small$\sim$}\hss\cr\raisebox{0.4ex}{$\subset$}}}}
}
\usepackage{mathrsfs}

\savesymbol{C}
\savesymbol{G}
\savesymbol{cc}
\usepackage[bold,full]{complexity}

\usepackage{xparse}

\usepackage{calc}
\newcommand\myxrightarrow[2][]{
    \xrightarrow[{\raisebox{1.3ex-\heightof{$\scriptstyle#1$}}[0pt]{$\scriptstyle#1$}}]{#2}}
\NewDocumentCommand{\trns}{ O{} m }{\myxrightarrow[#1]{#2}}

\newcommand*{\br}{\nobreak\discretionary{}{\hbox{}}{}}
\newcommand{\rexp}{$\mathrm{Th}(\mathbb R_{\mathsf{exp}})$}
 
\usepackage[T1]{fontenc}
 
\begin{document}

\title{The Big-O Problem}
\titlecomment{{\lsuper*}A prior version~\cite{ChistikovKMP20} appeared at CONCUR'20.
  The present paper contains full proofs and strengthens some of the results.
  }

  \author[D.~Chistikov]{Dmitry Chistikov\rsuper{a}}
  \address{Centre for Discrete Mathematics and its Applications (DIMAP) \& Department of Computer Science, University of Warwick, UK}
  \thanks{Dmitry Chistikov was supported in part by the Royal Society International Exchanges scheme (IEC\textbackslash{}R2\textbackslash{}170123).}
\author[S.~Kiefer]{Stefan Kiefer\rsuper{b}}
\address{Department of Computer Science, University of Oxford, UK}
\thanks{Stefan Kiefer was supported by a Royal Society Research Fellowship.}
\author[A.~S.~Murawski]{Andrzej S. Murawski\rsuper{b}}
\thanks{Andrzej S. Murawski was supported by a Royal Society Leverhulme Trust Senior Research Fellowship and the International Exchanges Scheme (IE161701).}
\author[D.~Purser]{David Purser\rsuper{c}}
\address{Max Planck Institute for Software Systems, Saarland Informatics Campus, Germany}
\thanks{
During the development of this work David Purser was affiliated with the University of Warwick, UK, where he was supported by the UK EPSRC Centre for Doctoral Training in Urban Science (EP/L016400/1)
and in part
by the Royal Society International Exchanges scheme (IEC\textbackslash{}R2\textbackslash{}170123).
}

\begin{abstract}
  \noindent Given two weighted automata, we consider the problem of whether one is big-O of the other, i.e.,  if the weight of every finite word in the first is not greater than some constant multiple of the weight in the second. 

We show that the problem is undecidable, even for the instantiation of weighted automata as labelled Markov chains. Moreover, even when it is known that one weighted automaton is big-O of another, 
the problem of finding or approximating the associated constant is also undecidable. 

Our positive results show that the big-O problem is polynomial-time solvable for unambiguous automata, $\coNP$-complete for unlabelled weighted automata (i.e., when the alphabet is a single character) and decidable, subject to Schanuel's conjecture, when the language is bounded (i.e., a subset of $w_1^*\dots w_m^*$ for some finite words $w_1,\dots,w_m$) or when the automaton has finite ambiguity. 

On labelled Markov chains, the problem can be restated as a ratio total variation distance, which, instead of finding the maximum difference between the probabilities of any two events, finds the maximum ratio between the probabilities of any two events. The problem is related to $\varepsilon$-differential privacy, for which the optimal constant of the big-O notation is exactly $\exp(\varepsilon)$.
\end{abstract}

\maketitle

\section{Introduction}\label{sec:intro}

Weighted automata over finite words
are a well-known and powerful model of computation,
a quantitative analogue of finite-state automata.
Special cases of weighted automata include
nondeterministic finite automata
and labelled Markov chains, two standard formalisms for modelling systems
and processes.
Algorithms for analysis of weighted automata have been studied
both in the early theory of computing and
more recently by the infinite-state systems and algorithmic verification
communities.

Given two weighted automata $\mathcal A$, $\mathcal B$
over an algebraic structure $(\mathcal S, {+}, {\times})$,
the equivalence problem asks whether the two associated functions~$f_{\mathcal A}, f_{\mathcal B} \colon \Sigma^* \to \mathcal S$
are equal: $f_{\mathcal A}(w) = f_{\mathcal B}(w)$ for all finite
words $w$ over the alphabet $\Sigma$.
Over the ring $(\mathbb Q, {+}, {\times})$,
equivalence
is decidable in polynomial time by the results of
Sch\"utzenberger~\cite{schutzenberger1961definition} and Tzeng~\cite{tzeng1992polynomial};
subsequently, fast parallel (\NC\ and \RNC) algorithms have been found
for this problem~\cite{Tzeng96,kiefer2013complexity}.
In contrast, for semirings the equivalence problem is hard:
undecidable~\cite{Krob94,AlmagorBK11} for the semiring $(\mathbb Q, {\max}, {+})$
and \PSPACE-hard~\cite{MS72} for the Boolean semiring
(for which weighted automata are usual nondeterministic finite automata
 and equivalence is equality of recognized languages).
For the ring $(\mathbb Q, {+}, {\times})$, replacing $=$ with $\le$ makes the problem harder:
the question of whether $f_{\mathcal A}(w) \le f_{\mathcal B}(w)$
for all $w \in \Sigma^*$ is undecidable---even if $f_{\mathcal A}$ is constant~\cite{paz2014introduction}.
This problem subsumes the universality problem for (Rabin) probabilistic automata,
yet another subclass of weighted automata (see, e.g.,~\cite{fijalkow2017undecidability}).

The problem of whether $f_{\mathcal A}(w) \le f_{\mathcal B}(w)$ holds for all $w \in \Sigma^*$ has often been considered, under different semirings, but also for other forms of weighted automata, in which, e.g., $\limsup$ or (discounted) limit averages are used to combine the weights along a run.
For example, most decidability and complexity results in \cite{ChatterjeeDH10,BansalCV18} are on such weighted automata.
Alternatively, to regain decidability in view of the undecidability result mentioned in the previous paragraph, one can consider semantic restrictions on the weighted automaton, e.g., the assumption that, for every~$w$, all runs on~$w$ have the same weight.
This route has been taken, e.g.,  in \cite{abs-1111-0862}.

In this paper, we introduce and study another natural problem,
in which the condition $f_{\mathcal A}(w) \le f_{\mathcal B}(w)$ 
is relaxed.
Given $\mathcal A$ and $\mathcal B$ as above,
is it true that there exists a constant $c > 0$ such that
\begin{equation*}
f_{\mathcal A}(w) \le c \cdot f_{\mathcal B}(w)
\qquad
\text{for all $w \in \Sigma^*$\,?}
\end{equation*}
Using standard mathematical notation, this condition asserts
that $f_{\mathcal A}(w) = O(f_{\mathcal B}(w))$ as $|w| \to \infty$,
and we refer to this problem as the \emph{big-O} problem accordingly.\footnote{There also exists a related but slightly different definition of big-O;
          see \cref{remark:asymptotic-big-oh:first} for details on the corresponding
          version of our big-O problem.}
The \emph{big-$\Theta$} problem
(which turns out to be computationally equivalent to the big-O problem),
in line with the $\Theta(\cdot)$ notation in analysis of algorithms,
asks whether $f_{\mathcal A} = O(f_{\mathcal B})$ and
$f_{\mathcal B} = O(f_{\mathcal A})$.

We restrict our attention to
the ring $(\mathbb Q, {+}, {\times})$ and
only consider \emph{non-negative weighted automata},
i.e., those in which all transitions have non-negative weights.
We remark that, even under this restriction,
weighted automata still form a superclass of
(Rabin) probabilistic automata, a non-trivial and rich model of computation.
Our initial motivation to study the big-O problem came from yet another formalism,
labelled Markov chains (LMCs).
One can think of the semantics of LMCs as giving a probability distribution
or subdistribution on the set of all finite words.
LMCs, often under the name Hidden Markov Models, are widely employed in
a diverse range of applications;
in computer-aided verification, they are perhaps the most fundamental model for
probabilistic systems, with model-checking tools such as Prism~\cite{KNP11} or
Storm~\cite{Storm} based on analyzing LMCs efficiently.
All the results in our paper (including hardness results) hold for LMCs too. Our main findings are as follows.
\begin{itemize}
\item
The big-O problem for non-negative WA and LMCs turns out to be \textbf{undecidable
in general}, by a reduction from nonemptiness for probabilistic automata.
\item
In the \textbf{unary case}, i.e.,
if the input alphabet $\Sigma$ is a singleton, the big-O problem
becomes decidable and is, in fact, complete for the complexity class \coNP.
Unary LMCs are a simple and pure probabilistic model of computation:
they run in discrete time and can terminate at any step; the big-O problem
refers to this termination probability in two LMCs (or two WA).
Our upper bound argument refines an analysis of growth of entries in
powers of non-negative matrices by Friedland and Schneider~\cite{schneider1986influence},
and
the lower bound is obtained by a reduction from unary NFA universality~\cite{stockmeyer1973word}.\item
In a more general \textbf{bounded case}, i.e.,
if the languages of all words $w$ associated with non-zero weight are included in
$w_1^* w_2^* \ldots w_m^*$ for some finite words $w_1, \ldots, w_m \in \Sigma^*$
(that is, are \emph{bounded in the sense of Ginsburg and Spanier};
 see~\cite[Chapter~5]{GinsburgMTCFL} and~\cite{ginsburg1964bounded}),
the big-O problem is decidable subject to Schanuel's conjecture.
This is a well-known conjecture in transcendental number theory~\cite{lang1966introduction}, which
implies that the first-order theory of the real numbers with the exponential
function is decidable~\cite{macintyre1996decidability}.
Intuitively,
our reliance on this conjecture is linked to
the expressions for the growth rate in powers of non-negative matrices.
These expressions are sums of terms of the form
$\rho^n \cdot n^k$, where $n$ is the length of a word, $k \in \mathbb N$,
and $\rho$ is an algebraic number.
Our algorithms (however implicitly) need to compare for equality
pairs of real numbers of the form
$\log \rho_1 / \log \rho_2$, where $\rho_i$ are algebraic,
and it is an open problem in number theory
whether there is an effective procedure for this task
(the four exponentials conjecture asks whether two such ratios
 can ever be equal; see, e.g., Waldschmidt~\cite[Sections~1.3 and~1.4]{Waldschmidt00}).

Bounded languages form a well-known subclass of regular languages.
In fact, a regular (or even context-free) language $L$ is bounded if and only if the number of
words of length~$n$ in $L$ is at most polynomial in~$n$. All other regular
languages have, in contrast, exponential growth rate
(a fact rediscovered multiple times; see, e.g., references in Gawrychowski et~al.~\cite{GawrychowskiKRS10}).
Bounded languages have been studied from combinatorial and algorithmic points
of view since the 1960s~\cite{ginsburg1964bounded,GawrychowskiKRS10},
and have recently been used, e.g., in the analysis of
quantitative information flow problems in computer security~\cite{Mestel19fsttcs,Mestel19csf}.
In the context of labelled Markov chains,
languages that are subsets of $a_1^* a_2^* \ldots a_m^*$ (for individual
letters $a_1, \ldots, a_m \in \Sigma$) model consecutive arrival of $m$~events
in a discrete-time system. It is curious that natural decision problems
for such simple systems can lead to intricate algorithmic questions in number theory
at the border of decidability.
\item
For \textbf{unambiguous automata}, i.e., where every word has at most one accepting
path, the big-O problem is also decidable and can be solved in polynomial time.
\item
For \textbf{finitely ambiguous automata}, i.e., where there exists $k$, such that every word has at most $k$ accepting
paths, the big-O problem is decidable subject to Schanuel's conjecture, similarly to the case of bounded languages.
\end{itemize}

\paragraph*{Further motivation and related work}

In the labelled Markov chain setting,
the big-O problem can be reformulated as a boundedness problem for the following function.
For two LMCs $\mathcal A$ and $\mathcal B$, define
the (asymmetric) \emph{ratio variation function} by
\[\tvra(\mathcal A, \mathcal B) = \sup_{E \subseteq \Sigma^*} \pbra{\frac{f_{\mathcal A}(E) }{ f_{\mathcal B}(E)}},\]
where $f_{\mathcal A}(E)$ and $f_{\mathcal B}(E)$ denote the total probability mass associated with
an arbitrary set of finite words $E \subseteq \Sigma^*$ in $\mathcal A$ and $\mathcal B$, respectively.
Here we assume $\frac{0}{0}=0$ and $\frac{x}{0}=\infty$ for $x>0$.
Observe that, because $\max(\frac{a}{b},\frac{c}{d}) \ge \frac{a+c}{b+d}$ for $a,b,c,d\ge 0$, the supremum over $E \subseteq \Sigma^*$
can be replaced with supremum over $w\in \Sigma^*$. Consequently, the big-O problem for LMCs is equivalent to deciding whether $\tvra(\mathcal A, \mathcal B) <\infty$, we give a formal proof of this in~\cref{sec:bigord}.

Finding the value of $\tvra$ amounts to  asking for the optimal (minimal) constant in the big-O notation.
Further, one can consider a symmetric variant, the \emph{ratio distance}:
$\tvrs(\mathcal A, \mathcal B) = \max\{\tvra(\mathcal A, \mathcal B),\tvra(\mathcal B, \mathcal A)\}$,
by analogy with big-$\Theta$ (\cref{prop:simplifytvrasa} also applies when $\tvra$ is replaced with $\tvrs$).
Now, $\tvrs$ is a ratio-oriented variant of the classic \emph{total variation distance} $\tvstandard$, defined by $\tvstandard(\mathcal A, \mathcal B) = \sup_{E \subseteq \Sigma^*} (f_{\mathcal A}(E) - f_{\mathcal B}(E))$,
which is a well-established way of comparing two labelled Markov chains~\cite{chen2014total,kiefer2018computing}.
We also consider
the problem of approximating $\tvra$ (as well as $\tvrs$)
to a given precision and the problem of comparing it with a given constant (threshold problem),
showing that both are undecidable.

The ratio distance~$\tvrs$ is  also equivalent to  the exponential of the  \emph{multiplicative total variation distance}
defined in \cite{chatzikokolakis2014generalized,DBLP:journals/corr/abs-0809-4794} in the context of differential privacy.
Consider a system $\lmc$, modelled by a single labelled Markov chain, where output words are observable to the environment but we want to protect the privacy of the starting configuration.  Let $R\subseteq Q\times Q$ be a symmetric relation, which relates the starting configurations intended to remain indistinguishable. Given $\varepsilon\ge 0$, we say that $\lmc$ is $\varepsilon$-differentially private (with respect to $R$) if, for all $(s,s')\in R$, we have
$f_s(E) \leq e^{\varepsilon}\cdot f_{s'}(E)$
for every (observable) set of traces $E\subseteq \Sigma^*$ \cite{dwork2006calibrating,chistikov2019asymmetric}.
\textbf{Here in the subscript of $f$ and elsewhere,
references to states $s$ and $s'$ replace references to LMCs/automata:
$\mathcal M$ stays implicit, and we specify
which state it is executed from.}\footnote{It is standard that two automata $\mathcal{A},\mathcal{B}$ can be merged into the same automaton by taking the disjoint union (relabeling states where necessary) and assigning $s,s'$ as the starting state of $\mathcal{A},\mathcal{B}$ respectively.} 
Note that there exists such an $\varepsilon$ if and only if $\tvra(s,s')<\infty$ for all $(s,s')\in R$ or, equivalently,
(the LMC $\mathcal M$ executed from)  $s$ is big-O of
(the LMC $\mathcal M$ executed from)  $s'$ for all $(s,s')\in R$.
In fact, the minimal such $\varepsilon$ satisfies $e^\varepsilon = \max_{(s,s')\in R} \tvra(s,s')$, thus $\tvra$ captures the level of differential privacy between $s$ and $s'$.

Our results show that even deciding whether the multiplicative total variation distance
is finite or $+\infty$ is, in general, impossible.
Likewise, it is undecidable whether a system modelled by a labelled Markov chain provides any degree of differential privacy, however low.

\section{Preliminaries}\label{sec:prelims}

Let $\mathbb{N},\mathbb{Z},\mathbb{Q}$ and $\mathbb{R}$ be the natural, integer, rational and real numbers respectively. When accompanied by a constraint in the subscript, we restrict to the subset of numbers satisfying the constraint. For example, $\mathbb{N}_{\ge k}$ are the natural numbers greater than or equal to $k$, and in particular, $\mathbb{N}_{>0}$ are the positive natural numbers.

\begin{defi}
A \emph{weighted automaton} $\wa$ over the $(\mathbb{Q},{+},{\times})$ semiring is defined as
a 4-tuple $\abra{Q, \Sigma{}, M, F}$, where
$Q$ is a finite set of \emph{states},
$\Sigma$ is a finite \emph{alphabet},
$M: \Sigma{} \to \mathbb{Q}^{Q\times Q}$ is a \emph{transition weighting function},
and $F \subseteq Q$ is a set of \emph{final states}.
We consider only \textit{non-negative} weighted automata, i.e. $M(a)(q,q') \ge 0$ for all $a\in \Sigma$ and $q,q' \in Q$.
\end{defi}
In complexity-theoretic arguments, we assume that each weight is given as a pair of integers (numerator and denominator) in binary.
The description size is then the number of bits required to represent $\abra{Q,\Sigma,M, F}$, including the bit size of the weights.

Each weighted automaton defines functions $\nu_s: \Sigma^* \to \mathbb{R}$, where for all $s\in Q$
\[\nu_s(w) = \sum_{t\in F} \pbra{M(a_1)\times M(a_2)\times \dots \times M(a_n)}_{s,t}
\qquad\text{ for  $w= a_1a_2\dots a_n\in\Sigma^\ast$ }
\]
and
$A\times B$ is standard matrix multiplication. We refer to $\nu_s(w)$ as \emph{the weight of $w$ from state $s$}.
Without loss of generality, a weighted automaton can have a single final state. If not,  introduce a new unique final state
$t$ s.t.\ $M(a)(q,t) = \sum_{q'\in F} M(a)(q,q')$ for all $q\in Q$,$a\in \Sigma$.

We will typically define weighted automata by listing transitions as $q\trns[a]{p}q'$ (to mean $M(a)(q,q') = p$) with the assumption that any unspecified transition has weight $0$.

 \begin{defi}\label{def:nfaof} We denote by $\lng{s}{\wa}$  the set of $w\in \Sigma^\ast$
with $\nu_s(w) > 0 $, that is, with positive weight from $s$. Equivalently, this is the language of $\nfaof{\wa}{s}$, the non-deterministic finite automaton (NFA) formed from the same set of states (and final states) as $\wa$, start state $s$, and transitions $q\trns{a} q'$ whenever $M(a)(q,q') > 0$. \end{defi}

Given $s,s'\in Q$, we say that \textbf{\emph{$s$ is big-O of $s'$}} if there exists $C > 0 $ such that
$\nu_s(w) \le C \cdot \nu_{s'}(w)$  for all $w\in \Sigma^*$.
The paper studies  the following problem.

\begin{defi}[\textsc{Big-O Problem}]\mbox{}\\
\begin{tabular}{ll}
\textsc{input} & Weighted automaton $\abra{Q, \Sigma{}, M, F}$ and $s,s'\in Q$\\
\textsc{output} & Is $s$ big-O of $s'$?
\end{tabular}
\end{defi}

\noindent In the paper we also work with labelled Markov chains. In particular, they will appear in examples and hardness (including undecidability) arguments.
As they are a special class of weighted automata, this will imply hardness (resp. undecidability) for weighted automata in general.
On the other hand, our decidability results will be phrased using weighted automata, which makes them applicable  to labelled Markov chains.
\begin{defi}
A \emph{labelled Markov chain} (LMC) is a (non-negative) weighted automaton  $\abra{Q, \Sigma{}, M, F}$ such that, for all $q\in Q\setminus F$,
we have $\sum_{q'\in Q}\sum_{a\in \Sigma{}} M(a)(q,q') = 1$  and   $M(a)(q,q') = 0$ for all  $a\in\Sigma{}, q\in F$ and $q' \in Q$.
\end{defi}
Since final states have no outgoing transitions, w.l.o.g., one can assume a unique final state.
For LMCs, the function $\nu_s$ can be extended to a measure on  the powerset of $\Sigma^\ast$ by
$\nu_s(E) = \sum_{ w\in E} \nu_s(w)$, where $E\subseteq \Sigma^*$. The measure is a subdistribution: $\sum_{w\in\Sigma^*} \nu_s(w) \le 1$.

Probabilistic automata are similar to LMCs, except that $M(a)$ is stochastic for every $a$, rather than
$\sum_{a\in \Sigma} M(a) $ being stochastic. \begin{defi}\label{defi:probautomata}
  A \textit{probabilistic automaton} $\mathcal{A}$ is a
\textit{non-negative weighted automaton} with a distinguished start state $q_s$ such that
$\sum_{q'\in Q} M(a)(q,q') = 1$ for all $q\in Q$ and $a \in \Sigma$. We use the notation $\pr_\mathcal{A}(w) = \nu_{q_s}(w)$, where $q_s$ is the start state of the probabilistic automaton $\mathcal{A}$.
\end{defi}

We will also consider \textit{unary} weighted automata, and similarly LMCs, where $|\Sigma|=1$. Then we will often omit  $\Sigma$
on the understanding that $\Sigma=\{a\}$, and describe  transitions with a single transition matrix $A=M(a)$ so that $\nu_s(a^n) = A_{s,t}^n$, where $t$ is the unique final state. Note that $A^n_{s,t}$ stands for $(A^n)(s,t)$,
and not $(A(s,t))^n$. Using the notation of regular expressions, we can write $\lng{s}{\wa}\subseteq a^\ast$.  It will turn out fruitful to consider several larger classes of languages:

\begin{defi}\label{def:bounded}  \label{def:letter-bounded} \label{def:plus-letter-bounded}\begin{sloppypar}
Let $L\subseteq \Sigma^\ast$.
$L$ is \emph{bounded}~\cite{ginsburg1964bounded} if $L\subseteq w_1^\ast w_2^\ast \cdots w_m^\ast$
for some $w_1, \dots, w_m \in \Sigma^*$.
$L$ is \emph{letter-bounded} if $L\subseteq \chr_1^*\chr_2^*\dots \chr_m^*$ for some
$\chr_1,\dots,\chr_m\in \Sigma$.
$L$ is \emph{plus-letter-bounded} if $L\subseteq\chr_1^+\chr_2^+\dots \chr_m^+$ for some
$\chr_1,\dots,\chr_m\in \Sigma$.\end{sloppypar}
\end{defi}
In each case, if the language of an NFA is suitably bounded, one can extract a corresponding bounding  regular expression~\cite{GawrychowskiKRS10}.

\section{Related problems}

\label{sec:bigord}

In this section we show three related problems to the Big-O problem: the ratio variation function, the Big-$\Theta$ problem and the eventually Big-O problem. In particular, \cref{obs:bigorvd}, relating the ratio variation function, will be useful when we show undecidability.

\subsection{Ratio variation functions and distances}

We introduced in Section~\ref{sec:intro} the ratio variation function $\tvra(s,s') = \sup_{E \subseteq \Sigma^*} (f_{s}(E) / f_{s'}(E))$, (now specified between two different starting states) and the ratio distance $\tvrs(s,s') = \max\{\tvra(s,s') ,\tvra(s',s) \}$ on labelled Markov chains.
The supremum is taken over \emph{sets} of words, reflecting its origin as a generalisation of the total variation distance. However, over labelled Markov chains this is not necessary, as the supremum function can be taken over individual words. The big-O problem is also taken over individual words, and taken in the formulation  $\tvra(s,s') = \sup_{w \in \Sigma^*} (f_{s}(w) / f_{s'}(w))$ the connection between the notion ratio variation function and the big-O problem becomes clear: \begin{obs}
  \label{obs:bigorvd} We have 
  $\tvra(s,s')< \infty$ if and only if $s$ is big-O of $s'$. Furthermore, if finite, $\tvra(s,s')$ is the minimum value of $C$ witnessing that $s$ is big-O of $s'$.\end{obs} We now formally show this reformulation of the ratio variation distance.

\begin{prop} \label{prop:simplifytvrasa}
For $\tvra$ on \textit{labelled Markov chains},  it is sufficient to consider the supremum over  $w\in \Sigma^*$ rather than $E\subseteq \Sigma^*$.\end{prop}

\begin{proof}[Proof of \cref{prop:simplifytvrasa}]
We will show we can approximate any event by a finite subset, then we can always simplify an event with more than one word, and not decrease the ratio.

Suppose $\frac{a+c}{b+d} > \frac{a}{b}$ and $\frac{a+c}{b+d} > \frac{c}{d}$.
By the first inequality, we have $ab + bc > ab + ad$, so $bc > ad$. But by the second inequality, we have $ad+cd > bc +cd$, so $ad > bc$. This is a contradiction.

By repeated application of this technique, a \emph{finite} set can always be simplified until the set has just a single element, resulting in a ratio which is no smaller.  That is, there exists $w'\in E$ such that,
\begin{equation}\label{eqn:simplifysums} \frac{f_{\mathcal A}(E) }{ f_{\mathcal B}(E)} =  \frac{\sum_{w\in E}f_{\mathcal A}(w)}{\sum_{w\in E}f_{\mathcal B}(w)} \le \frac{f_{\mathcal A}(w')}{f_{\mathcal B}(w')}.\end{equation}

Now, we show that the consideration of finite sets is sufficient: consider an event $E\subseteq\Sigma^*$, then for every $\lambda > 0 $ there is a $k\in\mathbb{N}$ such that $f_{\mathcal{A}}(E\cap \Sigma^{>k}) \le \lambda$, so then $f_{\mathcal{A}}(E\cap \Sigma^{\le k}) \le f_{\mathcal{A}}(E) \le f_{\mathcal{A}}(E\cap \Sigma^{\le k}) + \lambda$ \cite[Lemma 12]{kiefer2018computing}. For any $\varepsilon$, by choice of sufficiently small $\lambda$ there is a finite set $E'$ such that  $\frac{f_{\mathcal{A}}(E')}{f_{\mathcal{B}}(E')} - \varepsilon \le  \frac{f_{\mathcal{A}}(E)}{f_{\mathcal{B}}(E)} \le \frac{f_{\mathcal{A}}(E')}{f_{\mathcal{B}}(E')} +  \varepsilon$.

Consider $\sup_{E \subseteq \Sigma^*} \frac{f_{\mathcal{A}}(E)}{f_{\mathcal{B}}(E)} $, this is equivalent to $\lim_{k\to\infty}\sup_{E\subseteq \Sigma^* \cap \Sigma^{\le k}}\frac{f_{\mathcal{A}}(E)}{f_{\mathcal{B}}(E)} $ and by \cref{eqn:simplifysums} this is equivalent to $\lim_{k\to\infty}\sup_{w\in \Sigma^* \cap \Sigma^{\le k}}\frac{f_{\mathcal{A}}(w)}{f_{\mathcal{B}}(w)} =\sup_{w\in \Sigma^*}\frac{f_{\mathcal{A}}(w)}{f_{\mathcal{B}}(w)} $.\end{proof}

\subsection{The big-\texorpdfstring{$\Theta$}{Theta} problem}
One could consider whether $s$ is big-$\Theta$ of $s'$, defined as $s$ is big-O of $s'$ and $s'$ is big-O of $s$; or equivalently for LMCs, whether $\tvrs(s,s') <  \infty$. We note that these two notions reduce to each other, justifying our consideration of only the big-O problem. There is an obvious reduction from big-$\Theta$  to big-O making two oracle calls (a Cook reduction), but, as we show below, this can be strengthened to a single call preserving the answer (a Karp reduction). In the other direction, one can ask if $s$ big-O of $s'$ using big-$\Theta$ by asking if a linear combination of $s$ and $s'$ is big-$\Theta$ of $s'$: this ensures that the answer to $s'$ is big-O of $s$ and $s'$ is always true, essentially leaving only $s$ and $s'$ big-O of $s'$ to be checked, which depends only on whether $s$ is big-O of $s'$.

\begin{lem}\label{lemma:reductionbetweenotheta}
The big-O problem is interreducible with the big-$\Theta$ problem.
\end{lem}
\begin{proof}[Proof]
\begin{figure}[t]
\centering
\begin{subfigure}[t]{0.5\linewidth}
\centering
\includegraphics[width=0.7\linewidth]{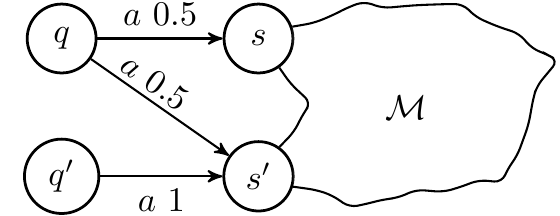}
\caption{Reduction to big-$\Theta$}\label{fig:bigtheta}

\end{subfigure}\begin{subfigure}[t]{0.49\linewidth}
\centering
\includegraphics[width=0.7\linewidth]{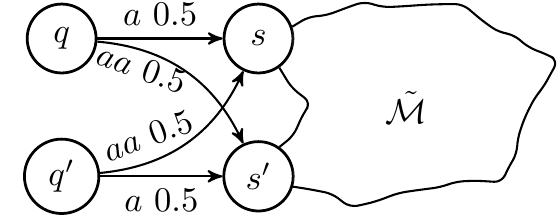}
\caption{Reduction to big-O}\label{fig:bigthetabigo}
\end{subfigure}

\caption{Reductions between big-O and big-$\Theta$}
\end{figure}

\begin{dir}[{big-O problem reduces to the big-$\Theta$ problem}]
To ask if $s$ is big-O of $s'$, add states $q,q'$ using the construction of \cref{fig:bigtheta}, then ask if $q$ is big-$\Theta$ of $q'$.
\[
\frac{\nu_{q}(aw)}{\nu_{q'}(aw)} = \frac{0.5\nu_{s}(w) + 0.5\nu_{s'}(w)}{\nu_{s'}(w)}  < C \iff \frac{\nu_{s}(w)}{\nu_{s'}(w)} < 2C -1
\]
\[
\frac{\nu_{q'}(aw)}{\nu_{q}(aw)}= \frac{\nu_{s'}(w)}{0.5\nu_{s}(w) + 0.5\nu_{s'}(w)} \le 2
\]
\end{dir}
\begin{dir}[{big-$\Theta$ problem reduces to the big-O problem}]
Given an automaton with weighting function $f$, let us construct a new automaton $f'$ by  replacing every transition $t = q\trns[a]{p}q'$ with $q\trns[a]{p}[t] \trns[a]{1} q'$, where $[t]$ is a new state. The effect of this procedure is that $f_s(a_1\dots a_n) = f'_s(a_1a_1\dots a_na_n)$. Any word of odd length has weight zero in $f'$. Given a word $w = a_1\dots a_n$, denote by $\tilde{w} = a_1 a_1 \dots a_n a_n$.

Let us choose any character $a\in \Sigma$. We now add states $q,q',\bullet_1,\bullet_2$ as follows:
\begin{enumerate}
\item $q\trns[a]{0.5}s$ and  $q\trns[a]{0.5}\bullet_1 \trns[a]{1} s'$
\item $q'\trns[a]{0.5}s'$ and  $q'\trns[a]{0.5}\bullet_2 \trns[a]{1} s$
\end{enumerate}

We claim that to ask if $s$ is big-$\Theta$ of $s'$ is equivalent to asking if $q$ is big-O of $q'$. The idea of the construction is provided in  \cref{fig:bigthetabigo}. Then we obtain the requisite bounds as follows:
\[\frac{\nu'_{q}(a\tilde{w})}{\nu'_{q'}(a\tilde{w})} = \frac{0.5\nu_{s}(w) }{0.5\nu_{s'}(w)} < C \iff \frac{\nu_{s}(w)}{\nu_{s'}(w)} < C, \]
\[\frac{\nu'_{q}(aa\tilde{w})}{\nu'_{q'}(aa\tilde{w})} = \frac{0.5\nu_{s'}(w) }{0.5\nu_{s}(w)} < C \iff \frac{\nu_{s'}(w)}{\nu_{s}(w)} < C.\]

Each of the reductions operates in logarithmic space.\qedhere\end{dir}\end{proof}

\subsection{The Eventually Big-O Problem}
\label{remark:asymptotic-big-oh:first}
Readers familiar with the big-O notation may recall definition on $f,g:\mathbb{N} \to \mathbb{N}$, of the form:
$f$ is $O(g)$ if $\exists C,k>0 $ $\forall n > k \  f(n) \le C\, g(n)$.
Despite excluding finitely many points, when $g(n)\ge 1$, it is equivalent to the definition we use,
$\exists {C > 0}\  \forall {n > 0} \  f(n) \le C\, g(n)$, by taking  $C$ large enough to deal with the finite prefix.

In the paper, though, we formally consider $s$ to not be big-O of $s'$ if there exists even a single word $w$ such that $\nu_s(w) > 0$ and $\nu_{s'}(w) = 0$.
However, for weighted automata, we could amend our definition to ``eventually big-O'' as follows: $\exists C> 0, k> 0 :  \forall w \in \Sigma^{\ge k} \ \nu_s(w) \le C \cdot \nu_{s'}(w)$.\footnote{It is instructive to juxtapose these definitions with classic definitions
in analysis,
following, e.g., de Bruijn's classic monograph~\cite[Section~1.2]{deBruijn}.
For two nonnegative functions $f, g \colon \mathbb{R} \to \mathbb{R}$:
\begin{itemize}[---]
\item
$f=O(g)$ on a set $S \subseteq \mathbb{R}$ if there is $c>0$ such that
$f(x) \le c \cdot g(x)$ for all $x \in S$,
\item
$f=O(g)$ as $x \to a$ (as $x \to \infty$, respectively)
if there is a neighbourhood of $a$ (of $\infty$, respectively)
such that $f=O(g)$ in that neighbourhood, in the sense defined above.
\end{itemize}
It is now clear that our definition of big-O from Section~\ref{sec:prelims}
asserts $\nu_s(w) = O(\nu_{s'}(w))$ for all $w\in \Sigma^*$,
i.e., on the set $\Sigma^*$.
In comparison, the ``eventually big-O'' asserts that
$\nu_s(w) = O(\nu_{s'}(w))$ as $|w| \to \infty$.
}

The Eventually Big-O Problem and Big-O Problem are also interreducible. We delay showing this equivalence to \cref{remark:asymptotic-big-oh}, as it depends on some further technical development.

Allowing additive error (in addition to the multiplicative error) in the inequality can also capture a similar relaxation of the Big-O problem.
Colcombet and Daviaud~\cite{ColcombetD13} study the problem of \emph{affine domination}, which asks whether there exists $c>0$ such that  $f_{s}(w) \le c\cdot f_{s'}(w) + c$ for all $w\in\Sigma^*$? 
The problem is shown decidable for weighted automata over the tropical  $(\mathbb{N}\cup\{\infty\},\min,+)$ semiring, contrasting the containment problem ($f_{s}(w) \le  f_{s'}(w)$ for all $w\in\Sigma^*$?) which is also undecidable in that semiring. 
We observe that both containment and the big-O problem are undecidable in the $(\mathbb{Q}_{\ge0},+,\times)$ semiring that we are considering.

\section{Big-O, Threshold and Approximation problems are undecidable}\label{sec:undeci}

We show that the big-O problem is undecidable.
We also establish undecidability for several other problems
related to computing and approximating the ratio variation distance.
Recall that this corresponds to identifying the optimal constant for  positive instances of the big-O problem
or the level of differential privacy between two states in a labelled Markov chain.
Results in this section are presented on ratio total variation distances on \textit{labelled Markov chains},
and thus apply to the big-O problem in the more general weighted automata.

In the following definition, a \emph{promise problem} (see e.g., ~\cite{promise}) restricts the inputs to the problem. An algorithm professing to decide the problem need only give a correct answer on the inputs conforming to the promise. If the input does not meet the promise, then the answer is not specified and can be arbitrary (including non-termination). In particular, it need not be decidable whether the promise holds.

\begin{defi}
\label{defi:ratiovaritionproblems}
The \emph{asymmetric threshold problem} takes
an LMC along with two states $s,s'$ and a constant $\theta$,
and asks if $\tvra(s,s') \le \theta$.
The variant under the promise of boundedness promises that $\tvra(s,s')< \infty$; i.e., the input to the problem is restricted to the instances where it is known that  $\tvra(s,s')< \infty$.

The strict variant of each problem replaces $\le$ with $<$.

The \emph{asymmetric additive approximation task} takes
an LMC, two states $s,s'$ and a constant $\gamma$,
and asks for $x$ such that $|\tvra(s,s') - x| \le \gamma$.
The \emph{asymmetric multiplicative approximation task} takes
an LMC, two states $s,s'$ and a constant $\gamma$,
and asks for $x$ such that ${1} - {\gamma}\le \frac{x}{\tvra(s,s')}\le 1+\gamma$.

In each case, the symmetric variant is obtained by replacing $\tvra$ with $\tvrs$.
\end{defi}

\begin{thm} \label{thm:approxboundundecidable}\hfill
\begin{itemize}
\item The big-O problem is undecidable, even for LMCs.
\item Each variant of the threshold problem (asymmetric/sym\-metric, non-strict/strict) is undecidable, even under the promise of boundedness.
\item All variants of the approximation tasks   (asymmetric/sym\-metric, additive/multiplicative) are recursively unsolvable, even under the promise of boundedness.
\end{itemize}
\end{thm}

\noindent The proof of \cref{thm:approxboundundecidable} will reduce from the emptiness problem for probabilistic automata (recall~\cref{defi:probautomata}). The problem \textsc{Empty} asks, given a probabilistic automaton $\mathcal{A}$,
if $\pr_\mathcal{A}(w) \le \frac{1}{2}$ for all words $w \in \Sigma^*$. (Is the language $\{w\in \Sigma^* \mid \pr_\mathcal{A}(w) > \frac{1}{2}\}$ empty?)
It is known to be undecidable~\cite{paz2014introduction,fijalkow2017undecidability}.   Recall, we use the notation $\pr_\mathcal{A}(w)$ to refer to the probability mass in the probabilistic automaton, and the notation $\nu_s(w)$ for the weight (from state $s$) in the labelled Markov chain.

\begin{proof}[Proof idea for \cref{thm:approxboundundecidable}] We reduce from \textsc{Empty}.
We assume we are given \emph{any} probabilistic automaton $\mathcal{A}$, for which we would like to decide the emptiness problem.
We first show a construction which will construct a specific labelled Markov chain from the probabilistic automaton.
The resulting labelled Markov chain will have the property that answering any of the questions of \cref{defi:ratiovaritionproblems} would allow us to answer the emptiness problem for the probabilistic automaton $\mathcal{A}$.
Therefore, an algorithm to decide any of the questions in \cref{defi:ratiovaritionproblems} would give an algorithm to decide the emptinenss problem.
Since the emptiness problem is undecidabile, meaning there is no algorithm that can decide whether any given probabilistic automaton is empty or not, then there can be no algorithm to decide any of these questions.

The construction creates two branches of a labelled Markov chain.
The first simulates the probabilistic automaton using the original weights multiplied by a scalar ($\frac{1}{4}$ in the case $|\Sigma| = 2$).
The other branch will process each letter from $\Sigma$ with equal weight (also $\frac{1}{4}$ in an infinite loop). Consequently,
 if there is a word accepted with probability greater than $\frac{1}{2}$, the ratio between the two branches will be greater than $1$.
 The construction will enable words to be processed repeatedly, so that the ratio can then be pumped unboundedly. Certain linear combinations of the branches entail a gap promise, entailing undecidability of the threshold and approximation tasks.
\end{proof}

We now implement the proof idea: first we show the properties required of the construction and (in Section~\ref{sec:undeci:implytheorem}) how these properties imply \cref{thm:approxboundundecidable}. Then we show the construction (in Section~\ref{sec:undeci:construction}) and prove that the desired properties hold (in Section~\ref{sec:undeci:proveprop}).

Given any probabilistic automaton $\mathcal{A}$, we construct a labelled Markov chain with three distinguished starting states $s,s'$ and $s''$. These states will exhibit the following the  properties:
\begin{properties}\hfill
  \label{pty:constrctuionprops}
\begin{itemize}
\item If $\pa \not\in\textsc{Empty}$ then $\tvra(s,s') = \infty$.
\item If $\pa \in\textsc{Empty}$ then $\tvra(s,s') \le 2$.
\item If $\pa \not\in\textsc{Empty}$ then $49 < \tvra(s,s'') \le 51$.
\item If $\pa\in\textsc{Empty}$ then $\tvra(s,s'') \le 2$.
\item $\tvra(s',s)\le 2$ and $\tvra(s'',s)\le 2$.
\end{itemize}
\end{properties}
\subsection{\texorpdfstring{\cref{pty:constrctuionprops}}{Properties~\ref{pty:constrctuionprops}} imply \texorpdfstring{\cref{thm:approxboundundecidable}}{Theorem~\ref{thm:approxboundundecidable}}}
\label{sec:undeci:implytheorem}

The following lemma plays a key role in proving the result. In its statement, ``undecidable to distinguish''
means that the corresponding \emph{promise problem} is undecidable. In other words, if the input  is not in one of the
two cases which should be distinguished between, the answer is not specified and can be arbitrary (including non-termination).

Due to the undecidability of $\textsc{Empty}$, a construction satisfying \cref{pty:constrctuionprops} immediately establishes the following undecidability results:
\begin{lem}\label{lemma:technical:boundedandgap} \hfill
\begin{itemize}
\item  Given an LMC along with two states $s,s'$ and constant $c$, it is undecidable to distinguish between $\tvra(s,s') \le c$ and $\tvra(s,s') = \infty$.
\item Given an LMC along with two states $s,s''$ and two numbers $c$ and $C$ such that  $c < C$, it is undecidable to distinguish between $\tvra(s,s'') \le c$ and $C \le \tvra(s,s'') < \infty$.
\end{itemize}
\end{lem}

\noindent Both statements remain true if $\tvra$ is replaced with $\tvrs$. This is because $\tvra(s',s)\le 2$ and $\tvra(s'',s)\le 2$ ensures that $\tvra(s,s')$ and $\tvra(s,s'')$ dominate in the computation of $\tvrs$, thus the proceeding four statements hold when replaced by $\tvrs$.

Before we give the construction satisfying \cref{pty:constrctuionprops}, we show how \cref{lemma:technical:boundedandgap} shows our theorem.
\begin{proof}[Proof of \cref{thm:approxboundundecidable}]
We reason by contradiction using~\cref{lemma:technical:boundedandgap}.
For the big-O problem,
it suffices to observe that, if it were decidable, one could use it to solve the first promise problem from the Lemma (recall that in a promise problem the input is guaranteed to fall into one of the two cases).
This would contradict ~\cref{lemma:technical:boundedandgap}.

Similarly, the decidability of the (asymmetric) threshold problem would allow us to distinguish between $\tvra(s,s') \le c$ and $C \le \tvra(s,s') < \infty$  (second promise problem from the Lemma) by considering the instance
$\tvra(s,s') \le \frac{c+C}{2}$ (non-strict variant) or $\tvra(s,s') < \frac{c+C}{2}$ (strict variant). A positive answer (regardless of the variant) implies $\tvra(s,s') \le c$, while a negative one yields $\tvra(s,s')\ge C$, which suffices
to distinguish the cases. Note that in both cases $\tvra(s,s')$ is bounded, so the reasoning remains valid  if it is known in advance that $\tvra(s,s')$ is bounded.

For additive  (asymmetric) approximation,  we observe that finding $x$ such that $|\tvra(s,s') - x| \le \frac{C-c}{4}$ and comparing it with $\frac{c+C}{2}$ makes it possible to distinguish between
$\tvra(s,s') \le c$ and $C \le \tvra(s,s') < \infty$. This is because $\tvra(s,s') \le c$ then implies $x < \frac{c+C}{2}$ and $C\le \tvra(s,s')$ implies $\frac{c+C}{2} < x$.

In the multiplicative case,  finding $x$ such that $1-\frac{C-c}{4C}\le \frac{x}{\tvra(s,s')}\le 1+\frac{C-c}{4C}$ and comparing $x$ with $ \frac{c+C}{2}$ yields an analogous argument.

 Since \cref{lemma:technical:boundedandgap} also applies to $\tvrs$, all of our results hold when $\tvra$ is replaced by  $\tvrs$.
\end{proof}
\subsection{The construction}
\label{sec:undeci:construction}
We now show how to build, given a probabilistic automaton, a labelled Markov chain that satisfies \cref{pty:constrctuionprops}. (We prove that the construction indeed satisfies these properties in the next section.)
\begin{construction}\label{construction:undec} For both cases, we reduce from \textsc{Empty}.
We show our construction for $\Sigma{} = \{a,b\}$, but the procedure can be generalised to arbitrary alphabets.

The construction will create two branches of a labelled Markov chain.
The first, from state $q_s$, will simulate the given probabilistic automaton using the original weights multiplied by the same scalar (in this case $\frac{1}{4}$).
The other branch, from state $s_0$, will process each letter from $\Sigma$ with equal weight (also $\frac{1}{4}$ in an infinite loop). Consequently,
 if there is a word accepted with probability greater than $\frac{1}{2}$, the ratio between the two branches will be greater than $1$.
 The construction will make it possible to process words repeatedly, so that the ratio can then be pumped unboundedly.

Formally, suppose we are given a probabilistic automaton $\pa = \abra{Q,\Sigma, M,F}$ with start state $q_s$. First observe that w.l.o.g. $q_s$ is not accepting, since in this case the empty word is accepted with probability $1$, and thus there is a word with probability greater than $\frac{1}{2}$ and a trivial positive instance of the big-O problem can be returned.

We construct the LMC
$\abra{Q',\Sigma', \delta, F'}$ taking $Q' =Q \uplus \{s, s', s'', s_0, t\}$ where $\uplus$ denotes disjoint union, $\Sigma' = \{a,b,\acc,\rej, {\vdash}\}$, $F' = \{t\}$ and $\delta$ as specified below.
First we simulate the probabilistic automaton with a scaling factor of $\frac{1}{4}$: for all $q,q'\in Q$,
\[
  q\trns[a]{\frac{1}{4}M(a)(q,q')} q' \quad  \quad  q\trns[b]{\frac{1}{4}M(b)(q,q')} q'.
\]
Originally accepting runs trigger a restart, while rejecting ones are  redirected to $t$:
\[
\text{if }q \in F: q\trns[\acc]{\frac{1}{2}} q_s \quad \text{ and }  \quad \text{if } q \not\in F: q\trns[\rej]{\frac{1}{2}} t.
\]
We then add a part of the chain which behaves equally, rather than according to the probabilistic automaton:
\[
 s_0\trns[a]{\frac{1}{4}} s_0 \quad\quad s_0\trns[b]{\frac{1}{4}} s_0 \quad\quad s_0\trns[\acc]{\frac{1}{4}} s_0  \quad\quad s_0\trns[\rej]{\frac{1}{4}} t.
\]
The construction is illustrated in \cref{fig:reduction}. To complete the reduction, we add the following transitions from $s,s',s''$:
  \[ s \trns[{\vdash}]{\frac{1}{2}} s_0 \quad\quad  s \trns[{\vdash}]{\frac{1}{2}} q_s
 \quad\quad   s' \trns[{\vdash}]{1} s_0
 \quad\quad   s'' \trns[{\vdash}]{\frac{99}{100}} s_0 \quad\quad s'' \trns[{\vdash}]{\frac{1}{100}} q_s.
\]

\begin{figure}
\centering
\includegraphics[width=0.7\linewidth]{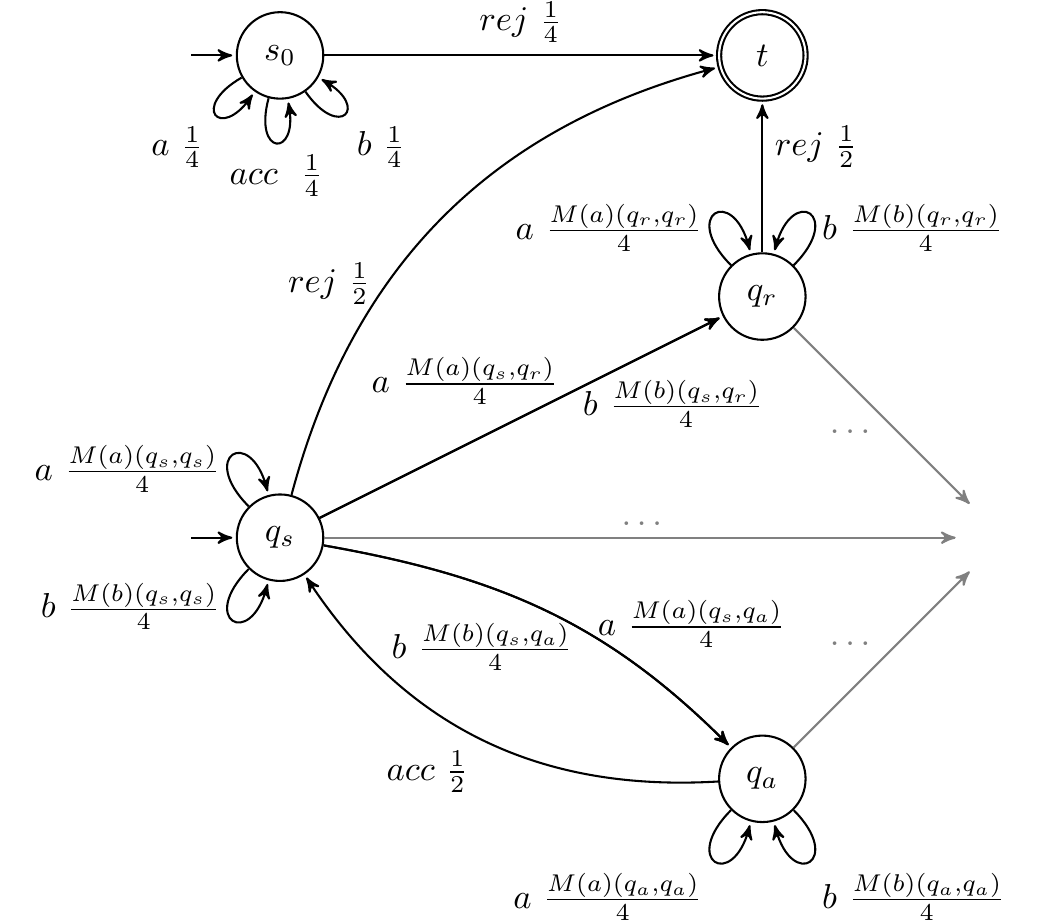}
\caption{Reduction; where $q_a$ represents accepting states of the probabilistic automaton, $q_r$ represents rejecting states and $q_s$ represents the start state (assumed to be rejecting).}
\label{fig:reduction}
\end{figure}

\end{construction}

\subsection{Proof that the construction satisfies \texorpdfstring{\cref{pty:constrctuionprops}}{Properties~\ref{pty:constrctuionprops}}}
\label{sec:undeci:proveprop}
\mbox{}\\It remains to show that \cref{construction:undec} satisfies \cref{pty:constrctuionprops}, which we show in the following two lemmas.

\begin{lem}\label{claim:boundedvsunbounded}~
If $\pa \not\in\textsc{Empty}$ then $\tvra(s,s') = \infty$.
If $\pa \in\textsc{Empty}$ then $\tvra(s,s') \le 2$.
In either case, $\tvra(s',s) \le 2$.
\end{lem}

\begin{proof}

First observe that
\begin{equation} \label{eqn:ratiossprime}
\frac{\nu_s({\vdash} w')}{\nu_{s'}({\vdash} w')} =\frac{\frac{1}{2}\nu_{s_0}(w') + \frac{1}{2}\nu_{q_s}(w')}{\nu_{s_0}(w')} = \frac{1}{2} + \frac{1}{2}\frac{\nu_{q_s}(w')}{\nu_{s_0}(w')}.
\end{equation}
If there is a word $w$ that is accepted by the automaton with probability
$>\frac{1}{2}$, then let $w'_i = (w\ \acc)^i\ \rej$ and we have
\begin{equation} \label{eqn:wordgoesunbounded}
\frac{\nu_{q_s}(w'_i)}{\nu_{s_0}(w'_i)}
= \frac{((\frac{1}{4})^{|w|}\pr(w)\frac{1}{2})^i \cdot \frac{1}{2}}{((\frac{1}{4})^{|w|}\frac{1}{4})^i\cdot \frac{1}{4}} = 2 \cdot (2 \pr(w))^i.
\end{equation}

\noindent Since $\pr(w) > \frac{1}{2}$ then $2\pr(w) > 1$ and we have:
\[
\lim_{i\to\infty}\frac{\nu_s({\vdash}w'_i)}{\nu_{s'}({\vdash}w'_i)} = \infty \quad \text{ and } \quad  \tvra(s,s') =  \tvrs(s,s') = \infty.
\]

If there is no such word then $\forall w\in\Sigma^*: \pr(w) \le \frac{1}{2}$, then probability ratio of all words is bounded. All words stating from states $s$ and $s'$ start with ${\vdash}$ and are terminated by $\rej$, so in general all words take the form $w = {\vdash}(w_1\ \acc)\dots (w_n\ \acc)(w_{n+1}\ \rej)$. Let us consider the probability of $w' = (w_1\ \acc)\dots (w_n\ \acc)(w_{n+1}\ \rej)$ words from $s_0$ and $q_s$. Then:
\begin{align}
\frac{\nu_{q_s}(w')}{\nu_{s_0}(w')} &= \frac{\pbra{\prod_{i=1}^n \frac{1}{2}(\frac{1}{4})^{|w_i|}\pr(w_i)}\pbra{(\frac{1}{4})^{|w_{n+1}|}(1-\pr(w_{n+1}))\frac{1}{2}}}{(\frac{1}{4})^{|w_1| + \dots + |w_n|}(\frac{1}{4})^n(\frac{1}{4})^{|w_{n+1}|}\frac{1}{4}}
\\& \le \frac{\pbra{(\frac{1}{4})^{|w_1| + \dots + |w_n|}(\frac{1}{2})^n(\frac{1}{2})^n}\pbra{(\frac{1}{4})^{|w_{n+1}|}\frac{1}{2}}}{(\frac{1}{4})^{|w_1| + \dots + |w_n| + n}(\frac{1}{4})^{|w_{n+1}|}\frac{1}{4}}
\tag{as $\forall i: \pr(w_i) \le \frac{1}{2}$}
\\&= 2. \label{eqn:boundedtwo}
\end{align}

Then using \cref{eqn:ratiossprime} for every word $w$ we have $\frac{1}{2} \le \frac{\nu_s(w)}{\nu_{s'}(w)} \le \frac{3}{2}$ and $\tvra(s,s')\le \frac{3}{2}$ and $\tvrs(s,s')\le 2$.
\end{proof}

\begin{lem}\label{claim:gapproblem}~
If $\pa \not\in\textsc{Empty}$ then $49 < \tvra(s,s'') \le 51$.
If $\pa\in\textsc{Empty}$ then $\tvra(s,s'') \le 2$. In either case, $\tvra(s'',s) \le 2$.
\end{lem}

\begin{proof}

We first observe that $\frac{\nu_{s''}({\vdash} w)}{\nu_s({\vdash} w)}$ is always $\le 2$, resulting in the only interesting direction being $\frac{\nu_s({\vdash} w)}{\nu_{s''}({\vdash} w)}$:
\begin{align*}
\frac{\nu_{s''}({\vdash} w)}{\nu_s({\vdash} w)}
&= \frac{\frac{99}{100}\nu_{s_0}(w)+ \frac{1}{100}\nu_{q_s}(w)}{\frac{1}{2}\nu_{s_0}(w) + \frac{1}{2}\nu_{q_s}(w)}
\\&= \frac{\frac{99}{100}\nu_{s_0}(w)}{\frac{1}{2}\nu_{s_0}(w) + \frac{1}{2}\nu_{q_s}(w)}+\frac{\frac{1}{100}\nu_{q_s}(w)}{\frac{1}{2}\nu_{s_0}(w) + \frac{1}{2}\nu_{q_s}(w)}
\\ & \le \frac{\frac{99}{100}\nu_{s_0}(w)}{\frac{1}{2}\nu_{s_0}(w)}+\frac{\frac{1}{100}\nu_{q_s}(w)}{\frac{1}{2}\nu_{q_s}(w)}
\\ & = \frac{2 \cdot 99}{100}+ \frac{2}{100} = 2.
\end{align*}

We observe that for all words ${\vdash} w$, the quantities $\tvra$ and $\tvrs$ are bounded:
\begin{align*}
\frac{\nu_s({\vdash} w)}{\nu_{s''}({\vdash} w)}
&= \frac{\frac{1}{2}\nu_{s_0}(w) + \frac{1}{2}\nu_{q_s}(w)}{\frac{99}{100}\nu_{s_0}(w)+ \frac{1}{100}\nu_{q_s}(w)}
\\&= \frac{\frac{1}{2}\nu_{s_0}(w)}{\frac{99}{100}\nu_{s_0}(w)+ \frac{1}{100}\nu_{q_s}(w)} +  \frac{\frac{1}{2}\nu_{q_s}(w)}{\frac{99}{100}\nu_{s_0}(w)+ \frac{1}{100}\nu_{q_s}(w)}
\\ & \le \frac{\frac{1}{2}\nu_{s_0}(w)}{\frac{99}{100}\nu_{s_0}(w)} +  \frac{\frac{1}{2}\nu_{q_s}(w)}{\frac{1}{100}\nu_{q_s}(w)}
\\ & \le \frac{100}{2\cdot 99} + \frac{100}{2} \le 51.
\end{align*}

If there is a word $w$ that is accepted by the automaton with probability
$>\frac{1}{2}$, then we consider the sequence of words $w'_i = (w\ \acc)^i\ \rej$:
\begin{equation*}
\frac{\nu_s({\vdash}w'_i)}{\nu_{s''}({\vdash}w'_i)}= \frac{\frac{1}{2}\nu_{s_0}(w'_i) + \frac{1}{2}\nu_{q_s}(w'_i)}{\frac{99}{100}\nu_{s_0}(w'_i)+ \frac{1}{100}\nu_{q_s}(w'_i)}
\ge \frac{\frac{1}{2}\nu_{q_s}(w'_i)}{\frac{99}{100}\nu_{s_0}(w'_i)+ \frac{1}{100}\nu_{q_s}(w'_i)}. \end{equation*}

By the previous proof (\cref{eqn:wordgoesunbounded}) we know $\frac{\nu_{q_s}(w'_i)}{\nu_{s_0}(w'_i)} \trns[i \to \infty]{} \infty$, thus $\frac{\nu_{s_0}(w'_i)}{\nu_{q_s}(w'_i)} \trns[i \to \infty]{} 0$.
Consider
\[\frac{\frac{99}{100}\nu_{s_0}(w'_i)+ \frac{1}{100}\nu_{q_s}(w'_i)}{\frac{1}{2}\nu_{q_s}(w'_i)} = \frac{2}{100} + \frac{2\cdot 99}{100}\pbra{\frac{\nu_{s_0}(w'_i)}{\nu_{q_s}(w'_i)}}
\trns[i \to \infty]{} \frac{2}{100}
\]
then

\[\frac{\frac{1}{2}\nu_{q_s}(w'_i)}{\frac{99}{100}\nu_{s_0}(w'_i)+ \frac{1}{100}\nu_{q_s}(w'_i)} \trns[i \to \infty]{} \frac{100}{2} = 50.\]

So for all $\varepsilon$ there exists an $i$ such that $\frac{\nu_s({\vdash}w'_i)}{\nu_{s''}({\vdash}w'_i)} \ge 50- \varepsilon$. In particular for example $\tvra(s,s'') \ge 49$.

If there is no such word then  $\pr(w) \le \frac{1}{2}$ for all $w\in\Sigma^*$. We show that the total variation distance will be small. All words starting from state $s$ and $s'$ start with ${\vdash}$ and are terminated by $\rej$, so in general all words take the form $w = {\vdash}(w_1\ \acc)\dots (w_n\ \acc)(w_{n+1}\ \rej)$. Let us consider the probability of such words from $s,s''$:

\begin{align*}
\frac{\nu_s(w)}{\nu_{s''}(w)}
= \frac{\frac{1}{2}\nu_{s_0}(w') + \frac{1}{2}\nu_{q_s}(w')}{\frac{99}{100}\nu_{s_0}(w')+ \frac{1}{100}\nu_{q_s}(w')}
& \le  \frac{\frac{1}{2}\nu_{s_0}(w') + \frac{1}{2}\nu_{q_s}(w')}{\frac{99}{100}\nu_{s_0}(w')}
\\& \le \frac{100}{99}\cdot \pbra{\frac{1}{2} + \frac{1}{2} \frac{\nu_{q_s}(w')}{\nu_{s_0}(w')}}
\\&\le \frac{100}{99}\cdot \frac{3}{2}   \tag{by \cref{eqn:boundedtwo}}
\\ &\le 2.
\end{align*}

This creates a significant gap between the case where there is a word with probability greater than one half and not; in particular if there exists $ w$ with $\pr(w) > \frac{1}{2}$ then $49 < \tvra(s,s'') \le 51$ and $ 49 < \tvrs(s,s'') \le 51$, but if no such word exists then $\tvra(s,s'') \le 2$ and $\tvrs(s,s'') \le 2$.
\end{proof}

\begin{rem}
The classic \emph{non-strict} threshold problem for the  total variation distance (i.e. whether $\tvstandard(s,s') \le \theta$) is  known to be undecidable~\cite{kiefer2018computing}, like our distances.
However,   it is not known if its strict variant  (i.e. whether $\tvstandard(s,s') < \theta$) is also undecidable. In contrast, in our case, both variants are undecidable. Further note that (additive) approximation of $\tvstandard$ is possible~\cite{kiefer2018computing,chen2014total}, but this is not the case for our distances $\tvra$ and $\tvrs$.
\end{rem}

\section{The relation to the \textsc{Value-1} Problem}\label{append:rel-val1}

In the previous section we have shown the undecidability of the big-O problem using the undecidability of the emptiness problem for probabilistic automata. Another proof of undecidability can be obtained using the \textsc{Value-1} problem, shown to be undecidable in~\cite{DBLP:conf/icalp/GimbertO10}: indeed, we will show in this section that the big-O problem and the \textsc{Value-1} problem are interreducible.

The \textsc{Value-1} problem asks whether there exists a sequence of words with probability tending to 1:

\medskip

\begin{defi}[\textsc{Value-1} problem]~\\
\begin{tabular}{ll}
\textsc{input} & Probabilistic automaton $\mathcal{A}$\\
\textsc{output} & For all $\delta > 0$, does there exist a word $w$ such that $\pr_\mathcal{A}(w) > 1-\delta$\,?
\end{tabular}
\end{defi}

\noindent We will exhibit a close, but not complete, connection between the \textsc{Value-1} problem and big-O problem by reducing in both directions between the two, see \cref{sec:val1:interreduction}.
In this section we are concerned only with matters of decidability, and thus do not consider whether the reductions are efficient.

Our main decidability results
(in Sections~\ref{sec:bigounaryconp}, \ref{sec:bounded}, and \ref{sec:ambiguity})
are independent of the findings of this section, so it can be skipped during the first reading.
However, \cref{sec:val1:decidability} below justifies these further developments,
showing why the connection exhibited here
does not entail decidability results for the big-O problem by a transfer
from decidable cases of \textsc{Value-1}.

\pagebreak
\subsection{Interreduction between \textsc{Value-1} and the big-O problem}\label{sec:val1:interreduction}

\begin{prop}\label{prop:v1tobigo}
The \textsc{Value-1} problem reduces to the big-O problem.
\end{prop}
\begin{proof} 

Assume we are given a probabilistic automaton $\pa = \abra{Q,\Sigma,M,F_{\pa}}$ and a dedicated starting state $q_0 \in Q$, which accepts words~$w$ with probability $\pr_\pa(w)$. First construct $\pa'$ in which words~$w$ are accepted with probability $\pr_{\pa'}(w) = 1 - \pr_\pa(w)$, by inverting accepting states, that is, $\pa' = \abra{Q,\Sigma,M,F_{\pa'}}$ with $F_{\pa'}= Q\setminus F_{\pa} $.

Let us first consider the idea of the proof. Consider the probabilistic automaton $\mathcal{B}$ such that $\mathcal{B}(w) = 1$ for all $w\in\Sigma^*$. Then $\mathcal{B}$ is not big-O of $\pa'$ if and only if there is a sequence of words for which $\pa'(w_i) \xrightarrow{i\to\infty} 0$. Since we wish to show the undecidability even for labelled Markov chains, we adjust the weight at each step by the size of the alphabet in order to encode a probabilistic automaton inside the labelled Markov chain. The remainder of the proof formally builds this construction.

Construct a Markov chain $ \lmc_{\pa'} = \abra{Q',\Sigma', M', F'}$, where $Q' =Q \cup \{s, s', s_0, \rej, \acc\}$, $\Sigma' = \Sigma \cup \{\$\}$ (assuming $\$\not\in\Sigma$) and $F' =\{\acc\}$. The probabilistic automaton $\pa'$ will be simulated by $\lmc_{\pa'}$. The map  $M'(a)(q,q')= p$ is described using the notation $q\trns[a]{p}q'$ as follows:

For all $q\in Q:$
\[
  \forall q' \in Q: q\trns[a]{(1/|\Sigma'|)M(a)(q,q')} q' \text{ for all } a\in \Sigma
\]
\[
\text{if }q \in F_{{\pa'}}: q\trns[\$]{(1/|\Sigma'|)} \acc \quad \text{ and }  \quad \text{if } q \not\in F_{{\pa'}}: q\trns[\$]{(1/|\Sigma'|)} \rej
\]
\[
  s' \trns[\$]{1} q_0 \quad\quad s \trns[\$]{1} s_0  \quad\quad s_0\trns[\$]{(1/|\Sigma'|)} \acc\quad\quad \text{ and }\quad\quad s_0\trns[a]{(1/|\Sigma'|)} s_0 \text{ for all } a\in \Sigma.
\]

Note the only words with positive probability are words of the form $\$ \Sigma^* \$ \subseteq \Sigma'^* $. Then given a word $w\in \Sigma^*$,
$\nu_s(\$ w\$) =\left(\frac{1}{|\Sigma| + 1}\right)^{|w\$|}$ and $\nu_{s'}(\$ w\$) = \left(\frac{1}{|\Sigma| + 1}\right)^{|w\$|}(1-\pr_\pa(w))$.

Then if there is a sequence of words $(w_i)_i$ for which $\pr_\pa(w_i)$ tends to 1 then $\frac{ \nu_s(\$w_i\$)}{ \nu_{s'}(\$w_i\$)} $ is unbounded. However, if there exists some $\gamma > 0$ so that for all $w \in \Sigma^*$ we have $  \pr_\pa(w) \le (1-\gamma)$ then $(1- \pr_\pa(w)) \ge \gamma$, and so $\frac{ \nu_s(\$w\$)}{ \nu_{s'}(\$w\$)}  \le \frac{1}{\gamma}$.
\end{proof}

\begin{prop}\label{prop:bigotov1}
The big-O problem reduces to the \textsc{Value-1} problem.
\end{prop}
\begin{proof}
We take an instance of the big-O problem, assumed to be a labelled Markov chain with states $s,s'$ and construct probabilistic automaton $\mathcal{A}'$ as an instance of the \textsc{Value-1} problem. This new probabilistic automaton will have the property that $s$ is not big-O of $s'$ if and only if $\mathcal{A}'$ is a positive instance of the \textsc{Value-1} problem. We start with the construction, followed by both directions of the implication.

\begin{figure}[t!]
\centering
\includegraphics[width=0.6\linewidth]{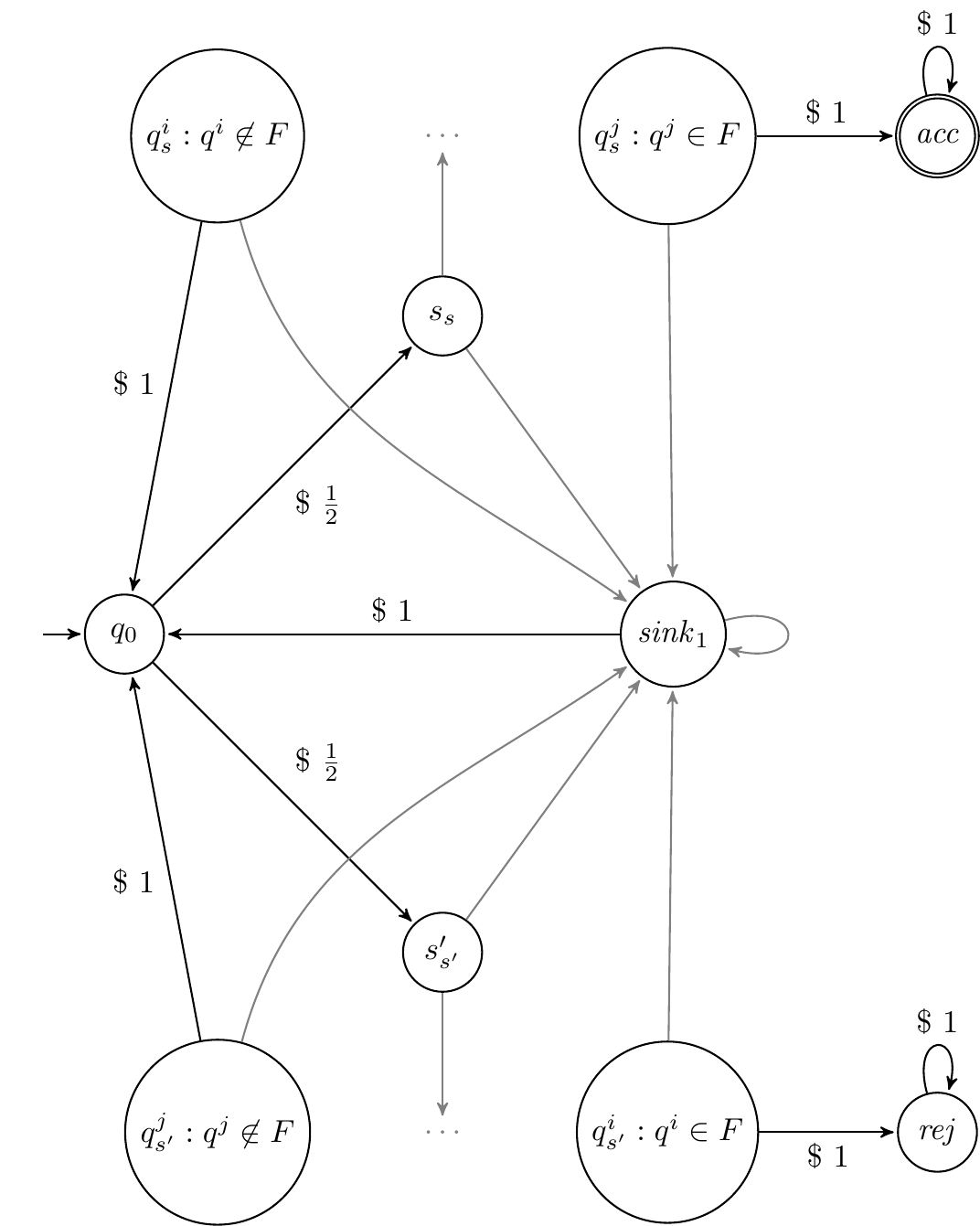}
\caption[Reduction to \textsc{Value-1}]{Reduction to \textsc{Value-1}. Only the effect of transitions on the  $\$$ symbol are shown in black, with the possibility to transition to the recoverable sink state $\sink_1$ depicted in grey (on symbols in $\Sigma$). All remaining transitions and $\sink_2$ are omitted.}
\label{bgofig:toval1}
\end{figure}

\begin{construction}
Given a labelled Markov chain $\lmc = \abra{Q,\Sigma, M, F}$ and $s,s' \in Q$, construct a probabilistic automaton $\pa = \abra{Q', \Sigma', M', F'}$. Each state of $Q$ will be duplicated, once for $s$ and once for $s'$; $Q_s = \{q_s \ | \ q \in Q\}$, $Q_{s'} = \{q_{s'} \ | \ q \in Q\}$. Let $Q' = Q_{s} \cup Q_{s'} \cup \{q_0,\acc,\rej, \sink_1,\sink_2\}$, $\Sigma' = \Sigma \cup \{\$\} $ and $F' = \{\acc\}$. The reduction can be seen in \cref{bgofig:toval1}.

Each transition of $\lmc$ will be simulated in each of the copies according to the probability in $\lmc$.
For every $q,q'\in Q$ and $a\in\Sigma$, let $M'(a)(q_s,q_s') = M(a)(q,q')$ and $M'(a)(q_{s'},q_{s'}') = M(a)(q,q')$. A probabilistic automaton should be stochastic for every $a\in \Sigma$, so there is unused probability for each character, which will divert to a (recoverable) sink. For every $q\in Q$ and $a\in\Sigma$, let
\[ M'(a)(q_{s},\sink_1) = 1- \sum_{q'\in Q}M(a)(q,q')\] and \[ M'(a)(q_{s'},\sink_1) = 1- \sum_{q'\in Q}M(a)(q,q').
\]

There is an additional control character $\$$: from $q_0$ the machine will pick either of the two machines with equal probability: $M(\$)(q_0,s_s) = M(\$)(q_0,{s'}_{s'}) = \frac{1}{2}$. If in the accepting, rejecting, or second (unrecoverable) sink state the system will stay there forever: $M'(a)(\acc,\acc) = M'(a)(\rej,\rej)= M'(a)(\sink_2,\sink_2) = 1$ for all $a\in\Sigma'$ (including $\$$).

The behaviour on $\$$ will differ in the two copies of $\lmc$. If in an accepting state reached from $s$, the system will accept. If in an accepting state reached from $s'$ the system will reject. Otherwise the system will restart back to $q_0$.

Formally,
\[
M'(\$)(q_s,\acc) = 1 \text{ when } q_s\in F  \text{ and }  M'(\$)(q_s,q_0) = 1\text{ when } q_s\not\in F
\] and
\[
M'(\$)(q_{s'},\rej)= 1 \text{ when } q_{s'}\in F  \text{ and }  M'(\$)(q_{s'},q_0)= 1 \text{ when } q_{s'}\not\in F.
\]

When in the recoverable sink state, the system restarts on $\$$, $M'(\$)(\sink_1,q_0) = 1$, and stays there for all $a\in\Sigma$, $M'(a)(\sink_1,\sink_1) = 1$.

It is intended that all $a\in\Sigma$ are unreadable from $q_0$: the system moves to the unrecoverable sink state $\sink_2$, $M'(a)(q_0,\sink_2) = 1$.

The idea is that if $\nu_s(w) \gg \nu_{s'}(w)$ then, by repeatedly reading  the word $w$, all of the probability mass will eventually move to $\acc$; otherwise a sufficiently large amount of mass will be lost to $\rej$.

Denote by $\pr_\pa(w)$ the probability of a word $w$ in the probabilistic automaton $\pa$, from state $q_0$. Here we use the notation $\nu$ to refer to the probability in the original labelled Markov chain $\lmc$. Further, the notation $\pr[q \trns{w} q']$ stands for $\left(M'(w_1)\times \dots \times M'(w_{|w|})\right)_{q,q'}$, i.e., the probability of transitioning from state $q$ to $q'$ while reading $w$ in $\pa$.
\end{construction}

\begin{dir}[Not big-O implies \textsc{Value-1}]

The proof shows that for every $\delta>0$ there exists $C > 0$ such that, for all $w\in\Sigma^*$, the inequality $\nu_s(w) > C \nu_{s'}(w) $ implies $\pr_\pa((\$w\$)^i) > 1-\delta$ for some appropriately chosen $i \in \mathbb{N}$.

Hence given $\delta$,  choose $C$ such that $(1-\frac{\delta}{2})\frac{C}{C+1} > 1-\delta$. Then by the not big-O property, choose a word~$w$ such that $\nu_s(w) = C' \nu_{s'}(w)$, with $C' > C$. Then $(1-\frac{\delta}{2})\frac{C'}{C'+1} >(1-\frac{\delta}{2})\frac{C}{C+1} > 1-\delta$.

The input word $\$w\$$ induces a (unary) Markov chain, represented by the matrix $A$, representing states $q_0, \acc$ and $\rej$ in the three positions, respectively:\[A = \left[
\begin{array}{ccccc}
0.5(1-\nu_s(w)) + 0.5(1-\nu_{s'}(w))  & 0.5\nu_s(w) & 0.5\nu_{s'}(w) \\
0 & 1 & 0 \\
0 & 0 & 1
\end{array}
\right]
\]
Then, for the word $(\$w\$)^i$, for $i \in \mathbb{N}$, starting from state $q_0$, observe:
\[
A^i
\trns[i\to\infty]{}
\left[\begin{array}{ccccc}
0 & C' x & x \\
0 & 1 & 0 \\
0 & 0 & 1 \\
\end{array}\right] \text{ with } C'x + x = 1.
\]

For each $i$, we have $A^i(q_0,\acc) +A^i(q_0,\rej) +A^i(q_0,q_0)= 1 $. Choose $i$ such that $A^i(q_0,q_0) \le \frac{\delta}{2}$. Then $A^i(q_0,\acc) +A^i(q_0,\rej) \ge 1 -\frac{\delta}{2} $, using the fact that  $A^i(q_0,\acc) = C' A^i(q_0,\rej)$,  obtaining $A^i(q_0,\acc) +\frac{A^i(q_0,\acc)}{C'} \ge 1 - \frac{\delta}{2}$. Hence $\pr_\pa((\$w\$)^i) = A^i(q_0,\acc) \ge (1 - \frac{\delta}{2})\frac{C'}{C'+1} > 1-\delta $, as required.
\end{dir}

\begin{dir}[big-O implies Not \textsc{Value-1}]

We know that there exists $C > 0$ such that, for all $w \in \Sigma^*$, $\nu_s(w) \le C \nu_{s'}(w)$, and should show that there exists $\delta > 0 $ such that every $w \in (\Sigma \cup\{\$\})^*$ has $\pr_\pa(w) \le 1-\delta$.

To move probability from $q_0$ to $\acc$ it is necessary to use words of the form $(\$ \Sigma^* \$)^*$ where $\Sigma$ is the alphabet of $\lmc$. Hence any word $w$ can be decomposed into $\$ w_m \$\$ w_{m-1} \$...\$ w_1 \$$.

Let us define $x_i = \pr[q_0 \trns{\$ w_i \$...\$ w_1 \$} \acc]$ and $y_i = \pr[q_0 \trns{\$ w_i \$...\$ w_1 \$} \rej]$, and observe $x_m =  \pr_\pa(w)$. We compute $x_i,y_i$ inductively and show $x_i\le C y_i$ for all $1\le i\le m$.

Consider reading $w_1$, the probability is such that
\begin{align*}x_1 &= \pr[q_0 \trns{\$ w_1 \$} \acc] =\frac{1}{2} \nu_s(w_1) \\ y_1 &= \pr[q_0 \trns{\$ w_1 \$} \rej] =\frac{1}{2} \nu_{s'}(w_1) \\ \pr[q_0 \trns{\$ w_1 \$} q_0]  &= 1-x_1-y_1.\end{align*}

Since there exists $C > 0$ such that, for all $w_i$, $ \nu_s(w_i)\le C \nu_{s'}(w_i) $, we have $x_1 \le C y_1$. Repeating this argument inductively, we show that $x_i \le C y_i$ for all $i$. Indeed:
\begin{align*}
x_i &= \pr[q_0 \trns{\$ w_i \$...\$ w_1 \$} \acc]
\\ &=\pr[q_0 \trns{\$ w_i \$} q_0] \cdot \pr[q_0 \trns{\$ w_{i-1} \$...\$ w_1 \$} \acc] + \pr[q_0 \trns{\$ w_i \$} \acc]
\\ &=(1- \frac{1}{2}\nu_s(w_i) -\frac{1}{2}\nu_{s'}(w_i) ) \cdot x_{i-1} + \frac{1}{2}\nu_s(w_i) \\
y_i &= \pr[q_0 \trns{\$ w_i \$...\$ w_1 \$} \rej]
\\ &=\pr[q_0 \trns{\$ w_i \$} q_0] \cdot \pr[q_0 \trns{\$ w_{i-1} \$...\$ w_1 \$} \rej] + \pr[q_0 \trns{\$ w_i \$} \rej] \\
&=(1- \frac{1}{2}\nu_s(w_i) -\frac{1}{2}\nu_{s'}(w_i) ) \cdot y_{i-1} + \frac{1}{2}\nu_{s'}(w_i).
\end{align*}
Hence, assuming $\nu_s(w_i)\le C \nu_{s'}(w_i)$ and $x_{i-1} \le C y_{i-1}$ we have: \begin{align*}x_i &=(1- \frac{1}{2}\nu_s(w_i) -\frac{1}{2}\nu_{s'}(w_i) ) \cdot x_{i-1} + \frac{1}{2}\nu_s(w_i)\\&\le (1- \frac{1}{2}\nu_s(w_i) -\frac{1}{2}\nu_{s'}(w_i) ) \cdot Cy_{i-1} + C \cdot \frac{1}{2}\nu_{s'}(w_i) \\&= C \cdot [(1- \frac{1}{2}\nu_s(w_i) -\frac{1}{2}\nu_{s'}(w_i) )y_{i-1} + \frac{1}{2}\nu_{s'}(w_i)]\\&\le Cy_{i}.\end{align*}

We have $x_m \le C y_m \le C(1-x_m) = C - C x_m$, so
$\pr_\pa(w) = x_m \le \frac{C}{C+1} < 1$, so the probability of reaching $\acc$ is bounded away from $1$ for every word.\qedhere\end{dir}\end{proof}

\subsection{What do decidable cases of the \textsc{Value-1} problem tell us about the big-O problem?}
\label{sec:val1:decidability}

The \textsc{Value-1} problem is undecidable in general, however it is decidable in the unary case  in $\coNP$ \cite{chadha2014decidable} and for \textit{leaktight automata} \cite{oualhadj2015deciding}. Note, however, that the construction of \cref{prop:bigotov1} combined with these decidability results \textit{does not} entail decidability results for the big-O problem. Firstly note that the construction adds an additional character, and so a unary instance of the big-O problem always has at least two characters when translated to the \textsc{Value-1} problem. Further, as we show in the remainder of this section, the construction does not result in a leaktight automaton.
The following argument, does not, of course, preclude the existence of an alternative construction which could, hypothetically, maintain these properties.

Let us recall the definition of leaktight automata~\cite{oualhadj2015deciding}.

\begin{defi}
A finite word $u$ is \emph{idempotent} if reading once or twice the word $u$ does not change qualitatively the transition probabilities. That is $\pr_\pa[q \trns{u} q'] > 0 \iff \pr_\pa[q \trns{uu} q'] > 0 $.

Let $u_n$ be a sequence of idempotent words.  Assume that the sequence of matrices $\pr_\pa(u_n)$ converges to a limit $M$, that this limit is idempotent and denote $\mathcal M$ the associated Markov chain. The sequence $u_n$ is a \textit{leak} if there exist $r, q\in Q$ such that the following three conditions hold:
\begin{enumerate}
\item $r$ and $q$ are recurrent in $M$, \item $\lim_{n \to \infty} \pr_\pa[r \trns{u_n} q] = 0$, \item for all $n$, $\pr_\pa[r \trns{u_n} q]>0$.
\end{enumerate}
An automaton is \emph{leaktight} if there is no leak.
\end{defi}

If there were no leak in the probabilistic automata obtained from the reduction
of~\cref{prop:bigotov1}, then the decidability  of the big-O problem for labelled Markov
chains would follow. However, this is not the case, i.e., \cref{prop:bigotov1} does not solve any cases of the big-O problem by reducing to known decidable fragments of \textsc{Value-1}.

\begin{prop}
\label{ex:leak}
The resulting automaton from the reduction of the big-O problem to the \textsc{Value-1} problem has a leak.
\end{prop}

\begin{proof}

Consider some infinite sequence of words $w_i$ growing in length, such that $\nu_s(w_i) > 0$ for every $i$.
Let $u_i = \$ w_i\$$.

Observe that this word is idempotent. For each starting state, consider the possible states with non-zero probability and from each of these the set of reachable states. Observe that in all cases the set reachable after one application is equal to the set reachable after two:
\begin{itemize}
\item $\acc \trns{\$ w_i\$} \acc \trns{\$ w_i\$} \acc $
\item $\rej \trns{\$ w_i\$} \rej \trns{\$ w_i\$} \rej $
\item $q_0 \trns{\$ w_i\$} q_0, \acc, \rej \trns{\$ w_i\$} q_0,\acc,\rej $

\item For $q$ accepting in $Q_s$:
$q \trns{\$ w_i\$} \acc \trns{\$ w_i\$} \acc $

\item For $q$ rejecting in $Q_s$:
$q \trns{\$ w_i\$}  \sink_2 \trns{\$ w_i\$} \sink_2 $

\item For $q$ accepting in $Q_{s'}$:
$q \trns{\$ w_i\$}  \rej \trns{\$ w_i\$} \rej $

\item For $q$ rejecting in $Q_{s'}$:
$q \trns{\$ w_i\$}  \sink_2 \trns{\$ w_i\$} \sink_2 $

\item
$\sink_1 \trns{\$ w_i\$}  \sink_2 \trns{\$ w_i\$} \sink_2 $
\item
$\sink_2 \trns{\$ w_i\$}  \sink_2 \trns{\$ w_i\$} \sink_2 $.
  \end{itemize}

\noindent Assume that the original labelled Markov chain has a sink, that is, the decision to terminate the word must be made by probability. Then for all $\lambda > 0$ there exists $n$ such that $\nu_s(\Sigma^{>n}) < \lambda$ and $\nu_{s'}(\Sigma^{>n}) < \lambda$ \cite[Lemma 12]{kiefer2018computing}.

Suppose the sequence of matrices $\pr_\pa(u_n)$ converges to a limit $M$ (this is with no loss of generality, because each matrix in the sequence belongs to a compact set, so a convergent subsequence must exist) and let $r = q_0$ and $q = \acc$.

By the observation above, for longer and longer words the probability of reaching $\acc$ from $q_0$ is diminishing. Thus $\lim \pr_\pa[r \trns{u_n} q] = 0$, and in $M$ we have $r$ and $q$ in different SCCs. The state $\acc$ is recurrent as it is deterministically looping on every character. Since the probability of reaching a final state is diminishing for longer and longer words, whenever $\$$ is read the state returns to $r$, hence all words return to $r$ with probability 1 in the limit. By the choice of words in the sequence, for every word $\nu_s(w_n) > 0$, we have $\pr_\pa[r \trns{u_n} q]>0$ for all $n$.

Hence a leak has been defined, even in the case where $\lmc$ is unary.
\end{proof}

 \section{The Language Containment condition} \label{sec:langequiv}

Towards decidability results, we identify a simple necessary  (but insufficient) condition for $s$ being big-O of $s'$.

\begin{defi}[LC condition]\label{def:lc}
A weighted automaton $\wa=\abra{Q,\Sigma,M,F}$ and $s,s'\in Q$ satisfy the
\emph{language containment condition} (LC) if  for all words $w$ with $\nu_s(w) > 0$  we also have $\nu_{s'}(w) > 0$.
Equivalently, $\lng{s}{\wa}\subseteq \lng{s'}{\wa}$.
\end{defi}
The condition can be verified by constructing
NFA $\nfaof{\wa}{s},\nfaof{\wa}{s'}$ that accept $\lng{s}{\wa}$ and $\lng{s'}{\wa}$ respectively and verifying $\lng{}{\nfaof{\wa}{s}}\subseteq \lng{}{\nfaof{\wa}{s'}}$.

\begin{figure}
\centering
\includegraphics[width=0.6\linewidth]{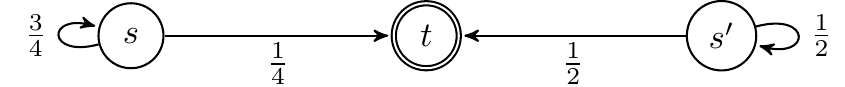}
\caption{Unbounded ratio but language equivalent.}\label{fig:uneq}
\end{figure}

\begin{rem}\label{rem:lc}
Recall  that NFA language containment is $\NL$-complete if the automata are in fact deterministic, in $\P$ if they are unambiguous \cite[Theorem~3]{colcombet2015unambiguity}, $\coNP$-complete if they are unary~\cite{stockmeyer1973word} and  $\PSPACE$-complete in general~\cite{MS72}.
In all cases this complexity level will match, or be lower than that for our respective algorithm for the big-O problem.
\end{rem}

We observe that, if $s$ is big-O of $s'$, the LC condition must hold and so  the LC condition is the first step in each of our verification routines.
\cref{eg:lccondition} shows that the condition alone is not sufficient to solve the big-O problem,
because two states can admit the same set of words with non-zero weight, yet the  weight ratios become unbounded.

\begin{exa}\label{eg:lccondition}

Consider the unary automaton $\wa$ in~\cref{fig:uneq}.
We have $\lng{s}{\wa}=\lng{s'}{\wa} = \{a^n\ |\ n \ge 1\}$, but
$\frac{\nu_s(a^n)}{\nu_{s'}(a^n)}= \frac{(0.75)^{n-1}\cdot 0.25}{(0.5)^{n-1} \cdot 0.5}= 0.5 \cdot 1.5^{n-1} \trns[n\to\infty]{} \infty.$

\end{exa}

\subsection{The Eventually Big-O Problem}
\label{remark:asymptotic-big-oh}
Recall from \cref{remark:asymptotic-big-oh:first} that the  eventually big-O problem asks whether there exists $C> 0, k> 0$ such that forall $w \in \Sigma^{\ge k}$ we have $ \nu_s(w) \le C \cdot \nu_{s'}(w)$. We now justify our focus on the Big-O problem by showing the close relationship between the two problems.

The big-O problem reduces to its eventual variant by checking both the LC condition and the eventually big-O condition. Thus our undecidability (and hardness) results
transfer to the eventually big-O problem. The eventually big-O problem can be solved via the big-O problem by ``fixing'' the LC condition through the addition of a branch from $s'$ that accepts all appropriate words with very low probability.

Let $\wa=\abra{Q,\Sigma,M,\{t\}}$ be a weighted automaton,  $s,s'\in Q$, and $s\neq s'$.
Below, whenever we write $\nu_s$ (resp. $\nu_{s'}$), this will refer to word weights from $s$ (resp. $s'$) in $\wa$.
We assume without loss of generality that the state $s'$ has no incoming transitions
(if this is not the case, a fresh copy of~$s'$ can be made).

Choose $\delta$ to be a real number such that $0 < \delta < 1$ and $\delta$ is smaller than any positive weight in $\wa$.
Construct $\wa'$ by adding the following transitions for all $x\in\Sigma$:
\[
s'\trns[x]{\delta} t \qquad s'\trns[x]{\delta} \bullet \qquad \bullet \trns[x]{\delta} \bullet \qquad \bullet \trns[x]{\delta} t,
\]
where $\bullet$ is a new state. Consequently, for any $w\in\Sigma^+$, we get:
\begin{enumerate}
\item the weight of $w$ in $\wa'$ from $s'$ is $\nu_{s'}(w)+\delta^{|w|}$,
\item if $\nu_{s'}(w)>0$ then $\nu_{s'}(w) > \delta^{|w|}$.
\end{enumerate}
\begin{prop}
$s$ is eventually big-O of $s'$ in $\wa$ if and only if $\lng{s}{\wa}\setminus\lng{s'}{\wa}$ is finite and $s$ is big-O of $s'$ in $\wa'$.
\end{prop}
\begin{proof}\hfill
\begin{dir}[$\Rightarrow$] Suppose $s$ is eventually big-O of $s'$ in $\wa$, i.e. there exist $C,k$ such that, for all $w\in \Sigma^{\ge k}$, $\nu_s(w)\le C \nu_{s'}(w)$.
Note that, for $w\in\Sigma^{\ge k}$, this implies that, whenever $\nu_s(w)>0$, we must also have $\nu_{s'}(w)>0$. Consequently, $\lng{s}{\wa}\setminus\lng{s'}{\wa}\subseteq \Sigma^{<k}$, i.e.
 $\lng{s}{\wa}\setminus\lng{s'}{\wa}$ must be finite.

 Let $w\in\Sigma^\ast$, $m= \max_{w\in \Sigma^{<k}} \nu_s (w)$ and $C'=\frac{m}{\delta^k}$.
 \begin{itemize}
 \item If $w\in \Sigma^{\ge k}$, we have $\nu_s(w)\le C \nu_{s'}(w) \le C(\nu_{s'}(w) + \delta^{|w|})$.
 \item If $w\in \Sigma^{<k}$, then $\nu_s(w)\le m = \frac{m}{\delta^k} \delta^k =C' \delta^k \le C'\delta^{|w|} \le C'(\nu_{s'}(w)+\delta^{|w|})$.
 Note that $\delta^k \le \delta^{|w|}$ follows from $w\in\Sigma^{<k}$ and $0<\delta<1$.
 \end{itemize}

 Taking  $\max(C,C')$ as the relevant constant, we can conclude that $s$ is big-O of $s'$ in $\wa'$.
 \end{dir}
 \begin{dir}[$\Leftarrow$]
Suppose $\lng{s}{\wa}\setminus\lng{s'}{\wa}$ is finite and $s$ is big-O of $s'$ in $\wa'$.
Because $\lng{s}{\wa}\setminus\lng{s'}{\wa}$ is finite, there exists $k$ such that, for all $w\in\Sigma^{\ge k}$, $\nu_s(w)>0$ implies $\nu_{s'}(w)>0$.
Because $s$ is big-O of $s'$ in $\wa'$, there exists $C$ such that $\nu_s(w)\le C(\nu_{s'}(w)+\delta^{|w|})$ for any $w\in\Sigma^\ast$.

Let $w\in\Sigma^{\ge k}$.
From $s$ being big-O of $s'$, we get $\nu_s(w)\le C (\nu_{s'}(w)+\delta^{|w|})$.
\begin{itemize}
\item If $\nu_s(w) > 0$ then $\nu_{s'}(w)>0$. By construction of $\wa'$, we get $\nu_{s'}(w)> \delta^{|w|}$, so
\[
\nu_s(w)\le C(\nu_{s'}(w) + \delta^{|w|}) < C(\nu_{s'}(w)+\nu_{s'}(w)) = 2C \nu_{s'}(w).
\]
\item If $\nu_s(w)=0$ then we also have $\nu_s(w) = 0 \le 2C \nu_{s'}(w)$.
\end{itemize}
Consequently, for any $w\in\Sigma^{\ge k}$, $\nu_s(w)\le 2C \nu_{s'}(w)$, i.e. $s$ is eventually big-O of $s'$ in $\wa$.\qedhere
\end{dir}
\end{proof}
The above argument relied on completing the automaton so that any word is accepted with some weight.
To transfer our decidability results for bounded languages, it will be necessary to complete the automaton with respect to a bound, i.e. the extra weights are added only for words from
$a_1^+\cdots a_m^+$, $a_1^\ast\cdots a_m^\ast$, $w_1^\ast\cdots w_k^\ast$ respectively.
This can be done easily by introducing the extra transitions according to DFA for the bounding language.

\section{The big-O problem for unary weighted automata is \texorpdfstring{$\coNP$}{coNP}-complete}\label{sec:bigounaryconp}
In this section we show $\coNP$-completeness in the unary case.
\begin{thm} \label{thm:tvconp}
The big-O problem for unary weighted automata is $\coNP$-complete. It is $\coNP$-hard even for unary labelled Markov chains.
\end{thm}

For the upper bound, our analysis will refine the analysis of the growth of powers of non-negative matrices of Friedland and Schneider~\cite{friedland1980growth,schneider1986influence} which gives the asymptotic order of growth of $A^n_{s,t} + A^{n+1}_{s,t} + \dots + A^{n+q}_{s,t} \approx \rho^nn^k$ for some $\rho,k$ and $q$, which smooths over the periodic behaviour (see \cref{thm:fands}).  Our results require a non-smoothed analysis, valid for each $n$. This is not
provided in \cite{friedland1980growth,schneider1986influence}, where the smoothing forces the existence of a single limit---which we don't require.  Our big-$\Theta$ lemma (\cref{lemma:unarythetabound}) will accurately characterise the asymptotic behaviour of $A^n_{s,t}$ by exhibiting the correct value of $\rho$ and $k$ for every word.

We show that these asymptotic behaviours can be captured using suitably defined finite automata. We then reduce the big-O problem to language containment problems on these automata. We start by introducing several definitions, and the result of Friedland and Schneider in \cref{sec:conpsec:hardness:prelims}, before implementing the proof of the upper bound in \cref{sec:conpsec:upperbound}. We show $\coNP$-hardness in \cref{sec:conphardness} to complete the proof of \cref{thm:tvconp}.

\subsection{Preliminaries}
\label{sec:conpsec:hardness:prelims}
Let $\wa$ be a unary non-negative weighted automaton with states $Q$, transition matrix $A$ and a unique final state $t$.
When we refer to a \textit{path} in $\wa$,  we mean a path in the NFA of $\wa$, i.e. paths only use transitions with non-zero weights and
states on a path may repeat.

\begin{defi}\hfill
\begin{itemize}
\item  A state $q$ can \textit{reach} $q'$ if there is a path from $q$ to $q'$. In particular, any state $q$ can always reach itself.
\item A \textit{strongly connected component} (SCC) $\scc \subseteq Q$ is a maximal set of states such that for each $q,q' \in \scc$, $q$ can reach $q'$. We denote by $SCC(q)$ the SCC of state $q$ and  by $A^\scc$ the $|\scc| \times |\scc|$ transition matrix of $\scc$. Note every state is in a SCC, even if it is a singleton.

\item
The DAG of $\wa$ is the directed acyclic graph of strongly connected components. Components $\scc, \scc'$ are connected by an edge if there exist $q\in\scc$ and $q'\in\scc'$ with $A(q,q') > 0$.

\item
The spectral radius of an $m\times m$ matrix $A$ is the largest absolute value of its eigenvalues. Recall the eigenvalues of $A$ are  $\{\lambda \in \mathbb{C} \mid \text{ exists vector } \vec{x} \in \mathbb{C}^m, \vec{x}\ne 0 \text{ with } A\vec{x} = \lambda \vec{x} \}$. The spectral radius of $\scc$, denoted by $\rho_\scc$, is the spectral radius of $A^\scc$. By $\rho(q)$ we denote the spectral radius of the SCC in which $q$ is a member. 

\item We denote by $\period^\scc$ the period of the SCC $\scc$: the greatest common divisor of return times for some state $s\in \scc$, i.e. $\operatorname*{gcd}\{t\in \mathbb{N} \ | \ A^t(s,s) > 0\}$. It is known that any choice of state in the SCC gives the same value (see e.g.,~\cite[Theorem 1.20]{sericola2013markov}). If  $A^\scc = [0]$ then $\period^\scc = 0$.
\item Let $\mathscr{P}(s,s')$ be the set of paths from the SCC of $s$ to the SCC of $s'$ in the DAG of $\wa$. Thus a path $\pi \in \mathscr{P}(s,s')$ is a sequence of SCCs $\scc_1, \dots, \scc_m$.
\item $\period(s,s')$, called the local period between $s$ and $s'$, is defined by $\period(s,s') = \operatorname*{lcm}\limits_{\pi \in \mathscr{P}(s,s')}\operatorname*{gcd}\limits_{\scc \in \pi}\quad \period^\scc$.
\item The spectral radius between states $s$ and $s'$, written $\rho(s,s')$, is the largest spectral radius of any SCC seen on a path from $s$ to $s'$:
$\rho(s,s') = \max_{\pi \in \mathscr{P}(s,s')} \rho(\pi)$,
where $\rho(\pi) = \max_{\scc \in \pi} \rho_\scc$ for  $\pi \in \mathscr{P}(s,s')$.
\item
 The following function captures the number of SCCs which attain the largest spectral radius on the path that has the most SCCs of maximal spectral radius. Let $k(s,s')  = \max_{\pi \in \mathscr{P}(s,s')} k(\pi) - 1$,
where, for $\pi \in \mathscr{P}(s,s')$, $k(\pi) = |\{\scc \in \pi \ |\  \rho_\scc = \rho(s,s') \}|$.

\end{itemize}
\end{defi}

\begin{lem}\label{lem:algebraic}
Given $A^\scc$, a representation of the value $\rho_\scc$ can be found in polynomial time. This representation will admit polynomial time testing of $\rho_\scc > \rho_{\scc'}$ and $\rho_\scc = \rho_{\scc'}$ and can be embedded into the first order theory of the reals.
\end{lem}

\begin{proof}
An algebraic number $z$ can be represented as a tuple $(p_z, a,b,r) \in \mathbb{Q}[x] \times \mathbb{Q}^3$. Here $p_z$ is a polynomial over $x$ and $a,b,r$ form an approximation such that $z$ is the only root of $p_z(x)$ with $|z- (a+bi)| \le r$.

Given a polynomial, one can find the representation of each of its roots in polynomial time using the root separation bound due to Mignotte~\cite{Mig82}. Operations such as addition and multiplication of two algebraic numbers, finding $|x|$, testing if $x > 0$ can be done in polynomial time in the size of the representation $(p,a,b,r)$, yielding the same representation~\cite{basu2005betti,cohen2013course,pan1996optimal}; see, e.g., \cite[Section~3]{ouaknine2014positivity} for a summary.

Any coefficient of the characteristic polynomial of an \textit{integer} matrix can be found in $\GapL$ \cite{hoang2003complexity}. $\GapL$ is the difference of two $\#\L$ calls, each of which can be found in $\NC^2 \subseteq \P$. Here the matrix will be rational; but it can be normalised to an integer matrix by a scalar, the least common multiple of the denominator of each rational. Whilst this number could be exponential, it is representable in polynomial space. The final eigenvalues can be  renormalised by this constant.

The characteristic polynomial of an $n\times n$ matrix has degree at most $n$, since each coefficient can be found in polynomial time, the whole characteristic polynomial can be found in this time. Thus by enumerating its roots (at most $n$), taking the modulus of each, and sorting them ($a > b \iff a + (-1) \times b >0$) we can find the spectral radius in this form $(p_z, a,b,r)$.

Note that the spectral radius is a real number, so that given the spectral radius in the form $(p_z, a,b,r)$ we actually have $b = 0$. Then the number can be encoded exactly in the first order theory of the reals using $\exists x : p_z(x) = 0 \wedge x - a \le r \wedge a - x \le r$.
\end{proof}

The asymptotic behaviours of weighted automata will be characterised using $(\rho,k)$-pairs:
\begin{defi}\label{defn:lexiorder} A  $(\rho,k)$-pair is an element of $\mathbb{R}\times\mathbb{N}$. The ordering on  $\mathbb{R}\times\mathbb{N}$ is lexicographic, i.e.
$(\rho_1, k_1) \le (\rho_2, k_2) \iff \rho_1 < \rho_2 \vee ( \rho_1 = \rho_2 \wedge k_1 \le k_2).$\end{defi}

Friedland and Schneider \cite{friedland1980growth,schneider1986influence} essentially use $(\rho,k)$-pairs to show the asymptotic behaviour of the powers of non-negative matrices. In particular they find the asymptotic behaviour of the sum of several $A_{s,s'}^n$, smoothing the periodic behaviour of the matrix.

\begin{thm}[Friedland and Schneider \cite{friedland1980growth,schneider1986influence}] \label{thm:fands} Let $A$ be an $m\times m$ non-negative matrix, inducing a unary weighted automaton $\wa$ with states $Q = \{1,\dots,m\}$. Given $s,t \in Q$, let $B_{s,t}^n = A_{s,t}^n + A_{s,t}^{n+1} + \dots + A_{s,t}^{n+\period(s,t)-1}$. Then $\lim_{n\to\infty} \frac{B_{s,t}^n}{\rho(s,t)^n n^{k(s,t)}} = c, \quad 0 < c < \infty$.

\end{thm}

In the case where the local period is $1$ ($\period(s,t)=\period(s',t) = 1$), \cref{thm:fands} can already be used to solve the big-O problem (in particular if the matrix $A$ is aperiodic). In this case  $A_{s,t}^n = B_{s,t}^n = \Theta( \rho(s,t)^n n^{k(s,t)} )$. Then to establish that $s$ is big-O of $s'$ we check that the language containment condition holds and that $(\rho(s,t),k(s,t)) \le (\rho(s',t),k(s',t))$.
However, this is not sufficient if the local period is not $1$.
\begin{exa}
\begin{figure}
\centering
\includegraphics[width=0.6\linewidth]{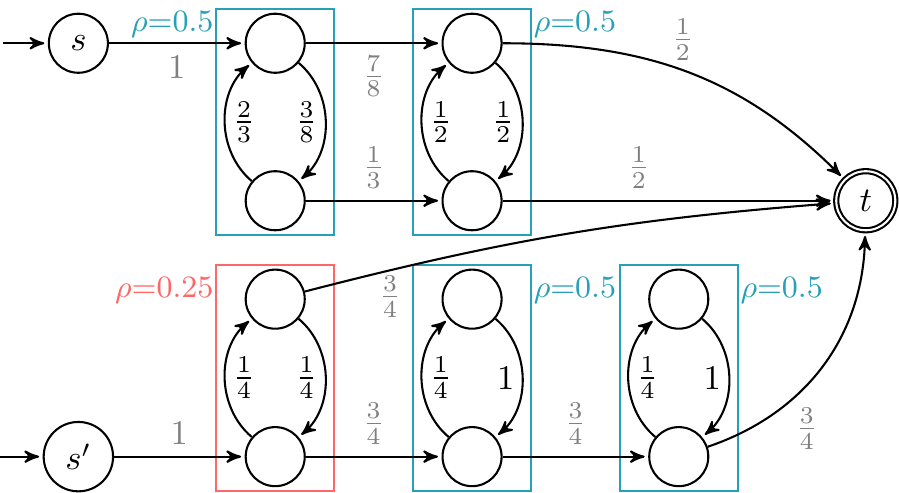}
\caption{Different rates for different phases.}
\label{fig:different}

\end{figure}
Consider the chains shown in \cref{fig:different} with local period $2$. The behaviour for $n\ge 3$ is
$A_{s,t}^n =  \Theta(0.5^n n )$ and $A_{s',t}^n =\Theta(0.25^n )$ when $n$ is odd and $A_{s',t}^n =\Theta(0.5^nn )$ when $n$ is even. However, \cref{thm:fands} tells us $B_{s,t}^n = \Theta(0.5^nn)$ and $B_{s',t}^n = \Theta(0.5^nn)$ suggesting the ratio is bounded, but in fact $s$ is not big-O $s'$ (although $s'$ is big-O of $s$) because $\frac{A_{s,t}^{2n+1}}{A_{s',t}^{2n+1}} \trns[n\to\infty]{} \infty$.

\end{exa}

\subsection{Upper bound: The unary big-O problem is in \texorpdfstring{$\coNP$}{coNP}}
\label{sec:conpsec:upperbound}
Let $\wa$ be a unary weighted automaton and suppose we are asked whether $s$ is big-O of $s'$. We assume  w.l.o.g.\label{conp-ub-assumptions}~(a) that there is a unique final state $t$ with no outgoing transitions, and (b) that $s,s'$ do not appear on any cycle (if this is not the case, copies of $s,s'$ and their transitions can be taken).

Next we define  a ``degree function'', which captures the asymptotic behaviour of each word $a^n$ by a ($\rho,k$)-pair, capturing the exponential and polynomial behaviours respectively.

\begin{defi}\label{def:unarydegreefunction}
Given a unary weighted automaton $\wa$, let $\dgf_{s,t} : \mathbb{N} \to \mathbb{R} \times \mathbb{N}$ be defined by $\dgf_{s,t}(n) = (\rho,k)$, where:
\begin{itemize}
  \item $\rho$ is the largest spectral radius of any vertex visited on any path of length $n$ from $s$ to $t$;
  \item the path from $s$ to $t$ that visits the most SCCs of spectral radius $\rho$ visits $k+1$ such SCCs;
  \item if there is no length-$n$ path  from $s$ to $t$, then $(\rho,k){=}(0,0)$.
\end{itemize}
\end{defi}

The set of \textit{admissible} $(\rho,k)$-pairs is the image of $\dgf_{s,t}$. Observe that this set is finite and of size at most $|Q|^2$: there can be no more than $|Q|$ values of $\rho$ (if at worst each state were its own SCC) and the value of $k$  is also bounded by the number of SCCs and thus $|Q|$.

To prepare for the proof of our key big-$\Theta$ lemma (stated below),
we next define the $(\rho,k)$-annotated version of the weighted automaton $\wa$.
Namely, in each state we record the relevant value of $(\rho, k)$ corresponding to the current run to the state.
\begin{defi}[The weighted automaton $\waannotated$] \label{defn:annotated} Given $\wa=\abra{Q,\Sigma,A,\{t\}}$ and $s \in Q$,
the weighted automaton $\waannotated$ has states of the form $(q,\rho, k)$ for all $q\in Q$ and all admissible $(\rho,k)$-pairs, the same $\Sigma$ and no final states.
For every transition $q\trns{p} q'$ from $\wa$ denoting $A(q,q') = p$, include the following transition in $\waannotated$ for
each admissible $(\rho,k)$:
\pagebreak
\begin{itemize}
\item $(q, \rho, k) \trns{p} (q', \rho, k)$ if $\SCC(q) {=}\SCC(q')$,
\item $(q, \rho, k) \trns{p} (q', \rho, k+1)$ if $\SCC(q){\ne}\SCC(q') \text{ and } \rho = \rho(q')$,
\item $(q, \rho, k) \trns{p} (q', \rho, k)$ if $\SCC(q){\ne}\SCC(q') \text{ and } \rho > \rho(q')$,
\item $(q, \rho, k) \trns{p} (q', \rho(q'), 0)$ if $\SCC(q){\ne}\SCC(q') \text{ and } \rho(q') > \rho$.
\end{itemize}

\end{defi}
Observe that the automaton $\waannotated$ is constructable in polynomial time given $\wa$. Indeed, the spectral radii of all SCCs can be computed and compared to each other in time polynomial in the size of $\wa$ (see \cref{lem:algebraic}).

Let $s,t\in Q$ be fixed. We are now ready to state and prove the key technical lemma of this subsection (cf. \cref{thm:fands}, Friedland and Schneider \cite{friedland1980growth,schneider1986influence}), where we assume the functions $\rho(n), k(n)$, defined by  $\dgf_{s,t}(n) = (\rho(n), k(n))$.
\begin{lem}[The big-$\Theta$ lemma] \label{lemma:unarythetabound}
There exist $c,C > 0$ such that, for every $n > |Q|$,
\[
  c \cdot \rho(n)^n n^{k(n)} \le A^n_{s,t} \le C \cdot \rho(n)^n n^{k(n)}.
\]
\end{lem}

\begin{proof}
\hfill

\begin{dir}[{lower bound}]
Let us fix $(\rho',k')$ and show the bound for the elements of $\{n\in\mathbb{N} \mid d_{s,t}(n)= (\rho',k')\}$.

A witnessing path in $\wa$ is a length-$n$ path from $s$ to $t$ that visits $k'+1$ SCCs of spectral radius $\rho'$ and no SCC with a  larger spectral radius.
Let $\pi = \scc_1\dots\scc_k \in \mathscr{P}(s,t)$  be the corresponding sequence of SCCs visited by some witnessing path and let $s_i, e_i$ ($1\le i\le k$) be the entry and exit points (respectively into and out of $\scc_i$) on that path,
i.e.  $s=s_1$, $\SCC(s_i) = \SCC(e_i) = \scc_i$ ($1 \le i \le k$), there is a transition (of positive weight) from $e_i$ to $s_{i+1}$ and $e_k=t$. We write $\vec{s_i,e_i}$ to represent the particular sequence of entry/exit points.

Let us define  a new unary weighted automaton  $\wa^{\vec{s_i,e_i}}$ to be a restriction of $\wa$ so that
the only transitions between its SCCs are from $e_i$ to $s_{i+1}$, for $1 \le i < k$,
i.e., the weight is reduced to zero for any violating transition.  There are finitely many such $\vec{s_i,e_i}$, so that if $n \in \{n\in\mathbb{N} \mid d_{s,t}(n)= (\rho',k')\}$ then there exists some $\wa^{\vec{s_i,e_i}}$ with transition matrix $D$ such that $D_{s,t}^n > 0$.

Let us restrict to a single choice of $\vec{s_i,e_i}$, and fix $D$ be the transition matrix of $\wa^{\vec{s_i,e_i}}$. Clearly $A_{s,t}^n \ge D_{s,t}^n$ for all $n$, since $\wa^{\vec{s_i,e_i}}$ is a restriction of $\wa$. We now show that $D_{s,t}^n \ge c_{\vec{s_i,e_i}}(\rho')^n n^{k'}$ for some $c_{\vec{s_i,e_i}} > 0$ whenever $D_{s,t}^{n}>0$. Hence, this will imply that $A_{s,t}^n\ge c_{\vec{s_i,e_i}}(\rho')^n n^{k'}$.

Note that, in $\wa^{\vec{s_i,e_i}}$, $\rho(s,t)=\rho'$ and $k(s,t)=k'$,
because  all paths from $s$ to $t$ must visit $k'+1$ SCC's with spectral radius $\rho'$. Therefore, by \cref{thm:fands} there exists $c'>0$ be such that
\[\lim_{m\to\infty} \frac{D_{s,t}^m + D_{s,t}^{m+1} + \dots + D_{s,t}^{m+T-1}}{(\rho')^m m^{k'}} = c', \]
where $\period$ is the local period from $s$ to $t$ in $\wa^{\vec{s_i,e_i}}$. By the definition of limit,
for all $\varepsilon > 0$ the ratio is at least $c' - \varepsilon$
if $m$ is big enough.
Pick $\varepsilon = c'/2$, then the inequality
$D_{s,t}^m + D_{s,t}^{m+1} + \dots + D_{s,t}^{m+T-1} \ge (c'/2) \cdot (\rho')^m m^{k'}$
holds for all $m \ge N$ where $N$ is a constant dependent
only on $\vec{s_i,e_i}$.

Next we shall show that whenever $D_{s,t}^{n}>0$, that is the length $n$ word has a positive path from $s$ to $t$ in $\wa^{\vec{s_i,e_i}}$, we have $D_{s,t}^{n+1} + \dots + D_{s,t}^{n+T-1} = 0$. This will imply $D_{s,t}^n  \ge (c'/2) (\rho')^n n^{k'}$ (if $n > N$)
and, hence, $A_{s,t}^n \ge (c'/2) (\rho')^n n^{k'}$. Let $L$ be the length of  the shortest path from $s$ to  $t$ in $\wa^{\vec{s_i,e_i}}$.
Observe that paths from $s$ to $t$ in $\wa^{\vec{s_i,e_i}}$ can only have lengths from
$\{L + \ell_1 \cdot \period^{\SCC(s_1)} + \dots + \ell_k \cdot \period^{\SCC(s_k)} \ | \ \ell_1, \dots, \ell_k \in \mathbb{N} \}$
and, thus,
$\{L + \ell\cdot \gcd\{\period^{\SCC(s_1)},\dots, \period^{\SCC(s_k)}\} \ | \ \ell\in \mathbb{N} \}$.
As $\mathscr{P}(s,t)=\{\pi\}$ in $\wa^{\vec{s_i,e_i}}$,  $T = \gcd\{\period^{\SCC(s_1)},\dots, \period^{\SCC(s_k)}\}$.
Consequently, all paths from $s$ to $t$ in $\wa^{\vec{s_i,e_i}}$ have length in $\{ L+\ell T\mid \ell \in\mathbb{N}\}$.
Hence,  whenever $D_{s,t}^n$ is positive, there are no paths which can
contribute positive value to $D_{s,t}^{n+1} + \dots + D_{s,t}^{n+T-1}$.

For small $n\le N$, we can always take $c_{\vec{s_i,e_i}}$ small enough so that $D_{s,t}^n \ge c_{\vec{s_i,e_i}}(\rho')^n n^{k'}$ when $D_{s,t}^n > 0$. Take $c'' = \min D^m_{s,t} / (\rho')^m m^{k'}$ where
the minimum is over all $m \le N$ for which $D^m_{s,t} > 0$. Then $D^m_{s,t} / (\rho')^m m^{k'} \ge c''$ for all such $m$.

This means we can choose $c_{\vec{s_i,e_i}} = \min(c'/2, c'') > 0$,
regardless of whether $n$ is big or small:
this constant depends on $\vec{s_i,e_i}$ but not on~$n$.

As $c_{\vec{s_i,e_i}}$ depends only on $\vec{s_i,e_i}$,
to finish the proof it suffices to take $c$
to be the smallest  among the finitely many $c_{\vec{s_i,e_i}}$.
\end{dir}
\begin{dir}[upper bound]
Let $N_{(\rho',k')} = \{n\ |\ \dgf_{s,t}(n) = (\rho',k')\}$. This gives a finite partition of $\mathbb{N}$ as $\bigcup_{(\rho,k)} N_{(\rho,k)}$.
For each $(\rho',k')$, we shall find a value $C_{(\rho',k')}$ so that, for $n \in N_{(\rho',k')}$, we have $A_{s,t}^n \le C_{(\rho',k')} (\rho')^n n^{k'}$.
Then, to have $A_{s,t}^n \le C  \rho(n)^n n^{k(n)}$  for all $n\in \mathbb{N}$, it will suffice to take $C$ to be the maximum  over all $C_{(\rho',k')}$.

Let us fix $(\rho',k')$. Consider $\wa^\bullet$ to be $\waannotated$ in which,  for every $(\rho,k) \le (\rho',k')$, we merge the states $(t,\rho,k)$ into a single final state $t'$.
Let us rename the state $(s,0,0)$ to $s'$.
(This merger and renaming are justified by assumptions~(a) and~(b) made
 at the beginning of this subsection, see p.~\pageref{conp-ub-assumptions}.)
Let $E$ be the corresponding transition matrix of $\wa^\bullet$.
Note that all paths from $s'$ to $t'$ in $\wa^\bullet$ go through at most $k'+1$ SCCs with spectral radius $\rho'$.

\begin{clm}
\label{c:a=e}
For all $n \in N_{(\rho',k')}$, we have $A_{s,t}^n = E_{s',t'}^n$.
\end{clm}

\begin{subproof}
Consider any path $s \to q_1\to\dots \to q_m \to t$ in $\wa$. There is a corresponding path  in $\wa^\bullet$,
however the states $q_i$ are annotated as $(q_i,\rho,k)$, where $\rho$ is the largest spectral radius
seen so far, and $k+1$ is the number of SCC's of that radius number seen so far.
The only paths removed are those terminating at $(t,\rho,k)$ with $(\rho, k) > (\rho',k')$. Since $\dgf_{s,t}(n) = (\rho',k')$, we know that no path visits more than $k'+1$ SCCs of spectral radius $\rho'$, or an SCC of spectral radius greater than $\rho'$. Consequently,  no such path is disallowed in $\wa^\bullet$. No paths were added either.
Because every SCC in $\wa$ remains a strongly connected component in $\wa^\bullet$ (duplicated with various $(\rho, k)$)
and its transition probability matrix (and hence the spectral radius) remains the same, we can conclude that
$A_{s,t}^n = E_{s',t'}^n$.
This completes the proof of Claim~\ref{c:a=e}.
\end{subproof}

\begin{clm}
\label{c:ub-ub}
There exists $C_{(\rho',k')}$ such that  $A_{s,t}^n \le C_{(\rho',k')} (\rho')^n n^{k'}$.
\end{clm}

\begin{subproof}
We have  $A_{s,t}^n = E_{s',t'}^n \le E_{s',t'}^n + E_{s',t'}^{n+1} + \dots + E_{s',t'}^{n+\period(s',t')-1}$, where $\period(s',t')$ is the local period between states $s'$ and $t'$ in $\wa^\bullet$.
By \cref{thm:fands}, there exists $C_{(\rho',k')}$ such that this quantity is bounded by $C_{(\rho',k')} (\rho')^n n^{k'}$.
Thus, for $n \in N_{(\rho',k')}$, we have $A_{s,t}^n \le C_{(\rho',k')} (\rho')^n n^{k'}$.
This concludes the proof of Claim~\ref{c:ub-ub}.
\end{subproof}

\noindent
By the argument above, the upper bound of Claim~\ref{c:ub-ub}
completes the proof of Lemma~\ref{lemma:unarythetabound}.\qedhere
\end{dir}
\end{proof}

\pagebreak

\begin{rem}\label{rem:whynowork}
When $n \le|Q|$ it is possible that $A^n_{s,t} > 0$ but every path may never take any loops. In which case our degree function allocates $d_{s,t}(n) = (0,0)$, and so Lemma~\ref{lemma:unarythetabound} would not hold as $A^n_{s,t} \not\le 0^nn^0$. The behaviour on short words will be handled by the language containment condition.
\end{rem}

We will now see how the big-$\Theta$ lemma (Lemma~\ref{lemma:unarythetabound})
enables us to characterise the big-O relation on states of weighted automata.
For the following lemma, recall the language containment (LC) condition from \cref{def:lc} and the ordering on $(\rho,k)$-pairs from \cref{defn:lexiorder}.

\begin{lem} \label{lem:boundedff}
A state $s$ is big-O of $s'$ if and only if the LC condition holds and, for all but finitely many $n \in \mathbb{N}$, we have $\dgf_{s,t}(n) \le \dgf_{s',t}(n)$.
\end{lem}

\begin{proof}
Let us start off by giving a short summary of the proof.
The idea is that, whenever $\dgf_{s,t}(n) \le \dgf_{s',t}(n)$, by \cref{lemma:unarythetabound},
we have $\nu_s(a^n) \le (\frac{C}{c}(\frac{\rho}{\rho'})^n n^{k-k'})\cdot\nu_{s'}(a^n) $, in which case either $\dgf_{s,t}(n) = \dgf_{s',t}(n)$ and $(\frac{\rho}{\rho'})^n n^{k-k'} =  1$ or $\lim_{n\to\infty}(\frac{\rho}{\rho'})^n n^{k-k'} = 0$ and so $(\frac{\rho}{\rho'})^n n^{k-k'} \le 1$ for all but finitely many $n$.
At the same time, whenever $\dgf_{s,t}(n) > \dgf_{s',t}(n)$, \cref{lemma:unarythetabound} yields $\nu_s(a^n) \ge (\frac{c}{C}(\frac{\rho}{\rho'})^n n^{k-k'})\cdot\nu_{s'}(a^n)$ but then $\lim_{n\to\infty}(\frac{\rho}{\rho'})^n n^{k-k'}= \infty$.

We now show how to fill in the details in this summary.
First we note some consequences of $\dgf_{s,t}(n) \le \dgf_{s',t}(n)$.
Suppose  $\dgf_{s,t}(n) = (\rho,k)$ and $\dgf_{s',t}(n) = (\rho',k')$. Thanks to Lemma~\ref{lemma:unarythetabound},
we have $\nu_s(a^n) \le (\frac{C}{c}(\frac{\rho}{\rho'})^n n^{k-k'})\cdot\nu_{s'}(a^n) $.
If $\dgf_{s,t}(n) \le \dgf_{s',t}(n)$ we can distinguish two cases: either  $(\rho,k)= (\rho',k')$ or $(\rho,k)<(\rho',k')$.
\begin{itemize}
\item In the former case, $(\frac{\rho}{\rho'})^n n^{k-k'} = 1$ and, thus,  ${\nu_s(a^n)} \le (\frac{C}{c})\cdot {\nu_{s'}(a^n)}$.
\item In the latter case, we have $\lim_{m\to\infty} (\frac{\rho}{\rho'})^m m^{k-k'}=0$ and, thus, $(\frac{\rho}{\rho'})^m m^{k-k'} < 1$ for
all but finitely many  $m$. Consequently, for all but finitely many $n$, we can conclude ${\nu_s(a^n)} \le (\frac{C}{c})\cdot {\nu_{s'}(a^n)}$.
\end{itemize}

Thanks to the above analysis,
if $\dgf_{s,t} (n) \le \dgf_{s',t}(n)$ holds for all but finitely many $n$, it follows that ${\nu_s(a^n)}\le (\frac{C}{c}) \cdot{\nu_{s'}(a^n)}$ for all but finitely many $n$.
Moreover, the language containment condition implies that $\nu_s(a^n)\le C'\cdot {\nu_{s'}(a^n)}$  for some $C'$ in the remaining (finitely many) cases.
Hence, $s$ is big-O of $s'$, which shows the right-to-left implication.

For the converse, recall that we have already established that ``$s$ is big-O of $s'$'' implies the language containment condition.
For the remaining part, we reason by contrapositive and suppose  that there are infinitely many $n$ with $\dgf_{s,t}(n) > \dgf_{s',t}(n)$.
As there are finitely many values in the range of $\dgf_{s,t}$ and $\dgf_{s',t}$,
there exist $(\rho,k)$ and $(\rho',k')$ such that $(\rho,k)>(\rho',k')$ and, for infinitely many $n$, $\dgf_{s,t}=(\rho,k)$ and $\dgf_{s',t}=(\rho',k')$.
Note that $(\rho',k') \ne (0,0)$, as otherwise $f_{s'}(a^n) = 0$ and,
by the language containment condition, $f_s(a^n) = 0$ and $(\rho, k) = (0, 0)$, a contradiction.
Therefore, $f_{s'}(a^n) > 0$ for all such~$n$, and moreover
Lemma~\ref{lemma:unarythetabound} yields $\nu_s(a^n) \ge (\frac{c}{C}(\frac{\rho}{\rho'})^n n^{k-k'})\cdot\nu_{s'}(a^n)$.
But $(\rho,k)>(\rho',k')$ implies
\[\lim_{m\to\infty} \left(\frac{\rho}{\rho'}\right)^m m^{k-k'} = \infty,
\]
 i.e. $(\frac{\rho}{\rho'})^n n^{k-k'}$ is unbounded.
Thus, $s$ cannot be big-O of $s'$.
\end{proof}

We are going to use the characterisation from \cref{lem:boundedff} to prove \cref{thm:tvconp}.
As already discussed, the LC condition can be checked via NFA inclusion testing. To tackle the ``for all but finitely many ...'' condition,
we introduce the concept of  {eventual inclusion}.

\begin{defi} Given sets $A,B$, we say $A$ is \textit{eventually included} in $B$,
written $A \abfsubset B$, if and only if the set difference $A \setminus B$ is finite.
\end{defi}

This relation has appeared under the name of ``almost inclusion'' in,
e.g.,~\cite{Senizergues93,FinkelT10,IanovskiMNN14}.
Minimisation of finite-state automata up to finitely many errors has been
studied under the name of ``hyper-minimization'' (see, e.g.,~\cite{BadrGS09,Maletti11}).

The next three lemmas relate deciding the big-O problem using the characterisation of  \cref{lem:boundedff} to eventual inclusion.
\begin{lem}\label{thm:abfsubsetconpbounded}
Given unary NFAs $\nfa_1, \nfa_2$, the problem $\lng{}{\nfa_1} \abfsubset \lng{}{\nfa_2}$ is in $\coNP$.
\end{lem}

\begin{proof}[Proof of \cref{thm:abfsubsetconpbounded}]
Let $M$ be a DFA  accepting $\lng{}{\nfa_1} \cap \overline{\lng{}{{\nfa_2}}}$ obtained through standard automata constructions, i.e. $|M| \le 2^{|\nfa_1|+|\nfa_2|}$.
Note that $\lng{}{\nfa_1} \abfsubset \lng{}{\nfa_2}$ if and only if $\lng{}{M}$ is finite.
Observe that $\lng{}{M}$ is infinite if and only if there exists $w\in \lng{}{M}$ with $|M| \le |w| \le 2|M|$.

Consequently, violation of eventual inclusion can be detected by
guessing $n\in \mathbb{N}$ such that $|M|\le n \le 2|M|$ and
verifying $a^n \in \lng{}{M}$.

Even though $M$ is of exponential size, it is possible to verify
$a^{n} \in \lng{}{M}$ in polynomial time, given $n$ in binary.
To this end, we use $\nfa_1, \nfa_2$ instead of $M$ and view their transition functions as matrices.
Then one can verify the condition using fast matrix exponentiation (by repeated squaring).
Because the bit size of $n$ must be polynomial in $|\nfa_1|+|\nfa_2|$, the lemma follows.
\end{proof}

\begin{lem}\label{lem:ei}
Suppose $\dgf_1,\dgf_2: \mathbb{N} \to X$, with $(X,\le)$ a finite total order.
Then $\dgf_1(n)\le \dgf_2(n)$ for all but finitely many $n$ if and only if
$\{ n\,|\, \dgf_1(n) \ge x\} \abfsubset \{ n\,|\, \dgf_2(n)\ge x\}$ for all $x\in X$.
\end{lem}
\begin{proof} 
The left-to-right implication is clear. For the opposite direction, observe that, because the order on $X$ is total, $\dgf_1(n) > \dgf_2(n)$
implies the existence of $x\in X$ such that $\dgf_1(n) \ge x$ and $\dgf_2(n) < x$ (it suffices to take $x=\dgf_1(n)$).
Because $X$ is finite,  $\dgf_1(n)>\dgf_2(n)$ for infinitely many $n$ implies failure of $\{ n\,|\, \dgf_1(n) \ge x\} \abfsubset \{ n\,|\, \dgf_2(n)\ge x\}$ for some $x$.
\end{proof}

\begin{lem} \label{lem:conpprocff}
Given a unary weighted automaton $\wa$, the associated problem whether $\dgf_{s,t}(n) \le \dgf_{s',t}(n)$ for all but finitely many $n \in \mathbb{N}$ is in $\coNP$.
\end{lem}
\begin{proof}
Given an admissible pair $x=(\rho,k)$, we construct an NFA $\nfa_{s,x}$ accepting $\{a^n \ | \ \dgf_{s,t}(n) \ge x\}$ (similarly $\nfa_{s',x}$ for $s'$), by taking the NFA $\nfaof{\waannotated}{s}$ (Definitions~\ref{def:nfaof},~\ref{defn:annotated}) with a suitable choice of accepting states.
Recall that states in $\waannotated$ are of the form $(q,\rho',k')$, where $q$ is a state from $\wa$ and $(\rho',k')$ is admissible.
If we designate states $(t,\rho',k')$ with $(\rho',k')\ge x$ as accepting, it will accept  $\{a^n \ | \ \dgf_{s,t}(n) \ge x\}$.
This is a polynomial-time construction.

Then, by Lemma~\ref{lem:ei}, the problem whether $\dgf_{s,t}(n) \le \dgf_{s',t}(n)$ for all but finitely many $n \in \mathbb{N}$
is equivalent to $\lng{}{\nfa_{s,x}}\abfsubset \lng{}{\nfa_{s',x}}$ for all admissible $x$.
As there are at most $|Q|^2$ values of $x$ and each can be verified co-nondeterministically in $\coNP$, it suffices to show that deciding $\lng{}{\nfa_{s,x}}\abfsubset \lng{}{\nfa_{s',x}}$ is in $\coNP$ for each $x$.
This is the case by  Lemma~\ref{thm:abfsubsetconpbounded}.
\end{proof}

\cref{lem:boundedff}, \cref{rem:lc}, and \cref{lem:conpprocff} together complete the  upper bound result for \cref{thm:tvconp}.

\begin{rem}
\cref{lem:ei} may appear simpler using $\{ n\,|\, \dgf_1(n) = x\} \abfsubset \{ n\,|\, \dgf_2(n) \ge x\}$.
However, it does not seem possible to construct an NFA
for $\{a^n \ | \ \dgf_{s,t}(n) = x\}$ in polynomial time.  Taking just $(t,\rho,k)$, for $x = (\rho,k)$, as accepting
would not be correct, as there could be paths of the same length ending in $(t,\rho',k')$ with $(\rho',k')>(\rho,k)$.
Using $\dgf_1(n) \ge x$ instead of $\dgf_1(n) = x$ avoids this problem.
\end{rem}

\begin{rem}An alternative approach for obtaining an upper bound could be to compute the Jordan normal form of the transition matrix and consider
its powers. Instead of the interplay of strongly connected components
in the transition graph, we would need to consider linear combinations of
the $n$th powers of complex numbers (such as roots of unity).
It is not clear this algebraic approach leads to
a  representation more convenient for our purposes.
\end{rem}

\subsection{\texorpdfstring{$\coNP$}{coNP}-hardness for unary LMC}
\label{sec:conphardness}

Given a unary NFA $\nfa$, the \textit{NFA universality problem} asks if $\lng{}{\nfa} = \{a^n\ |\ n\in \mathbb{N}\}$. This problem is $\coNP$-complete \cite{stockmeyer1973word}.
We exhibit a polynomial-time reduction from  (a variant of) the unary universality problem  to the big-O problem on unary Markov chains.

\begin{thm}\label{thm:boundedhard}
The big-O problem is $\coNP$-hard on unary Markov chains.
\end{thm}

Let us first consider a particular form of unary NFAs.
\begin{defi}\label{defn:chro}
A unary NFA $\nfa = \abra{Q, \to, q_s,F}$ is in \textit{Chrobak normal form} \cite{chrobak1986finite} if
\begin{itemize}
  \item $Q = S \uplus C^1 \uplus \dots \uplus C^m$ and $q_s\in S$;
  \item $S=\{ s_1,\cdots,s_k \}$, $q_s=s_1\in S$ and transitions between states from $S$ form a path $s_1 \trns{a} s_2 \trns{a}  \dots \trns{a} s_k$;
  \item $C^i=\{c^i_0,\cdots, c^i_{|C^i|-1}\}$ ($1\le i\le m$) and transitions between states from  $C^i$ form  a cycle $c^i_0 \trns{a} c^i_1 \trns{a}  \dots \trns{a} {c^i_{|C^i|-1}} \trns{a} c^i_0$;
  \item the remaining transitions connect the end of the path to each cycle:  $s_k \trns{a} c^i_0$ for all $1\le i\le m$.
\end{itemize}
\end{defi}
Any unary NFA can be translated to this representation with at most quadratic blow-up in the size of the machine~\cite{chrobak1986finite}, and such representation can be found in polynomial time~\cite{to2009unary,martinez2002efficient}.
In addition, to simplify our arguments, we introduce a \textit{restricted} Chrobak normal form, which requires that there is exactly one accepting state in each cycle.
This restricted form can be found with at most a further quadratic blow-up over Chrobak normal form, by creating copies of cycles---one for each accepting state in the cycle.

Observe that $S\subseteq F$ is a necessary condition for the universality of a unary NFA in Chrobak normal form. Consequently, the universality problem for unary NFA  \textit{in restricted
Chrobak normal form} such that $k=1$  is already $\coNP$-hard. This is the problem we are going to reduce from in the following.

\begin{proof}[Proof of \cref{thm:boundedhard}]
Let $\nfa = \abra{Q, \trns{}, q_s, F}$ be a unary NFA in restricted Chrobak normal form with $k=1$.
We will construct a unary Markov chain  $\lmc$, depicted in \cref{fig:nfatolmc}, with states $Q' =  Q\cup \{s, u, v, t\}$, where $t$ is final.
\begin{figure}[t]
\centering
\includegraphics[width=0.83\linewidth]{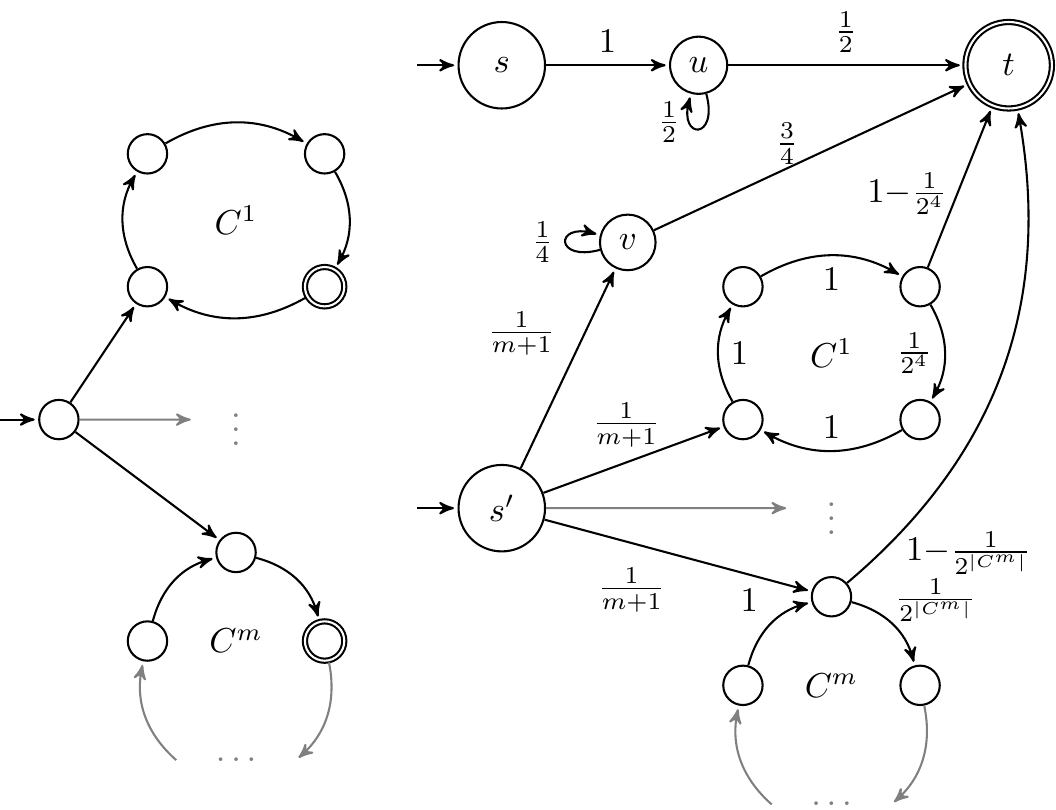}
\caption{Reduction from NFA (left) to LMC (right)}
\label{fig:nfatolmc}
\end{figure}
The branch starting from $s$, defined below, guarantees $\nu_{s}(a^n) = \Theta\pbra{(\frac{1}{2})^n}$.
\[s \trns{1} u \quad \quad u\trns{\frac{1}{2}}u \quad \quad u\trns{\frac{1}{2}}t\]
We take $s'=q_s$ and create a similar branch from $s'$, albeit with a smaller weight,
to create paths of weight $\Theta\pbra{(\frac{1}{4})^n}$ when reading $a^n$.
\[s'\trns{\frac{1}{m+1}}v \quad\quad v\trns{\frac{1}{4}}v \quad \quad v\trns{\frac{3}{4}}t\]
Moreover, we add weights to the original NFA transitions from $\nfa$ as follows:
\[\begin{array}{lll}
s' \trns{\frac{1}{m+1}} c^i_0 \quad (1\le i \le m)&&
c^i_{j\ominus 1} \trns{(\frac{1}{2})^{|C^i|}} c^i_j \quad (c^i_j\in F) \\[3mm]
c^i_{j\ominus 1} \trns{1-(\frac{1}{2})^{|C^i|}} t \quad (c^i_j\in F)&& c^i_{j\ominus 1} \trns{1} c^i_j \quad (c^i_j\not\in F) \
\end{array}\]
where $j\ominus 1=(|C^i| + j-1)\bmod |C^i|$.
Note that the weights have been selected as if each letter were read with weight $\frac{1}{2}$ except for a bounded number of transitions, where the bound is $\max |C^i|$.
Consequently, whenever there are accepting paths for $a^n$ in $\nfa$, their overall weight  in $\lmc$ will be $\Theta\pbra{(\frac{1}{2})^n}$.

It it easy to check that the reduction produces an LMC and can be carried out in polynomial time. It remains to argue that the reduction is correct.

If $\nfa$ is not universal, there exists $n$ such that $a^n\not\in F$. Because of the cyclic structure of Chrobak normal form, $a^{n_k}\not\in F$ for $n_k = n + k  q$, where
$q = \operatorname*{lcm} \{|C^1|,\dots, |C^m|\}$ and $k\in \mathbb{N}$.
Then,  by the earlier observations about growth,
there exists $C>0$ such that $\sup_k  \frac{\nu_s(a^{n_k})}{\nu_{s'}(a^{n_k})}   = \sup_k C\, \frac{(1/2)^{n_k}}{ (1/4)^{n_k}} = \sup_k C\,2^{n_k} = \infty$,
i.e. $s$ is not big-O of $s'$.

If $\nfa$ is universal then, starting from $s'$ in $\lmc$, every word $a^n$ will have a path weighted $\Theta\pbra{(\frac{1}{4})^n}$ as well as
paths weighted $\Theta\pbra{(\frac{1}{2})^n}$. Hence,
there exists $C>0$ such that
\[
\sup_n \pbra{\frac{\nu_s (a^{n}) }{\nu_{s'} (a^{n}) } }\le
\sup_n \pbra{C\, \frac{ (\frac{1}{2})^{n} }{ (\frac{1}{4})^{n} + (\frac{1}{2})^{n} }}\le C,
\]
i.e. $s$ is big-O of $s'$.
\end{proof}

\begin{rem}We note that the $\frac{1}{4}$ branch via state $v$ is not strictly necessary, but it demonstrates that the problem is hard even if the LC condition is satisfied (i.e., ``it can be the numbers that make the hardness'').
\end{rem}

\section{Decidability for weighted automata with bounded languages}\label{sec:bounded}

In this section we consider the big-O problem for a weighted automaton $\wa$ and states $s,s'$ such that $\lng{s}{\wa}$, $\lng{s'}{\wa}$ are bounded.
Throughout the section, we assume that the LC condition has already been checked, i.e. $\lng{s}{\wa}\subseteq\lng{s'}{\wa}$.
We will show that the problem is conditionally decidable, subject to Schanuel's conjecture.
We give a quick introduction to this conjecture in \cref{sec:logicaltheoriesschan}.

\begin{thm}\label{thm:bounded}
Given a weighted automaton $\wa = \abra{Q,\Sigma,M,F}$, $s,s'\in Q$, with $\lng{s}{\wa}$ and $\lng{s'}{\wa}$ bounded, it is decidable whether $s$ is big-O of $s'$, subject to Schanuel's conjecture.
\end{thm}

Before delving into proof details, we illustrate the arising challenges using a simple representative example in \cref{sec:boundednotunary}. After that, the proof of \cref{thm:bounded} is spread over \cref{sec:plusletterbounded,sec:letterboundedcase,sec:boundedtoletterbounded}. First, a version of \cref{thm:bounded} restricted to  plus-letter-bounded languages is shown (\cref{lem:plbdecide} in \cref{sec:plusletterbounded}). Subsequently, letter-bounded languages are also captured by reduction to this case (\cref{lem:lbtoplb} in \cref{sec:letterboundedcase}). Finally, the case of bounded languages is reduced to the case of letter-bounded languages (\cref{lem:btolb}) in \cref{sec:boundedtoletterbounded}.

\subsection{Logical theories of arithmetic and Schanuel's conjecture}
\label{sec:logicaltheoriesschan}\mbox{}\\In \emph{first-order logical theories of arithmetic,}
variables denote numbers (from $\mathbb Z$ or $\mathbb R$, as appropriate),
and atomic predicates are equalities and inequalities between terms
built from variables and function symbols.
Nullary function symbols are constants, always from $\mathbb Z$.
If binary addition and multiplication are available, then:
\begin{itemize}
\item for $\mathbb R$
we obtain the first-order theory of the reals, where the truth value of
sentences is decidable due to the celebrated Tarski--Seidenberg theorem~\cite[Chapter~11 and Theorem~2.77]{BPR};
\item for $\mathbb Z$,
the first-order theory of the integers is, in contrast, undecidable
(see, e.g,~\cite{poonen2003hilbert}).
\end{itemize}
In the case of $\mathbb R$, adding the unary symbol for the exponential function
$x \mapsto e^x$ leads to
\emph{the first-order theory of the real numbers with exponential
function} (\rexp).
Logarithms base~2, for example, are easily expressible in~\rexp.
The decidability of~\rexp\ is an open problem
and hinges upon Schanuel's conjecture~\cite{macintyre1996decidability}.

\emph{Schanuel's conjecture}~\cite{lang1966introduction} is a unifying conjecture of
transcendental number theory, saying that for all $z_1,\dots,z_n \in \mathbb C$ linearly
independent over $\mathbb{Q}$ the field extension
$\mathbb{Q}(z_1,\dots,z_n,e^{z_1},\dots,e^{z_n})$ has  transcendence degree at
least $n$ over $\mathbb{Q}$,
meaning that for some $S \subseteq \{z_1, \dots, z_n,e^{z_1}, \dots, e^{z_n}\}$
of cardinality~$n$, say $S = \{s_1, \ldots, s_n\}$,
the only polynomial $p$ over $\mathbb Q$ satisfying
$p(s_1, \ldots, s_n) = 0$ is $p \equiv 0$. See, e.g., Waldschmidt's book~\cite[Section~1.4]{Waldschmidt00} for
  further context.
If indeed true, this conjecture would
generalise several known results, including the Lindemann--Weierstrass theorem
and Baker's theorem,
and would entail the decidability of~\rexp.
A recent exciting line of research
reduces problems from verification~\cite{DAVIAUD202178,Majumdar20}, linear dynamical systems~\cite{AlmagorCO018,ChonevOW16},
and symbolic computation~\cite{Huang18} to the decision problem for~\rexp.

\subsection{Difference to the unary case}
\label{sec:boundednotunary}

In the unary case, it was sufficient to consider the \textit{relative order} between spectral radii, with careful handling of the periodic behaviour.
This approach is insufficient in the bounded case. \cref{exa:relativeorderinsuff} highlights that
 the actual values of the spectral radii have to be examined.

\begin{exa}[Relative orderings are insufficient]
\label{exa:relativeorderinsuff}
\begin{figure}[t]
\centering
\includegraphics[width=0.7\linewidth]{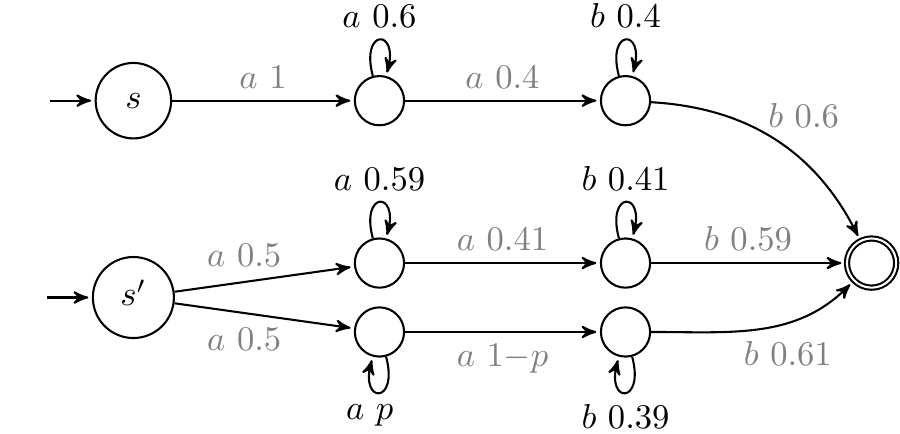}
\caption{Relative orderings are the same, but the big-O relations are different. }
\label{fig:relativeordernotsuff}
\end{figure}
Consider the LMC in \cref{fig:relativeordernotsuff}, with $0.61\le p \le 0.62$. We have $\nu_s(a^m b^n) =\Theta (0.6^m 0.4^n)$ and $\nu_{s'}(a^m b^n) =\Theta(p^m 0.39^n + 0.59^m 0.41^n)$. Note that neither $0.59^m 0.41^n$ nor $p^m 0.39^n$ dominate, nor are dominated by, $0.6^m 0.4^n$ for any value of $0.61\le p \le 0.62$. That is, there are values of $m,n$ where $0.59^m 0.41^n \gg 0.6^m 0.4^n$ (in particular large $n$) and values of $m,n$ where $0.59^m 0.41^n \ll 0.6^m 0.4^n$ (in particular large $m$); similarly for $p^m 0.39^n$ vs $0.6^m 0.4^n$ (but the cases in which $n$ or $m$ needs to be large are swapped). However, the big-O status can be different for different values of $ p \in[0.61,0.62]$, despite the same relative ordering between spectral radii. When $p=0.62$, the ratio  turns out to be bounded:
\[\frac{\nu_s(aa^ma b^{n}b)}{\nu_{s'}(aa^ma b^{n}b)}
=\frac{1\cdot  0.6^m \cdot 0.4\cdot 0.4^n\cdot 0.6}{
  0.5\cdot  0.59^m \cdot 0.41\cdot 0.41^n\cdot 0.59 + 0.5\cdot  0.62^m \cdot 0.38\cdot 0.39^n\cdot 0.61
}\le \frac{1600}{1579}
\]
for all $m,n$ (in particular, maximal at $m=n=0$, $\frac{\nu_s(aab)}{\nu_{s'}(aab)} = \frac{1600}{1579}$, for all larger $m,n$ the ratio is smaller).

In contrast, when $p=0.61$, we have $\frac{\nu_s(a^m b^{0.66m})}{\nu_{s'}(a^m b^{0.66m})}\trns[m\to\infty]{}\infty$.  To see that, observe there is an $x \in \mathbb{Q}$ such that  $0.61 \cdot 0.39^x < 0.6 \cdot 0.4^x$ and $0.59 \cdot 0.41^x < 0.6 \cdot 0.4^x$, e.g., $x = 0.66$. We can take $m = x n$ such that $m, n \in \mathbb{N}$ and $m, n \to \infty$. Whilst useful for illustration in this example, this effect is not limited to a linear relation between the characters, and so heavier machinery is required.
\end{exa}

We first prove \cref{thm:bounded} for the plus-letter-bounded case, which is the most technically involved; the other bounded cases will be reduced to it.
 In the plus-letter-bounded case, we will characterise the behaviour of such automata, generalising $(\rho,k)$-pairs of the unary case.
We will  need to rely upon the first-order theory of the reals with exponentials to compare these behaviours.

\subsection{The plus-letter-bounded case}
\label{sec:plusletterbounded}
We assume $\lng{s'}{\wa} \subseteq a_1^+\cdots a_m^+$, where $a_1, \cdots, a_n\in\Sigma$ and, because the LC condition holds, we also have  $\lng{s}{\wa} \subseteq a_1^+\cdots a_m^+$. In the plus-letter-bounded cases, without loss of generality, we assume $\chr_i\ne \chr_j$ for $i\neq j$.  Then any word $w = a_1^{n_1}\dots a_m^{n_m}$ is uniquely specified by a vector $(n_1,\dots,n_m)\in \mathbb{N}^m_{>0}$, where $n_i$ is the number of $a_i$'s in $w$.

\begin{prop}\label{claim:plbadditionaform}
In the plus-letter-bounded cases, without loss of generality, we may assume $\chr_i\ne \chr_j$ for $i\neq j$.
\end{prop}
\begin{proof}
We show how to reduce the big-O problem in the plus-letter-bounded case to the version of the same problem where $\chr_i\ne \chr_j$ (for $i\neq j$). Suppose, as above, $\lng{s'}{\wa} \subseteq a_1^+\cdots a_m^+$.
We can assume $a_i\neq a_{i+1}$ ($1\le i < m$) because $a^+ a^+$ can be replaced with $a^+$ (even though this results in a less restrictive ``bounding'' expression).
If $\wa$ had an $a$-labelled transition that can be used in two different blocks $a_i^+$ and $a_j^+$ ($j \ge i+2$) within $a_1^+\cdots a_m^+$, then $a = a_i = a_j$ and this transition could be used to ``skip'' the block $a_{i+1}^+$, i.e., there would exist a word $w \in \lng{s'}{\wa}$ that does not use the $a_{i+1}^+$ block, contradicting  $\lng{s'}{\wa} \subseteq a_1^+\cdots a_m^+$.
Therefore, every transition can be associated with exactly one block.
So we can define a fresh alphabet $\Sigma' = \{b_1, \ldots, b_m\}$ and relabel each transition associated with the $i$th block by~$b_i$.
\end{proof}

Our proof plan, for the plus-letter-bounded case (i.e., in the present \cref{sec:plusletterbounded}), will pass the following milestones:
\begin{enumerate}[1.]
\item Generalise the degree function $\dgf$ from \cref{def:unarydegreefunction} to non-unary words.
      (\Cref{def:boundedcompat} and \cref{def:boundeddegree} below.)
\item Generalise the big-$\Theta$ lemma (\cref{lemma:unarythetabound}), to lower- and upper-bound
      the values $\nu_s(a_1^{n_1}a_2^{n_2}\br\dots a_m^{n_m})$ by more convenient expressions,
      sometimes called ``exponential polynomials''.
      (\Cref{lemma:thetaletterbounded}.)
\item Characterise the negative instances of the big-O problem.
      (\Cref{lem:natlogicalformulationfirst}.)
\item Show that this characterisation is effective, assuming Schanuel's conjecture.
      (\Cref{lem:real}.)
\end{enumerate}
A comment on the third milestone is in order. In the unary case,
an effective characterisation of positive vis-a-vis negative instances of the big-O problem
(\cref{lem:boundedff}) can be obtained by ``mimicking'' the behaviour
of the original automaton, once the big-$\Theta$ lemma is in place.
As it turns out, the non-unary case
creates a more subtle, multivariate situation. Our new characterisation
(\Cref{lem:natlogicalformulationfirst}) will involve, in particular, the existence
of a sequence of vectors with some properties. Accordingly, new arguments will be required
(involving further automata and logic) to show algorithmic decidability.

We will now carry out the plan above.

Like in \cref{def:unarydegreefunction}, we define a degree function $\dgf$, which will be used to study
the asymptotic behaviour of words. This time we will associate a separate $(\rho,k)$ pair to each of the $m$ characters and, consequently,
words will induce sequences of the form $(\rho_1,k_1)\cdots(\rho_m,k_m)$.

Further, as there may be multiple, incomparable behaviours, words will induce sets of such sequences, i.e. $\dgf\colon \mathbb{N}^m \to \mathcal{P}( (\mathbb{R}\times \mathbb{N})^m)$.
For the sake of comparisons, it will be convenient to focus on maximal elements with respect to the pointwise order on $ (\mathbb{R} \times \mathbb{N})^m$, written $\le$,
where the lexicographic order (recall \cref{defn:lexiorder}) is used to compare elements of $\mathbb{R} \times \mathbb{N}$.

Recall from the discussion in \cref{rem:whynowork}, that our big-$\Theta$ lemma (\cref{lemma:unarythetabound}) does not capture the asymptotics when $n \le |Q|$ because it cannot distinguish between short paths with no loops and non-paths. In the unary case this is inconsequential as small words are covered by the \textit{finitely many} exceptions and the LC condition. However, here, a small number of one character may be used to enable access to a particular part of the automaton in another character. For this case, we introduce a new number $\delta= \frac{1}{2}\min_{\scc{} : \rho_\scc{} > 0} \rho_\scc{}$ which is strictly smaller than the spectral radius of every non-zero SCC (so will not dominate with the partial order), but non-zero. This allows us to distinguish between no-path using the $(0,0)$-pair and a short non-looping path using the $(\delta,0)$-pair.

\begin{defi}
\label{def:boundedcompat}
Let $\hat{\rho}= (\rho_1,k_1), \cdots, (\rho_m,k_m)\in (\mathbb{R} \times \mathbb{N})^m$.
An $a_1^{n_1}a_2^{n_2}\dots a_m^{n_m}$-labelled path from $s$ (to the final state) is \emph{compatible with} $\hat{\rho}$
if, for each $i = 1,\dots, m$, it visits $k_i + 1$ SCCs with spectral radius $\rho_i$ while reading $a_i$, unless the path visits only singletons with no loops, in which case $(\rho_i,k_i) = (\delta,0)$. The notation $(\rho,k)\in \hat{\rho}$ is used for ``$(\rho,k)$ is an element of  $\hat{\rho}$''.
\end{defi}
\begin{defi}
\label{def:boundeddegree}
Let $\dgf_{s} : \mathbb{N}^m \to \mathcal{P}((\mathbb{R} \times \mathbb{N})^m)$ be s.t.:
$\hat{\rho}\in \dgf_s(n_1,\dots,n_m) $ if and only if
\begin{enumerate}[(1)]
\item there exists an  $a_1^{n_1}a_2^{n_2}\dots a_m^{n_m}$-labelled path from $s$ to the final state compatible with $\hat{\rho}$, and
\item for every $a_1^{n_1}a_2^{n_2}\dots a_m^{n_m}$-labelled path from $s$ compatible with $\hat{\sigma}$ s.t.\ $\hat{\rho}\le\hat{\sigma}$, we have $\hat{\rho}=\hat{\sigma}$.
\end{enumerate}
\end{defi}

\noindent Observe that $\hat{\rho}$ may range over at most $|Q|^{2m}$ possible values. Let $\mathcal{D}$ be the set of all of them, so that $\dgf_{s}: \mathbb{N}^m \to \mathcal{P}(\mathcal{D})$.
In this extended setting, the big-$\Theta$ lemma (\cref{lemma:unarythetabound}) may be generalised as follows.
\begin{lem} \label{lemma:thetaletterbounded}
Denote $\z{n_1,\cdots,n_m} =  \sum_{\hat{\rho}\in \dgf_s(n_1,\dots,n_m)}\; \prod_{(\rho_i,k_i) \in \hat{\rho}} \ \rho_i^{n_i} \cdot n_i^{k_i}.$\\
There exist $c,C > 0$ such that for all $n_1\dots,n_m\in \mathbb{N}$:
\[
c\cdot \z{n_1,\cdots,n_m} \le  \nu_s(a_1^{n_1}a_2^{n_2}\dots a_m^{n_m}) \le C\cdot \z{n_1,\cdots,n_m}.
\]
\end{lem}

\begin{proof}

\begin{align*}
\nu_s(a_1^{n_1}a_2^{n_2}\dots a_m^{n_m}) &=  (M(a_1)^{n_1}\times M(a_2)^{n_2}\times \dots \times M(a_m)^{n_m})_{s,t}
\\ &= \sum_{q_1 \in Q} M(a_1)^{n_1}_{s,q_1}(\times M(a_2)^{n_2}\times \dots \times M(a_m)^{n_m})_{q_1,t}
\\ & \quad \vdots
\\ &= \sum_{(q_1,\dots,q_{m-1})\in Q^{m-1}} M(a_1)^{n_1}_{s,q_1}\times M(a_2)^{n_2}_{q_1,q_2}\times \dots \times M(a_m)^{n_m}_{q_{m-1},t}
\end{align*}
Let $s = q_0$ and $t = q_m$, then a sequence $q_0,\dots,q_m$ describes a possible path through the automaton.

By \cref{lemma:unarythetabound} in the unary case, for each $M(a_i)^{n_i}_{q_{i-1},q_i} >0$, there is a $(\rho_{q_{i-1},q_i},k_{q_{i-1},q_i}),c,C$, such that if $n_i > |Q|$ \begin{equation}\label{eq:boundingcc}c\rho_{q_{i-1},q_i}^{n_i}n_i^{k_{q_{i-1},q_i}} \le M(a_i)^{n_i}_{q_{i-1},q_i} \le C \rho_{q_{i-1},q_i}^{n_i}n_i^{k_{q_{i-1},q_i}}.\end{equation}

Otherwise if $n_i \le |Q|$, there exists $c,C$ satisfying \cref{eq:boundingcc} regardless of the value of $\dgf_{q_{i-1},q_i}(n_i) = (\rho,k)$, since there are at most $|Q|$ characters thus bounding $M(a_i)^{n_i}_{q_{i-1},q_i}$. In particular \cref{eq:boundingcc} holds for $(\rho,k) = (\delta,0)$, that is, if no path of length $n_i$ between $q_{i-1}$ and $q_i$ goes through a loop.

When $M(a_i)^{n_i}_{q_{i-1},q_i}$ is zero, $M(a_i)^{n_i}_{q_{i-1},q_i}$ correspond to the $(0,0)$-pair, and $M(a_i)^{n_i}_{q_{i-1},q_i} = 0^{n_i}n_i^{0} = 0$).

Take $c,C$ so that $C$ is maximised over all such $C$ and $c$ is minimised over all such $c$. We have  $
c^{m} \tilde{z}(n_1,\dots,n_m)  \le \nu_s(a_1^{n_1}a_2^{n_2}\dots a_m^{n_m}) \le  C^{m}\tilde{z}(n_1,\dots,n_m $, where
\[
\tilde{z}(n_1,\dots,n_m) = \sum_{(q_1,\dots,q_{m-1})\in Q^{m-1}} \rho_{s,q_1}^{n_1}n_1^{k_{s,q_1}}
\cdot \ldots \cdot \rho_{q_{m-1},t}^{n_m}n_m^{k_{q_{m-1},t}}.
\]

We observe that $\tilde{z}(n_1,\dots,n_m)$ is nearly $z(n_1,\dots,n_m)$.  First we omit paths which contain a $(0,0)$ element, then each summand $\rho_{s,q_1}^{n_1}n_1^{k_{s,q_1}}\cdot
\ldots \cdot \rho_{q_{m-1},t}^{n_m}n_m^{k_{q_{m-1},t}}$ of the sum corresponds with a candidate element $\hat{\rho}$ of $\dgf_s(n_1,\dots,n_m)$.
Given any two elements where $\hat{\rho}= \hat{\sigma}$ only one representative need be kept in $\dgf_s(n_1,\dots,n_m)$ by replacing $c$ and $C$ with $c/2$ and $2C$.

Similarly, we only need to keep maximal elements. By standard manipulations, given $\hat{\sigma}=(\sigma_1,\ell_1)\dots(\sigma_m,\ell_m),\hat{\rho}=(\rho_1,k_1)\dots(\rho_m,k_m)$,  such that $(\sigma_i,\ell_i) \le (\rho_i,k_i)$ for all $i\in\nat{m}$, then for some $c',C'$ we have
\[
 c' \prod_{(\rho_i,k_i) \in \hat{\rho}} \ \rho_i^{n_i} \; n_i^{k_i}
 \le \prod_{(\sigma_i,\ell_i) \in \hat{\sigma}} \ \sigma_i^{n_i} \; n_i^{\ell_i}
  + \prod_{(\rho_i,k_i) \in \hat{\rho}} \ \rho_i^{n_i} \; n_i^{k_i}
  \le C' \prod_{(\rho_i,k_i) \in \hat{\rho}} \ \rho_i^{n_i} \; n_i^{k_i}.\] and so  $\hat{\sigma}$ can be omitted from $\dgf_s(n_1,\dots,n_m)$ by replacing $c$ and $C$ with $cc'$ and $CC'$.
\end{proof}

The following lemma provides the key characterisation of negative instances of the big-O problem, in the plus-letter-bounded case and assuming the LC condition. Here and below, we write $n(t)$ to refer to the the $t$th vector in a sequence $n\colon \mathbb{N}\to\mathbb{N}^m$.

\begin{lem}[Main lemma]\label{lem:natlogicalformulationfirst}
Assume  $\lng{s}{\wa}\subseteq \lng{s'}{\wa}$.  Then
$s$ is not big-O of $s'$ if and only if there exists a sequence $n \colon \mathbb{N} \to \mathbb{N}^m$ and $X\in\adm$, $\mathcal{Y}\subseteq\adm$ such that
\begin{enumerate}[(a)]
\item \label{item:condition1} $X\in \dgf_s(n(t))$ and $\mathcal{Y}=\dgf_{s'}(n(t))$ for all $t$, and
\item \label{item:condition2} for all $j \in h_{\mathcal{Y}}$, the sequence $n$ satisfies
\[\sum_{i = 1}^m \alpha_{j,i} \,n(t)_i +  p_{j,i}\log n(t)_i \trns[t\to\infty]{} -\infty,\]
\end{enumerate}
where $h_{\mathcal{Y}} \subseteq \{1,\dots, |\mathcal{Y}|\}$, $\alpha_{j,i}\in \mathbb{R}$, $p_{j,i}\in \mathbb{Z}$ ($1\le i \le m$) are uniquely determined by $X$ and $\mathcal{Y}$ (in a way detailed below),  $h_{\mathcal{Y}}$  and $p_{j,i}$'s are effectively computable and $\alpha_{j,i}$'s are first-order expressible (with exponential function).

\end{lem}
\begin{proof}
Observe that then $s$ is \emph{not} big-O of $s'$ iff there exists an infinite sequence of words such that, for  all $C >0$,
the sequence contains a word $w$ such that $\frac{\nu_s(w)}{\nu_{s'}(w)} > C$.
Thanks to \cref{lemma:thetaletterbounded}, this is equivalent to the existence of  a sequence $n : \mathbb{N} \to \mathbb{N}^m$ such that \begin{equation*}
\frac{\displaystyle\sum_{X\in \dgf_{s}(n(t)_1,\dots, n(t)_m)} \; \prod_{(\rho_i,k_i) \in X} \  \rho_i^{n(t)_i}\; n(t)_i^{k_i}}{\displaystyle\sum_{Y\in \dgf_{s'}(n(t)_1,\dots, n(t)_m) \;}\prod_{(\sigma_i,\ell_i) \in Y} \  \sigma_i^{n(t)_i}\; n(t)_i^{\ell_i}} \trns[t\to\infty]{} \infty,
\end{equation*}
where $n(t)_i$ denotes  the $i$th component of $n(t)$.
Since there are finitely many possible values of $\dgf_s$ and $\dgf_{s'}$, it suffices to look for sequences $n$ such that $\dgf_s(n(t))$ and $\dgf_{s'}(n(t))$ are fixed.
Further, because of the sum in the numerator, only one $X \in \mathcal{X}$ is required such that $X \in \dgf_s(n_1,\dots,n_m)$.
Thus, we need to determine
whether there exist $X\in\adm$, $\mathcal{Y}\subseteq \adm$ and $n: \mathbb{N} \to \mathbb{N}^m$
such that $X \in \dgf_s(n(t))$, $\dgf_{s'}(n(t)) =\mathcal{Y}$ (for all $t$) and
\begin{equation*}
\frac{\prod_{i=1}^m \rho_i^{n(t)_i}\; n(t)_i^{k_i}}{\sum_{j=1}^{h_\mathcal{Y}} \prod_{i=1}^m \sigma_{ji}^{n(t)_i}\; n(t)_i^{\ell_{ji}}} \trns[t\to\infty]{} \infty.
\end{equation*}
where $X = (\rho_1, k_1)\cdots (\rho_m, k_m)$,
$\mathcal{Y} = \{ Y_1,\cdots, Y_{|\mathcal{Y}|} \}$, and $Y_j= (\sigma_{j1},\ell_{j1})\cdots(\sigma_{jm},\ell_{jm})$ ($1\le j\le |\mathcal{Y}|$).
Taking the reciprocal and requiring each of the summands to go to zero, we obtain
\[
\frac{\prod_{i=1}^m \sigma_{ji}^{n(t)_i}\; n(t)_i^{\ell_{ji}}}{\prod_{i=1}^m \rho_i^{n(t)_i} \; n(t)_i^{k_i}}= \prod_{i = 1}^m\left(\frac{\sigma_{ji}}{\rho_i}\right)^{{n(t)_i}} {n(t)_i}^{\ell_{ji} - k_i }  \trns[t\to\infty]{} 0
\qquad \text{for all $1\le j\le |\mathcal{Y}|$.}
\]
If we take  logarithms, letting $\alpha_{j,i} = \log(\frac{\sigma_{ji}}{\rho_i})$ and $p_{j,i} = \ell_{ji} - k_i$, we get
\[
\sum_{i = 1}^m \alpha_{j,i} \,n(t)_i +  p_{j,i}\log n(t)_i \trns[t\to\infty]{} -\infty\] for all $j $ in $h_\mathcal{Y}= \{ 1\le j \le |\mathcal{Y}| \mid  \sigma_{ji} >0 \text{ for all }  1\le i\le m \}.$

The number $\alpha_{j,i}$ is the logarithm of the ratio of two algebraic numbers, which are not given explicitly. However, they admit an unambiguous, first-order expressible characterisation (recall \cref{lem:algebraic}). The logarithm is encoded using the exponential function: $\log(z)$ is the $x \in \mathbb{R}$ such that $\exp(x) = z$.
\end{proof}

\cref{lem:natlogicalformulationfirst} identifies violation of the big-O property using two conditions. In the remainder of this subsection we will handle  Condition~\ref{item:condition1} using automata-theoretic tools (the Parikh theorem and semi-linear sets) and Condition~\ref{item:condition2} using logics. In summary, the characterisation of \cref{lem:natlogicalformulationfirst} will be expressed in the first-order theory of the reals with exponentiation, which is decidable subject to Schanuel's conjecture.

\paragraph*{Condition~\ref{item:condition1} via automata}

It turns out that sequences $n$ satisfying Condition~\ref{item:condition1} in \cref{lem:natlogicalformulationfirst} can be captured by a finite automaton.  In more detail, for any  $X\in \mathcal{D}$, there exists an automaton $\mathcal{N}^s_{X}$ such that
$\lng{}{\mathcal{N}^s_{X}} = \{ a_1^{n_1}\cdots a_m^{n_m} \ | \ X\in \dgf_s(n_1,\cdots,n_m) \}$. For any $\mathcal{Y}\subseteq \mathcal{D}$, there exists an automaton $\mathcal{N}^s_{\mathcal{Y}}$ such that
$\lng{}{\mathcal{N}^s_{\mathcal{Y}}} = \{ a_1^{n_1}\cdots a_m^{n_m} \ | \ \dgf_s(n_1,\cdots,n_m) = \mathcal{Y} \}$. The relevant automaton capturing $X$ and $\mathcal{Y}$ is then found by taking  the intersection of $\lng{}{\mathcal{N}^s_{{X}}}$ and $\lng{}{\mathcal{N}^{s'}_{\mathcal{Y}}}$.
\begin{lem} \label{lem:detectors}
For any $X\in\adm$ and $\mathcal{Y}\subseteq\adm$, there exists an automaton $\mathcal{N}_{X,\mathcal{Y}}$ such that
$\lng{}{\mathcal{N}_{X,\mathcal{Y}}} = \{ a_1^{n_1}\cdots a_m^{n_m}\ |\  X\in \dgf_s(n_1,\cdots,n_m),\, \mathcal{Y} = \dgf_{s'}(n_1,\cdots, n_m) \}$.
\end{lem}

\begin{proof}
Let $\hat{\rho}=(\rho_1,k_1)\cdots (\rho_m,k_m)$.
Building on \cref{defn:annotated},
one can construct an automaton $\mathcal{N}^s_{{\ge \hat{\rho}}}$ with
\[
  \lng{}{\mathcal{N}^s_{{\ge \hat{\rho}}}} = \{a_1^{n_1} \dots a_m^{n_m}\ |\ { \exists\hat{\sigma} \in \dgf_s(n_1,\dots,n_m)\,\, \hat{\sigma}\ge \hat{\rho} }\}
\]
by tracking the current maximum spectral radius seen and the number of different SCCs with this spectral radius. If the only states seen so far have been singletons with no loops (formally having spectral radius $0$), the value should be tracked as $(\delta,0)$ regardless of how many have been seen.

Passage from states reading $a_j$ to states reading $a_{j+1}$ is allowed only if the tracked value is at least $(\rho_j, k_j)$,
and states should be final if the tracked value of $a_m$ is at least $(\rho_m, k_m)$.
As previously, comparisons are with respect to the partial order of \cref{defn:lexiorder}.

Similarly, one can construct $\mathcal{N}^s_{{>\hat{\rho}}}$ with
\[
  \lng{}{\mathcal{N}^s_{{>\hat{\rho}}}} = \{a_1^{n_1}\dots a_m^{n_m}\ |\ { \exists\hat{\sigma} \in \dgf_s(n_1,\dots,n_m)\,\, \hat{\sigma} > \hat{\rho} }\}.
\]
The construction is the same as for $\mathcal{N}^s_{{\ge \hat{\rho}}}$ except that, in order to accept, we need to be sure that at least
one of the `at least' comparisons was strict. This can be achieved by maintaining an extra bit at run time.

Note that $\lng{}{\mathcal{N}^s_{{\ge \hat{\rho}}}} \setminus \mathcal{N}^s_{{>\hat{\rho}}}$ contains  all $a_1^{n_1} \dots a_m^{n_m}$ such that there exists $\hat{\sigma}\in \dgf_s(n_1,\cdots,n_m)$ with
$\hat{\sigma}\ge \hat{\rho}$ and, for all $\hat{\tau}\in \dgf_s(n_1,\cdots,n_m)$, we do \emph{not} have $\hat{\tau}>\hat{\rho}$. Consequently, we must have $\hat{\rho}\in \dgf_s(n_1,\cdots,n_m)$, which
implies (by maximality) that we cannot have $\hat{\tau} > \hat{\rho}$ for any $\hat{\sigma}\in   \dgf_s(n_1,\cdots,n_m)$. Hence,
\[
\lng{}{\mathcal{N}^s_{{\ge \hat{\rho}}}} \setminus \lng{}{\mathcal{N}^s_{{>\hat{\rho}}}} = \{ a_1^{n_1}\dots a_m^{n_m} \ |\ \hat{\rho}\in \dgf_s(n_1,\cdots,n_m) \}.
\]
Denote by $\mathcal{N}^s_{{\hat{\rho}}}$ the corresponding automaton;
we can take $\mathcal{N}^s_{X} = \mathcal{N}^s_{{\hat{\rho}}}$.

Given $\mathcal{Y}\subseteq \mathcal{D}$, we can then take $\mathcal{N}^{s'}_{\mathcal{Y}}$ to be the automaton corresponding to
\[
\bigcap_{\hat{\rho}\in\mathcal{Y}} \lng{}{\mathcal{N}^{s'}_{{\hat{\rho}}}} \cap \bigcap_{\hat{\rho}\in \mathcal{D}\setminus\mathcal{Y}} \pbra{a_1^+\cdots a_m^+ \setminus \lng{}{\mathcal{N}^{s'}_{{\hat{\rho}}}}}.
\]
The relevant automaton $\mathcal{N}_{X,\mathcal{Y}}$   is then found by taking  the intersection of $\lng{}{\mathcal{N}^s_{{X}}}$ and $\lng{}{\mathcal{N}^{s'}_{\mathcal{Y}}}$.\qedhere
\end{proof}

Because of our $a_i\neq a_{j}$ assumption, the vector $(n_1,\cdots,n_m)$ indicates the number of occurrences of each character.
The set of such vectors derived from the language of an automaton  is known  as the \emph{Parikh image} of this language~\cite{parikh1966context}.
It is  well known that the Parikh image of every regular language is a semi-linear set, i.e. a finite union of linear sets
(a linear set  has the form $\{\basevec+\lambda_1\vec{r}^1 + \dots + \lambda_s \vec{r}^s\ |\   \lambda_1,\dots,\lambda_s \in \mathbb{N} \}$,
where $\basevec \in \mathbb{N}^m$ is the base vector and $\vec{r}^1,\cdots \vec{r}^s\in\mathbb{N}^m$ are called period vectors).
However, since $\lng{}{\mathcal{N}_{X,\mathcal{Y}}} \subseteq \chr_1^+\chr_2^+\dots \chr_m^+$,
the linear sets are of a very particular form,  where each $\vec{r}^i$ is a constant multiple of the $i$th unit vector.

\begin{lem}\label{claim:finiteunion}
The Parikh image of every plus-letter-bounded regular language
is a finite union of linear sets, where each period vector is a constant
multiple of some unit vector.

In particular,
the language of $\mathcal{N}_{X,\mathcal{Y}}$ can be effectively decomposed as $\lng{}{\mathcal{N}_{X,\mathcal{Y}}} = \bigcup\nolimits_{k = 1}^{S_{X,\mathcal{Y}}}\mathcal{L}_k$,
where
$\mathcal{L}_k = \left\{ \chr_1^{\basevecraw_{k1}+ \pvecraw_{k1}\lambda_1}\cdots \chr_m^{\basevecraw_{km}+\pvecraw_{km}\lambda_m } \ | \ \lambda_1,\cdots,\lambda_m\in\mathbb{N}\right\}$, $S_{X,\mathcal{Y}}\in\mathbb{N}$ and  $b_{ki},r_{ki}\in\mathbb{N}$  $(1\le k \le S_{X,\mathcal{Y}}$, $1\le i \le m)$.
\end{lem}

\begin{proof}
Consider the machine $\mathcal{N}_{X,\mathcal{Y}}$, accepting a language which is a subset of $\chr_1^+\chr_2^+\dots \chr_m^+$, with any state not reachable from the starting state or not leading to an accepting state removed. To induce a form with the property we want, we intersect $\mathcal{N}_{X,\mathcal{Y}}$ with the standard DFA\footnote{By DFA we permit a partial transition function, that is 0 or 1 transition for each character from every state, rather than exactly 1.} for $\chr_1^+\chr_2^+\dots \chr_m^+$, without changing the language.

Hence every state corresponds to reading from exactly one character block of $\chr_1,\chr_2,\dots,\chr_m$. At each state there can be at most two characters enabled, either the character to remain in the current character block, or the character to move to the next. Every state can be labelled as \begin{itemize} \item  \textit{only having transition for $\chr_{i}$}; or \item \textit{also having transition with $\chr_{i+1}$}.\end{itemize} Consider all possible choices of automaton formed by restricting $\mathcal{N}_{X,\mathcal{Y}}$ so that there is a single state which is allowed to transition from $\chr_{i}$ to $\chr_{i+1}$ for each $i$ and any other state which had this property in $\mathcal{N}_{X,\mathcal{Y}}$ has its $\chr_{i+1}$ transitions removed (but keeps its $\chr_i$ transitions). Each such choice corresponds with a partition of the accepting runs of $\mathcal{N}_{X,\mathcal{Y}}$.

Thus $\lng{}{\mathcal{N}_{X,\mathcal{Y}}}$ is the finite union over the languages induced by all such machines. We further show that such machines can further be expressed as a finite union of  linear sets in the form prescribed.

Let us assume $\mathcal{N}_{X,\mathcal{Y}}^j$ is such a machine with a single state capable of transitioning from $\chr_{i}$ to $\chr_{i+1}$ for each $i$, and again remove any state not reachable from the starting state or not leading to an accepting state. The part of the machine reading $a_i$ has a single starting state and a single final state, which is a unary NFA when the transitions to $\chr_{i+1}$ are discarded.

This unary NFA can be converted to Chrobak normal form; the section of  $\mathcal{N}_{X,\mathcal{Y}}^j$ corresponding to $a_i$ can be replaced with this unary NFA, and any accepting state has additionally the transitions for transitioning from $\chr_{i}$ to $\chr_{i+1}$ of the single such state in $\mathcal{N}_{X,\mathcal{Y}}^j$.

Let us repeat the process above for all $i$, decomposing $\mathcal{N}_{X,\mathcal{Y}}^j$ into the subsets of languages where there are exactly one state transitioning from $\chr_{i}$ to $\chr_{i+1}$. Let $\mathcal{N}_{X,\mathcal{Y}}^j = \bigcup_k \mathcal{N}_{X,\mathcal{Y}}^{j,k}$, a finite union; where each $k$ corresponds to a selection of accepting states $(q_1,\dots,q_m)$ with $q_l$ being the accepting state in the Chrobak normal form for $\chr_l$.

Consider such an $\mathcal{N}_{X,\mathcal{Y}}^{j,k}$. The steps spent in each block corresponding to $\chr_i$ is either formed by the finite path or the a single cycle at the end of the path. If the transition occurs in the finite path then $\basevecraw_{ki}$ is the length of the path to that transition and $\pvecraw_{ki}$ is zero. If the transition occurs in the cycle at the end of the path, then $\basevecraw_{ki}$ is the length of the path to that transition from the start of the path and $\pvecraw_{ki}$ is the length of the cycle. In $\mathcal{N}_{X,\mathcal{Y}}^{j,k}$ the time spent in block $\chr_i$ has no influence on the time spent in $\chr_j$ for $j\ne i$. Then $\lng{}{\mathcal{N}_{X,\mathcal{Y}}^{j,k}} = \{\chr^{n_1}\chr^{n_2}\dots \chr^{n_m} |  \exists \vec{\lambda}\in\mathbb{N}^m \text{ s.t. }\forall i\in[m] n_i = \basevecraw_{ki} + \pvecraw_{ki}\cdot \lambda_i \}$. The language $\lng{}{\mathcal{N}_{X,\mathcal{Y}}}$ is the union over all $\lng{}{\mathcal{N}_{X,\mathcal{Y}}^{j,k}}$.
\end{proof}

\cref{claim:finiteunion} captures Condition~\ref{item:condition1} of \cref{lem:natlogicalformulationfirst} precisely.
\paragraph*{Condition~\ref{item:condition2} via logic}
With \cref{claim:finiteunion} in place, we now move on to add   Condition~\ref{item:condition2} to the existing machinery. In fact, the logical formulae in the following lemmas will express the conjunction of   both conditions of \cref{lem:natlogicalformulationfirst}.

\begin{lem}\label{lem:natlogicalformulation}
Assume  $\lng{s}{\wa}\subseteq \lng{s'}{\wa}$. Then
$s$ is not big-O of $s'$ if and only if there exists  $X\in\adm$, $\mathcal{Y}\subseteq\adm$, $1\le k\le S_{X,\mathcal{Y}}$ such that
\
\[
\forall C < 0\  \exists \vec{\lambda} \in \mathbb{N}^m \\
\bigwedge_{j\in h_\mathcal{Y}}\ \sum_{i=1}^m \alpha_{j,i} \,(\basevecraw_{ki} + \pvecraw_{ki} \, \lambda_i) +  p_{j,i}\log(\basevecraw_{ki} + \pvecraw_{ki} \,\lambda_i) < C,
\]
where $h_\mathcal{Y}, \alpha_{j,i}, p_{j,i}$ (resp.  $\basevecraw_{ki},  \pvecraw_{ki}$) satisfy the same conditions as in~\cref{lem:natlogicalformulationfirst} (resp.~\ref{claim:finiteunion}).
\end{lem}

Note that the formula of \cref{lem:natlogicalformulation} uses quantification over natural numbers. Our next step will be to replace integer variables with real variables. In other words, we will obtain an equivalent condition in the first-order theory of the reals with exponentiation, as follows.
\begin{lem}\label{lem:real}
Assume  $\lng{s}{\wa}\subseteq \lng{s'}{\wa}$. Then
$s$ is not big-O of $s'$ if and only if there exist $X\in\adm$, $\mathcal{Y}\subseteq\adm$, $1\le k\le S_{X,\mathcal{Y}}$ and $U\subseteq  \{i \in \{1,\cdots,m\} \ | \ \pvecraw_{ki} > 0 \}$ such that
\[
\forall C < 0 \  \exists \revec \in \mathbb{R}^{|U|}_{\ge B_k}\\
\bigwedge_{j\in h_\mathcal{Y}}\   \sum_{i \in U} \alpha_{j,i}\, \pvecraw_{ki}\, \revar_i + p_{j,i}\log(\revar_i) < C,
\]
where $B_k=\max_i{\basevecraw_{ki}}$ and $h_\mathcal{Y}, \alpha_{j,i}, p_{j,i}, \basevecraw_{ki},  \pvecraw_{ki}$ are as in~\cref{lem:natlogicalformulation}.
\end{lem}
\begin{proof}[Proof idea]
Compare the logical characterisation in \cref{lem:real} and \cref{lem:natlogicalformulation}. The first difference to note is that the effect of $\basevecraw_{ki}$'s is simply a constant offset, and so the sequence would tend to $-\infty$ with or without its presence. Similar simplifications can be made inside the logarithm: the multiplicative effect of $\pvecraw_{ki}$ inside the logarithm can be extracted as an additive offset and thus similarly be discarded.

The second crucial difference is to relax the variable domains from integers to  reals. If each of the $\lambda_i$ in the satisfying assignment is sufficiently large, we show we  can relax the condition to real numbers rather than integers without affecting whether the sequence goes to $-\infty$. To do this, we test sets of indices $U$, where if $i\in U$ then $\lambda_i$ needs to be arbitrarily large over all $C$ (i.e. unbounded). The positions where $\lambda_i$ is always bounded are again a constant offset and are omitted.
\end{proof}
\begin{proof}[Proof of \cref{lem:real}]

First we argue that we can restrict to some subset of the components which enable the satisfying choice of $\lambda$ to be sufficiently large in all components.
\begin{clm} \label{claim:referencenats}
The assertion
\begin{equation}\label{eqn:claim:referencenats}
\forall C\  \exists \vec{\lambda} \in \mathbb{N}^m \\
\bigwedge_{j \in h_\mathcal{Y}}\ \sum_{i=1}^m \alpha_{j,i}\, (\basevecraw_{ki} + \pvecraw_{ki} \, \lambda_i) +  p_{j,i}\log(\basevecraw_{ki} + \pvecraw_{ki}\, \lambda_i) < C
\end{equation}
 holds if and only if the following assertion holds for some $U \subseteq [m]$:\begin{equation}\label{eqn:claim:referenceboundednat}
\forall C\  \exists \vec{\lambda} \in \mathbb{N}^{U}_{\ge \max_i{\basevecraw_{ki}}}
\bigwedge_{j \in h_\mathcal{Y}}\ \sum_{i \in U} \alpha_{j,i}\, (\basevecraw_{ki} + \pvecraw_{ki}\, \lambda_i) +  p_{j,i}\log(\basevecraw_{ki} + \pvecraw_{ki}\, \lambda_i) < C.
\end{equation}
\end{clm}
\begin{subproof}[Proof of \cref{claim:referencenats}]

First note that \cref{eqn:claim:referenceboundednat} immediately implies \cref{eqn:claim:referencenats}. We show the converse.

Recall we can alternatively characterise the formulation as a sequence $n: \mathbb{N} \to \mathbb{N}^m$. That is, for each negative integer $C$, the choice of  $\vec{\lambda}$ corresponds to $n(C)$ in the sequence.

Note that in the sequence $n$ some components may be bounded. Either because $\pvecraw_{ki}=0$, or the choice of $n$ makes it so. Suppose there exists a $\theta >0$ such that $n(t)_x \le \theta$ for some $x \in [m]$, then  $\sum_{i = 1}^m \alpha_{j,i}\, n(t)_i +  p_{j,i}\log(n(t)_i) \le \sum_{i = 1, i\ne x}^m \alpha_{j,i}\, n(t)_i +  p_{j,i}\log(n(t)_i) + |\alpha_{j,i}|\, \theta +  |p_{j,i}|\theta$. Hence the sequence  $\sum_{i = 1, i\ne x}^m \alpha_{j,i}\, n(t)_i +  p_{j,i}\log(n(t)_i)$ goes to $-\infty$ as well.

Consider each choice of components $B \subseteq [m]$ which will be bounded. For some components there will be no choice as $\pvecraw_{ki} = 0$. Let us assume that the chosen set is maximal with respect to set-inclusion; that is, there should be no subsequence maintaining the property with fewer components unbounded. Let the remaining unbounded components be $U = [m] \setminus B$.

Since each remaining component is not bounded, there is always a later point in the sequence in which the value is larger; thus one can take a subsequence of $n(t)$ so that $n(t)_i \le n(t+1)_i$ for every $t$. Repeat for every remaining component $i \in U$; this can be done as  the minimal choice of unbounded components has been selected. Hence, without loss of generality if there exists some sequence, then for any $\theta$, there exists a subsequence of $n(t)$, such that $n(t)_i > \theta$ for all $i \in U$. To enable a more succinct analysis later, restrict $n(t)$ to those in which $\lambda_i \ge \max_i{\basevecraw_{ki}}$ where $n(t)_i = \basevecraw_{ki} + \pvecraw_{ki} \, \lambda_i$ for some $\lambda_i$.\end{subproof}

Next we argue that the offset component $\basevec$ does not affect whether the formula holds and that  we can relax the restriction of $\vec{\lambda}$ from naturals to positive reals and maintain the satisfiability of the formula.  The advantage here is that this relaxation can be solved with the first order theory of the reals with exponential function; which is decidable subject to Schanuel's conjecture.

\begin{clm}\label{claim:equalnaturalreal}The assertion
\begin{equation}\label{equn:naturalversion}
\forall C\  \exists \vec{\lambda} \in \mathbb{N}^{U}_{\ge \max_i{\basevecraw_{ki}}}
\bigwedge_{j \in h_\mathcal{Y}}\ \sum_{i \in U} \alpha_{j,i}\, (\basevecraw_{ki} + \pvecraw_{ki}\, \lambda_i) +  p_{j,i}\log(\basevecraw_{ki} + \pvecraw_{ki}\, \lambda_i) < C
\end{equation}
holds if and only if  the following assertion holds:
\begin{equation}\label{equn:realsversion}
\forall C\  \exists \revec \in \mathbb{R}^{U}_{\ge \max_i{\basevecraw_{ki}}}
\bigwedge_{j \in h_\mathcal{Y}}\ \sum_{i \in U} \alpha_{j,i}\, \pvecraw_{ki}\, \revar_i + \sum_{i \in U}  p_{j,i}\log(\revar_i) < C.
\end{equation}
\end{clm}

\begin{subproof}[Proof of \cref{claim:equalnaturalreal}]
Observe that \[\sum_{i \in U} \alpha_{j,i}\, (\basevecraw_{ki} + \pvecraw_{ki}\, \lambda_i) = \sum_{i \in U} \alpha_{j,i}\, \basevecraw_{ki} + \sum_{i \in U} \alpha_{j,i}\, \pvecraw_{ki}\, \lambda_i \] and that $\sum_{i \in U} \alpha_{j,i}\, \basevecraw_{ki}$ is constant so it does not affect whether the sequence goes to $-\infty$, hence \cref{equn:naturalversion}  holds if and only if
\begin{equation}
\forall C\  \exists \vec{\lambda} \in \mathbb{N}^{U}_{\ge \max_i{\basevecraw_{ki}}}
\bigwedge_{j \in h_\mathcal{Y}}\ \sum_{i \in U} \alpha_{j,i}\, \pvecraw_{ki}\, \lambda_i +  p_{j,i}\log(\basevecraw_{ki} + \pvecraw_{ki}\, \lambda_i) < C \end{equation}

Now let us extract the log component by using the following rewriting \[\log(\basevecraw_{ki} + \pvecraw_{ki}\, \lambda_i) = 
\log\pbra{\lambda_i\,\pbra{\frac{\basevecraw_{ki}}{\lambda_i} + \pvecraw_{ki}}} = 
\log(\lambda_i) + \log\pbra{\frac{\basevecraw_{ki}}{\lambda_i} + \pvecraw_{ki}}.\]

Since $\pvecraw_{ki} \ge 1$ and $\lambda_i \ge \basevecraw_{ki}$ we have $\log\pbra{\frac{\basevecraw_{ki}}{\lambda_i} + \pvecraw_{ki}}  \le \log(\pvecraw_{ki} + 1)$, which is constant. Hence \cref{equn:naturalversion} is equivalent to:
\begin{equation}\label{equn:reducednatural}
\forall C'\  \exists \vec{\lambda} \in \mathbb{N}^{U}_{\ge \max_i{\basevecraw_{ki}}}
\bigwedge_{j \in h_\mathcal{Y}}\ \sum_{i \in U} \alpha_{j,i}\, \pvecraw_{ki}\, \lambda_i + \sum_{i \in U}  p_{j,i}\log(\lambda_i) < C'
\end{equation}

We now show that this is equivalent to \cref{equn:realsversion}. Clearly  \cref{equn:reducednatural} implies \cref{equn:realsversion}. Now consider  \cref{equn:realsversion} holding, and we show the \cref{equn:reducednatural} is satisfied, by exhibiting a choice of $ \vec{\lambda} \in \mathbb{N}^{U}_{\ge \max_i{\basevecraw_{ki}}}$ for every $C'$.
Given $C' < 0 $, let $C = C' - \max_j\sum_{i \in U} |\alpha_{j,i}|\pvecraw_{ki} - \max_j\sum_{i \in U} |p_{j,i}|$, and choose $\revec \in \mathbb{R}^{|U|}_{\ge \max_i{\basevecraw_{ki}}}$ satisfying \cref{equn:realsversion}.
Now let $\revar_i = \intvar_i + y_i$, with $y_i  <  1, \intvar_i = \lfloor \revar_i \rfloor$. First observe that since $\revar_i \ge \max_i \basevecraw_{ki}$, an integer, also $\intvar_i \ge \max_i \basevecraw_{ki}$.
Observe that ${\abs{\sum_{i \in U} \alpha_{j,i}\, \pvecraw_{ki}\,  y_i } }  \le \sum_{i \in U} |\alpha_{j,i}|\pvecraw_{ki}$. Since
\[ \sum_{i \in U} \alpha_{j,i}\, \pvecraw_{ki}\, \intvar_i + {\sum_{i \in U} \alpha_{j,i}\, \pvecraw_{ki}\,  y_i} + \sum_{i \in U}  p_{j,i}\log(\intvar_i + y_i) < C
\]
we have
\[ \sum_{i \in U} \alpha_{j,i}\, \pvecraw_{ki}\, \intvar_i  + \sum_{i \in U}  p_{j,i}\log(\intvar_i + y_i) < C  + \sum_{i \in U} |\alpha_{j,i}|\pvecraw_{ki}
\]
Let us again rewrite $\log(\intvar_i + y_i) = \log\pbra{\intvar_i\pbra{1 + \frac{y_i}{\intvar_i}}}= \log(\intvar_i) + \log\pbra{1 + \frac{y_i}{\intvar_i}}$. Then since $\intvar_i > y_i$, $\log\pbra{1 + \frac{y_i}{\intvar_i}} \le 1$, so \[{\abs{\sum_{i \in U}  p_{j,i}\log\pbra{1 + \frac{y_i}{\intvar_i}} }} \le \sum_{i \in U}  |p_{j,i}|.\] 
We thus have
\[ \sum_{i \in U} \alpha_{j,i}\, \pvecraw_{ki}\, \intvar_i  + \sum_{i \in U}  p_{j,i}\log(\intvar_i)< C  + \sum_{i \in U} |\alpha_{j,i}|\pvecraw_{ki} + \sum_{i \in U} |p_{j,i}| \le C'
\]
and hence, \cref{equn:reducednatural} holds.\end{subproof}

Together \cref{claim:equalnaturalreal} and \cref{claim:referencenats} reformulate the logical condition of \cref{lem:natlogicalformulation} concluding the proof of \cref{lem:real}.
\end{proof}

By testing the LC condition and the condition from~\cref{lem:real}  for each possible $X,\mathcal{Y}, k,U$, in turn using the (conditionally decidable) first-order theory of the reals with exponential function, we have:
\begin{lem}\label{lem:plbdecide}
Given a weighted automaton $\wa$ and states $s,s'$ such that $\lng{s}{\wa}$ and $\lng{s'}{\wa}$ are  plus-letter-bounded, it is decidable whether $s$ is big-O $s'$, subject to Schanuel's conjecture.
\end{lem}

\subsection{The letter-bounded case}\label{sec:letterboundedcase}

Here we consider the case where $\lng{s}{\wa}$ and $\lng{s'}{\wa}$ are letter-bounded,
$\lng{s}{\wa}$ and $\lng{s'}{\wa}$ are subsets of $a_1^*\dots a_m^*$ for some $a_1,\dots,a_m\in\Sigma$,
which is a relaxation of the preceding case.
For the plus-letter-bounded case, we relied on a 1-1 correspondence between numeric vectors and words.
This correspondence no longer holds in the letter-bounded case: for example, $a^n$ matches  $a^* b^* a^*$, but it could correspond to $(n,0,0)$, $(0,0,n)$, as well as any $(n_1,0,n_2)$ with $n_1+n_2=n$.
Still, there is a reduction to the plus-letter-bounded case.

\begin{lem}\label{lem:lbtoplb}
The big-O problem for $\wa, s, s'$ with $\lng{s}{\wa}$ and $\lng{s'}{\wa}$ letter-bounded reduces to the plus-letter-bounded case.
\end{lem}

\begin{proof}
Suppose the LC condition holds and  $\lng{s}{\wa}\subseteq \lng{s'}{\wa} \subseteq a_1^\ast\cdots a_m^\ast$.
Let $I$ be the set of strictly increasing sequences $\vec{\imath} = i_1\cdots i_k$ of integers between $1$ and $m$.
Given $ \vec{\imath}\in I$, let $\wa_{\vec{\imath}}$ be the weighted automaton obtained by intersecting $\wa$ with a DFA for $a_{i_1}^+\cdots a_{i_k}^+$ whose initial state is $q$.
Note that $s$ is big-O of $s'$ (in $\wa$) iff $(s,q)$ is big-O of $(s',q)$ in $\wa_{\vec{\imath}}$ for all $\vec{\imath}\in I$, because $a_1^\ast\cdots a_m^\ast = \bigcup_{\vec{\imath}\in I} a_{i_1}^+\cdots a_{i_k}^+$.
Because the big-O problem for each $\wa_{\vec{\imath}}$, $(s,q)$, $(s',q)$ falls into the plus-letter-bounded case, the results follows from~\cref{lem:plbdecide}.
\end{proof}

\subsection{The bounded case}
\label{sec:boundedtoletterbounded}

Here we consider the case where $\lng{s}{\wa}$ and $\lng{s'}{\wa}$ are bounded, which is a relaxation of letter-boundedness (see \cref{def:bounded}): $\lng{s}{\wa}$ and $\lng{s'}{\wa}$ are subsets of $w_1^*\dots w_m^*$ for some $w_1,\dots,w_m\in\Sigma^*$. We show a reduction to the letter-bounded case from \cref{sec:letterboundedcase}.

\begin{lem}\label{lem:btolb}
The big-O problem for $\wa, s, s'$ with $\lng{s}{\wa}$ and $\lng{s'}{\wa}$ bounded reduces to the letter-bounded case.
\end{lem}

Suppose $\wa$ is bounded over $w_1^*\dots w_m^*$, our approach is to construct a new weighted automaton $\wa'$ letter-bounded over a new alphabet $a_1^*\dots a_m^*$  with the following property. For \textit{every} decomposition of a word $w$, as $w_1^{n_1}\dots w_m^{n_m}$, the weight of $a_1^{n_1}\dots a_m^{n_m}$ in $\wa'$ is equal to the weight of $w$ in $\wa$.

\begin{sloppypar}To showcase the difference to the letter-bounded case, consider the language $(abab)^*a^*b^*(ab)^*$. Observe that, for example the word $(ab)^4$ can be decomposed in a number of ways: $(abab)^2a^0b^0(ab)^0$, $(abab)^1a^1b^1(ab)^1$, $(abab)^1a^0b^0(ab)^2$, $(abab)^0a^1b^1(ab)^3$ or $(abab)^0a^0b^0(ab)^4$.
One must be careful to consider all such decompositions.\end{sloppypar}

\begin{proof}[Proof of \cref{lem:btolb}]
Let $\wa = \abra{Q,\Sigma,M,F}$. Then we have $w_1,\dots,w_m$ such that for all $w$  with $\nu_s(w) > 0$, $w = w_1^{n_1}\dots w_m^{n_m}$ for some $n_1,\dots,n_m\in \mathbb{N}$. Let us assume $w_i = b_{i,1}b_{i,2},\dots,b_{i,|w_i|}$.

Given a word $w$, there may be multiple paths $\pi_1,\pi_2,\dots$ from $s$ to $t$ respecting that word. Further there may be multiple decomposition vectors $\vec{n}_1,\vec{n}_2,\dots \in \mathbb{N}^m$ such that $\vec{n}_i = (n_1,\dots,n_m)$ and  $w = w_1^{n_1}\dots w_m^{n_m}$. Our goal will be to construct a weighted automaton $\wa'$ with states $\hat{s}$ and $\hat{s'}$ letter-bounded over $a_1^*\dots a_m^*$ such that, for every word $w$,
the weight of $a_1^{n_1}\dots a_m^{n_m}$ in $\wa'$ from $\hat{s}$ (resp. $\hat{s'}$), for every valid decomposition vector $\vec{n} \in \mathbb{N}^m$ of $w$,
will be the sum of the weights of all paths $\pi_1,\pi_2,\dots$ respecting $w$ in $\wa$  from $s$ (resp. $s'$).
To compute $\wa'$, we will define a transducer and apply it to our automaton $\wa$.

A nondeterministic finite transducer is an NFA with transitions labelled by pairs from  $\Sigma \times (\Sigma' \cup \{\varepsilon\})$, denoted by $a/b$  for $a\in\Sigma$ and $b\in\Sigma' \cup \{\varepsilon\}$. In our construction, we only require edges of this form,
i.e. we do not consider a definition with transitions labelled with $\varepsilon $ in the first component (e.g., $\varepsilon/a$).  Our transducer induces a translation $\mathcal{T}: \Sigma^* \to \Sigma'^*$.

Consider the set of regular expressions $w_{i_1}^+\dots w_{i_{m'}}^+$ each induced by a sequence $\vec{\imath}=(i_1,\dots,i_{m'}) \in \mathbb{N}^{m'}$,  $m' \le m $, with $1  \le i_1 < \dots < i_{m'} \le m$. Note that two sequences $(i'_1,\dots,i'_{m'})$, $(i''_1,\dots,i''_{m''})$ may yield the same expression $w_{i_1}^+\dots w_{i_m}^+$, in which case we need not consider more than one. The transducer $\mathcal{T}$ will be defined as follows.

For each $\vec{\imath} = (i_1,\dots,i_{m'})$  described above, build the following automaton. For each $i_j$, construct the following section, which simply reads the word $w_{i_j}$:
\[ f^{\vec{\imath}}_j\trns{b_{i_j,1}/\varepsilon} s^{\vec{\imath}}_j\trns{b_{i_j,2 }/\varepsilon} \cdot \trns{b_{i_j,3 }/\varepsilon} \cdot  \ldots \cdot \trns{b_{i_j,|w_i| -1}/\varepsilon} e^{\vec{\imath}}_j.\]
Then, on the final character, nondeterministically restart or move to the next word, emitting a character representing the word:
\[e^{\vec{\imath}}_j \trns{b_{i_j, |w_{i_j}|}/a_{i_j}} f^{\vec{\imath}}_j \quad \text{ and }\quad  e^{\vec{\imath}}_j \trns{b_{i_j, |w_{i_j}|}/a_{i_j}} f^{\vec{\imath}}_{j+1}.\]

The transducer $\mathcal{T}$ is defined by the union of the above transitions over all $\vec{\imath}$.
We also add a global start state $q_0$, from which we would like to  move nondeterministically to $f_1^{\vec{\imath}}$ for each ${\vec{\imath}}$.
To achieve this and avoid $\varepsilon$ transitions,  we duplicate the transitions $f^{{\vec{\imath}}}_1 \trns{x} s^{{\vec{\imath}}}_1$ with $q_0\trns{x}s^{{\vec{\imath}}}_1$.
Observe that the valid output sequences are $(\varepsilon^* a_1)^*(\varepsilon^* a_2)^* \dots(\varepsilon^* a_m)^*$.

Assume $\wa = \abra{Q,\Sigma,M,\{t\}}$ and $\mathcal{T}=\abra{Q',\Sigma\times (\Sigma'\cup\{\varepsilon\}), \to, q_0}$. Then construct the weighted automaton $\mathcal{T}(\wa)=\abra{Q\times Q',\Sigma',M^\mathcal{T},\{t\}\times Q'}$ using a product construction. The probability is associated in the following way $M^\mathcal{T}(a) ((s,q),(s',q'))= p$ if there is a transition $q\trns{b/a} q'$ in $\mathcal{T}$ and $s \trns[b]{p} s'$ in $\wa$. Note that, by this definition, there is a matrix $M^\mathcal{T}(\varepsilon)$; however, in every run of $\mathcal{T}(\wa)$ there can be a finite number of $\varepsilon$'s in a row, at most $r =\max_{1\le i\le m} |w_i| -1$.

Now let $\wa'$ be a copy of $\mathcal{T}(\wa)$ with $\varepsilon$ removed: $M'(a_i) = (\sum_{x=0}^r M^\mathcal{T}(\varepsilon)^x) M^\mathcal{T}(a^i)$. Then $\nu_\wa(w) = \nu_{\wa'}(a_1^{n_1}\dots a_m^{n_m})$ for all $n_1,\dots,n_m$ such that $w = w_1^{n_1}\dots w_m^{n_m}$. Hence, $\wa'$ is a weighted automaton with letter-bounded languages from $(s,q_0)$ and $(s',q_0)$ such that $(s,q_0)$ is big-O of $(s',q_0)$ in $\wa'$ if and only if $s$ is big-O of $s'$ in $\wa$.
\end{proof}

\section{Analysis by Ambiguity}
\label{sec:ambiguity}
When a problem is shown to be undecidable, it is often only undecidable for relatively complex instances. We have already considered the problem when the language of the weighted automaton is restricted to bounded languages. It is also interesting to consider whether decision  problem are decidable depending on the ambiguity of the instances~\cite{FijalkowR017,DAVIAUD202178}.

We say a weighted automaton $\wa$ is \textit{unambiguous from a state $s$} if every word has at most one accepting path in $\nfaof{\wa}{s}$. A weighted automaton $\wa$ is \textit{finitely ambiguous from a state $s$} if there exists a constant $k$ such that every word has at most $k$ accepting paths in $\nfaof{\wa}{s}$.
A weighted automaton may also be polynomially ambiguous, if there exists a polynomial $p(\cdot)$ such that every word $w$ has at most $p(|w|)$ accepting paths in $\nfaof{\wa}{s}$. If a weighted automaton is not polynomially ambiguous, then it is \emph{exponentially ambiguous}.

\cref{thm:approxboundundecidable} proves that the big-O problem is undecidable. However, the weighted automaton resulting from the reduction, for which the big-O problem cannot be decided, is exponentially ambiguous. This raises the question whether it is decidable for restricted ambiguities. We answer this positively, showing polynomial time decidability for  unambiguous weighted automata, and decidability for finite ambiguous weighted automata subject to Schanuel's conjecture (\cref{exa:relativeorderinsuff} is finitely ambiguous).

We leave open whether the big-O problem is decidable for polynomially ambiguous weighted automata.

\subsection{Unambiguous weighted automata}

In this section, we prove the polynomial-time solvability in the unambiguous case.

\begin{thm}\label{thm:detpolly}
If a weighted automaton $\wa$ is unambiguous from states $s$ and $s'$, the big-O problem is decidable in polynomial time.
\end{thm}

\begin{proof}
Let $\wa=\abra{Q, \Sigma{}, M, F}$ be a weighted automaton.
Suppose $s,s'\in Q$, $t$ is a unique final state, and $\wa$ is unambiguous from $s, s'$.

If $\wa$ fails the LC condition (recall that it can be checked in polynomial time), we return {no}. Otherwise,
let us construct a weighted automaton  $\wa'$  through a restricted product construction involving two copies of $\wa$:
for all $q_1, q_2, q_1', q_2'\in Q$, we add edges $(q_1,q_1') \trns[a]{p} (q_2,q_2')$ provided
$M(a)(q_1,q_2)>0$, $M(a)(q_1',q_2')>0$ and $p=\frac{M(a)(q_1,q_2)}{M(a)(q_1',q_2')}$.
Note that there exists a positively-weighted $w$-labelled path from $(s,s')$ to $(t,t)$ in $\wa'$
iff $w\in \lng{s}{\wa}\cap \lng{s'}{\wa}$. By the LC condition, this is equivalent to $w\in \lng{s}{\wa}$,
and, to examine the big-O problem, it suffices to consider only such words.

By unambiguity of $\wa$ from $s$ and $s'$, for any $w\in  \lng{s}{\wa}$, there can be exactly one positively-weighted path from $(s,s')$ to $(t,t)$ in $\wa'$.
Consequently, the product of weights along this path is equal to $\nu_s(w)/\nu_{s'}(w)$.
Hence,  $s$ is not big-O of $s'$ (for $\wa$) if and only there exists a positively-weighted path from $(s,s')$ to $(t,t)$ in $\wa'$
that contains a cycle such that the product of the weights in that cycle is greater than $1$.

Thus, to decide the big-O problem for $s,s'$, it suffices to be able to detect such cycles.
This can be done, for instance, by a modified version of
the Bellman-Ford algorithm~\cite{CLR90} applied to the weighted directed graph consisting of positively-weighted edges of $\wa'$.
The algorithm is normally used to find negative cycles in the sense that the sum of weights is negative.
To adapt it to our setting, we can apply the logarithm function to the weights.
However, to preserve rationality of weights and polynomial-time complexity,
we cannot afford to do that explicitly. Instead, whenever $\log(x)<\log(y)$ would be tested, we test $x<y$ and, whenever $\log(x)+\log(y)$ would be performed,
we compute $xy$ instead.
\end{proof}

Note the relevant behaviours are those on cycles---transitions which are taken at most once are of little significance to the big-O problem.
Such transitions have at most a constant multiplicative effect on the ratio. This is the case whether or not the system is unambiguous.

\subsection{Finitely ambiguous weighted automata}
A weighted automaton is finitely ambiguous if there exists a constant $k$ such that for every word $w$, there are at most $k$ accepting paths in  $\nfaof{\wa}{\mathcal{A}}$. Here show the following result for the Big-O problem on finitely ambiguous weighted automata, a problem left open in~\cite{ChistikovKMP20}:

\begin{thm}\label{thm:finambigdecstschanuels}
  Given two finitely ambiguous weighted automata, $\mathcal{A},\mathcal{B}$, it is decidable, subject to Schanuel's conjecture, whether $\mathcal{A}$ is big-O of $\mathcal{B}$.
\end{thm}

\cref{thm:finambigdecstschanuels} is proven by encoding the big-O problem into  the first order theory of the reals with exponential function in~\cref{lemma:finrealversion}. As a consequence, the Big-O problem is decidabile, conditional on Schanuel's conjecture (recall from~\cref{sec:logicaltheoriesschan} that Schanuel's conjecture entails decidability of ~\rexp).

Daviaud~et~al.~\cite{DAVIAUD202178} show, in Proposition 11 (and Proposition 15 therein), the following characterisation of attainable weights for pairs of finitely ambiguous weighted automata, which they use in the inclusion problem. In the following we adopt the notation $f_{\mathcal{A}}$ for the weight from the dedicated starting state of weighted automaton $\mathcal{A}$ and similarly for a weighted automaton $\mathcal{B}$.

\begin{thm}[\cite{DAVIAUD202178}]
Given two finitely ambiguous weighted automata, $\mathcal{A},\mathcal{B}$, one can compute a finite set of tuples $\Delta$ with elements $({p},{q}^1,\dots, {q}^{k}, {r},{s}^1,\dots, {s}^{\ell})$ where $p\in\mathbb{Q}_{>0}^{k},r\in\mathbb{Q}_{>0}^{\ell}$, $q^{i},s^{i} \in \mathbb{Q}_{>0}^{m}$ for some $m$, such that:
\begin{itemize}
\item for all $({p},{q}^1,\dots, {q}^{k}, {r},{s}^1,\dots, {s}^{\ell}) \in \Delta$ and ${n} \in \mathbb{N}^m$ there exists a word $w\in \Sigma^*$ such that
\begin{equation}
\label{eq:wa-pairs}
f_\mathcal{A}(w) = \sum_{i=1}^{k} p_i{({q}^i_1)}^{n_1}\dots, {({q}^i_m)}^{n_m}
\quad \text{ and } \quad f_\mathcal{B}(w) = \sum_{i=1}^{\ell} r_i{({s}^i_1)}^{n_1}\dots, {({s}^i_m)}^{n_m};
\end{equation}
\item for all words  $w\in \Sigma^*$ there exist  $({p},{q}^1,\dots, {q}^{k}, {r},{s}^1,\dots {s}^{\ell}) \in \Delta$ and ${n} \in \mathbb{N}^m$  such that Equation~\eqref{eq:wa-pairs} holds.
\end{itemize}

\end{thm}

Since $\Delta$ is a finite set, if there is a sequence of words $(w_i)_i$ such that $\frac{f_\mathcal{A}(w_i)}{f_\mathcal{B}(w_i)} \to\infty$, then there is an infinite subsequence such that every $w_i$ can be associated with a single choice of $({p},{q}^1,\dots {q}^{k}, {r},{s}^1,\dots {s}^{\ell}) \in \Delta$. As such, we can characterise the big-O condition using the following logical characterisation:
\begin{lem}
$\mathcal{A}$ is not big-O of $\mathcal{B}$ if and only if there exists $({p},{q}^1,\dots {q}^{k}, {r},{s}^1,\dots {s}^{\ell}) \in \Delta$, such that for all $C > 0$, there exists ${n} \in \mathbb{N}^m$ such that  \[ \sum_{i=1}^{k} p_i{({q}^i_1)}^{n_1}\dots {({q}^i_m)}^{n_m} \ge C\sum_{i=1}^{\ell} r_i{({s}^i_1)}^{n_1}\dots {({s}^i_m)}^{n_m}.\]
\end{lem}
Since we can iterate through $({p},{q}^1,\dots {q}^{k}, {r},{s}^1,\dots {s}^{\ell}) \in \Delta$, it remains to show that it is possible to check the logical characterisation. Like in \cref{claim:equalnaturalreal}, we show this logical characterisation using naturals is equivalent to a statement that  can be encoded in \rexp, the theory of the reals with exponential function.

\begin{lem}\label{lemma:finrealversion}
$\mathcal{A}$ is not big-O of $\mathcal{B}$ if and only if there exists $({p},{q}^1,\dots {q}^{k}, {r},{s}^1,\dots {s}^{\ell}) \in \Delta$, such that for all $C > 0$, there exists ${x} \in \mathbb{R}_{\ge 0}^m$ such that  \[ \sum_{i=1}^{k} p_i{({q}^i_1)}^{x_1}\dots {({q}^i_m)}^{x_m} \ge C \sum_{i=1}^{\ell} r_i{({s}^i_1)}^{x_1}\dots {({s}^i_m)}^{x_m}.\]
\end{lem}

\begin{proof}
Clearly if the real formulation is unsatisfied, then the formulation with naturals is unsatisfied. It remains to show that if the real formulation is satisfied, then so too is the formulation with naturals. We assume the condition in \cref{lemma:finrealversion} is satisfied for $({p},{q}^1,\dots {q}^{k}, {r},{s}^1,\dots {s}^{\ell}) \in \Delta$ fixed and we  show: \[\forall C > 0, \ \exists {n} \in \mathbb{N}^m \text{ s.t. } \sum_{i=1}^{k} p_i{({q}^i_1)}^{n_1}\dots {({q}^i_m)}^{n_m} \ge C\sum_{i=1}^{\ell} r_i{({s}^i_1)}^{n_1}\dots {({s}^i_m)}^{n_m}.\]

Let $C$ be given, we show the existence of a relevant vector ${n}$. Let us choose $C' = TC$, the exact value of $T\in\mathbb{R}_{+}$ will be chosen later, so that by assumption there exists ${x} \in \mathbb{R}_{\ge 0}^m$ such that  \begin{equation}\label{eq:mainin} \sum_{i=1}^{k} p_i{({q}^i_1)}^{x_1}\dots {({q}^i_m)}^{x_m} \ge C' \sum_{i=1}^{\ell} r_i{({s}^i_1)}^{x_1}\dots {({s}^i_m)}^{x_m}.\end{equation}

We decompose ${x}$ into its integer  and fractional parts. Let ${n},{y}$ be such that ${x} = {n} + {y}$, ${0}\le{y}< {1}$ and ${n} \in \mathbb{N}^m $.

Let $\displaystyle A = \max_{1\le i \le k}   \prod_{1\le j \le m} \max_{0\le z \le 1} {({q}^i_j)}^{z} $ and observe that
\begin{align} \notag \sum_{i=1}^{k} p_i{({q}^i_1)}^{x_1}\dots {({q}^i_m)}^{x_m} &= \sum_{i=1}^{k} p_i{({q}^i_1)}^{n_1}\dots {({q}^i_m)}^{n_m}{({q}^i_1)}^{y_1}\dots {({q}^i_m)}^{y_m}
\\
\label{eq:firstin}&\le A \sum_{i=1}^{k} p_i{({q}^i_1)}^{n_1}\dots {({q}^i_m)}^{n_m}.\end{align}

Similarly, let $\displaystyle B = \min_{1\le i \le \ell}   \prod_{1\le j \le m} \min_{0\le z \le 1} {({s}^i_j)}^{z} $, and observe that
\begin{align}
\sum_{i=1}^{\ell} r_i{({s}^i_1)}^{x_1}\dots {({s}^i_m)}^{x_m} &= \sum_{i=1}^{\ell} r_i{({s}^i_1)}^{n_1}\dots {({s}^i_m)}^{n_m}{({s}^i_1)}^{y_1}\dots {({s}^i_m)}^{y_m}\notag
\\
\label{eq:thirdin}&\ge B \sum_{i=1}^{\ell} r_i{({s}^i_1)}^{n_1}\dots {({s}^i_m)}^{n_m}.\end{align}

Finally, letting $T=\frac{A}{B}$ we can conclude that
\begin{align*}
 \sum_{i=1}^{k} p_i{({q}^i_1)}^{n_1}\dots {({q}^i_m)}^{n_m} &\ge
\frac{1}{A}\sum_{i=1}^{k} p_i{({q}^i_1)}^{x_1}\dots {({q}^i_m)}^{x_m}
&\text{by \cref{eq:firstin}}
\\ & \ge\frac{1}{A}  C'\sum_{i=1}^{\ell} r_i{({s}^i_1)}^{x_1}\dots {({s}^i_m)}^{x_m}
&\text{by \cref{eq:mainin}}
\\ & \ge  \frac{ C'}{A}B \sum_{i=1}^{\ell} r_i{({s}^i_1)}^{n_1}\dots {({s}^i_m)}^{n_m}&\text{by \cref{eq:thirdin}}
\\ & \ge C \sum_{i=1}^{\ell} r_i{({s}^i_1)}^{n_1}\dots {({s}^i_m)}^{n_m},&\text{as $C' = C\frac{A}{B}$}
\end{align*}
as required.
\end{proof}

\section{Conclusion}
Despite undecidability results, we have identified several decidable cases of the big-O problem.
However, for  bounded languages, the result depends on a conjecture from number theory, leaving open the exact borderline between decidability
and undecidability.

Natural directions for future work include the analogous problem for infinite words, further analysis on ambiguity (e.g., is the big-O problem decidable for polynomially-ambiguous weighted automata?), and the extension to negative edge weights.

\section*{Acknowledgement}

The authors would like to thank to Engel Lefaucheux, Filip Mazowiecki, Jo\"el Ouaknine, and James Worrell for discussions during the development of this work.

\bibliographystyle{alphaurl}
\bibliography{refs.bib}
\end{document}